\documentclass[11pt]{article}

\usepackage{amssymb,amsfonts,amsmath,amsthm}
\usepackage{longtable}
\usepackage{graphicx}
\usepackage{amsthm}
\usepackage{booktabs}

 \newtheorem{Theorem}{Theorem}[section]
  \newtheorem{Lemma}[Theorem]{Lemma}
  \newtheorem*{Remark}{Remark}

\usepackage{environ}
\NewEnviron{myequation}{%
\begin{equation}
\scalebox{0.9}{$\BODY$}
\end{equation}
}

\begin{document}

\title{Effects of Viral and Cytokine Delays on Dynamics of Autoimmunity}

\author{F. Fatehi Chenar,\hspace{0.5cm}Y.N. Kyrychko, \hspace{0.5cm}K.B. Blyuss\thanks{Corresponding author. Email: k.blyuss@sussex.ac.uk}\\\\ Department of Mathematics, University of Sussex, Falmer,\\
Brighton, BN1 9QH, United Kingdom}

\maketitle

\abstract{A major contribution to the onset and development of autoimmune disease is known to come from infections. An important practical problem is identifying the precise mechanism by which the breakdown of immune tolerance as a result of immune response to infection leads to autoimmunity. In this paper, we develop a mathematical model of immune response to a viral infection, which includes T cells with different activation thresholds, regulatory T cells (Tregs), and~a cytokine mediating immune dynamics. Particular emphasis is made on the role of time delays associated with the processes of infection and mounting the immune response. Stability analysis of various steady states of the model allows us to identify parameter regions associated with different types of immune behaviour, such as, normal clearance of infection, chronic infection, and autoimmune dynamics. Numerical simulations are used to illustrate different dynamical regimes, and to identify basins of attraction of different dynamical states. An important result of the analysis is that not only the parameters of the system, but also the initial level of infection and the initial state of the immune system determine the progress and outcome of the dynamics.}

\section{Introduction}

An immune system can only be viewed as effective when it can robustly identify and destroy pathogen-infected cells, while also distinguishing such cells from healthy cells. In the case of breakdown of immune tolerance, the immune system fails to discriminate between self-antigens and foreign antigens, which results in autoimmune disease, i.e., undesired destruction of healthy cells. Under normal conditions, once foreign epitopes are presented on antigen presenting cells (APCs) to T cells, this results in the proliferation of T cells and eliciting their effector function. While this mechanism is responsible for a successful clearance of infection, cross-reactivity between epitopes associated with foreign and self-antigens can lead to a T cell response against healthy host cells \cite{Mason98,And00}.

For many autoimmune diseases, the disease occurs in a specific organ or part of the body, such as retina in uveitis, central nervous system in multiple sclerosis, or pancreatic $\beta$-cells in type-1 diabetes \cite{Kerr08,Prat02,San10}. It is extremely difficult to identify the specific causes of autoimmunity in individual patients, as it usually has contributions from a number of internal and external factors, including a genetic predisposition, age, previous immune challenges, exposure to pathogens etc., \cite{RB15,Caf08,Li08,Guil11}. Even though genetic predisposition is known to play a very significant role, it is believed that some additional environmental triggers are required for the onset of autoimmunity, and these are usually represented by infections \cite{Ger12,Mal13}. A very recent work has experimentally identified a gut bacterium that, when present in mice and humans, can~migrate to other parts of the body, facilitating~subsequent triggering of autoimmune disease in those organs \cite{Vieira18}. Various mechanisms of onset of pathogen-induced autoimmune disease have been identified, including bystander activation~\cite{fuji3} and molecular mimicry~\cite{von,erco}, which is particularly important in the context of autoimmunity caused by viral infections.

A number of mathematical models have looked into dynamics of onset and development of autoimmune disease. Segel et al. \cite{sege} analysed 
interactions between effector and regulatory T cells in the context of T cell vaccination, without explicitly specifying possible causes of autoimmunity. Similar~models were later studied by Borghans et al.~\cite{borg1,borg2} who demonstrated possible onset of autoimmune state, defined as stable above-threshold oscillations in the number of autoreactive cells, as a result of interactions between regulatory and autoreactive T cells. Le\'on et al. \cite{leon1,leon2,Leon2004} and \mbox{Carneiro et al. \cite{carn}} have studied interactions between different T cells, with an emphasis on the suppressing role of regulatory T cells in the dynamics of immune response and control of autoimmunity. Alexander and Wahl \cite{alex} have also looked into the role of regulatory T cells, in particular focusing on their interactions with professional APCs and effector cells for the purpose of controlling immune response. \mbox{Iwami et al. \cite{iwam1,iwam2}} explicitly included a separate compartment representing the viral population in their models of immune response, and showed that the functional form of the growth function for susceptible host cells can have a significant effect on the resulting immune dynamics. Despite being able to explain the emergence of autoimmunity as a by-product of immune response to infection, these models were not able to exhibit another practically important dynamical regime of normal viral clearance. For the case of pathogen-induced autoimmunity arising through bystander activation, Burroughs et al. \cite{burr1,Burr11b,burr2} have developed a model that investigates interactions between T cells and interleukin-2 (IL-2), an important cytokine, in mediating the onset of autoimmunity.

Among various parts of the immune system involved in coordinating an effective immune response, a particularly significant role is known to be played by the T cells, with experimental evidence suggesting that regulatory T cells are vitally important for controlling autoimmunity \cite{Saka04,Jos12,Font03,Khat03}. To account for this fact in mathematical models, Alexander and Wahl \cite{alex} and \mbox{Burroughs et al. \cite{burr1,Burr11b}} have explicitly included a separate compartment for regulatory T cells that are activated by autoantigens and suppress the activity of autoreactive T cells. Another framework for modelling the effects of T cells on autoimmune dynamics is by using the idea that T cells have different or {\it tunable activation thresholds} (TAT), which result in different immune functionality of the same T cells, and also allow T cells to adjust their response to simulation by autoantigens. This approach was proposed for the analysis of the peripheral and central T cell dynamics \cite{gros,GPaul00,GSinger96}, it has also been used to study differences in activation/response thresholds that are dependent on the activation state of the \mbox{T cell \cite{Bonn05}}. Murine and human experiments have confirmed that activation of T cells can indeed change dynamically during their circulation \cite{Bit02,Nich00,Roe11,Stef02}. To model this feature, \mbox{van den Berg and Rand \cite{berg}}, \mbox{and Scherer et al. \cite{Scher04}} developed stochastic models for the tuning of activation thresholds.

Blyuss and Nicholson \cite{blyu12,blyu15} have proposed a mathematical model of autoimmunity resulting from immune response to a viral infection through a mechanism of molecular mimicry. This model explicitly includes the virus population and two types of T cells with different activation thresholds, and it also accounts for a biologically realistic scenario where infection and autoimmune response can occur in different organs of the host. Besides the normal viral clearance and chronic infection, in some parameter regime the model also exhibits an autoimmune state characterised by stable oscillations in the amounts of cell populations. From a clinical perspective, such behaviour is to be expected, as~it is associated with relapses and remissions that have been observed in a number of autoimmune diseases, such as autoimmune thyroid disease, MS, and uveitis \cite{Bezra95,Davies97,Nyla12}. One deficiency of this model is the fact that the oscillations associated with autoimmune regime can only occur if the amount of free virus and the number of infected cells are also exhibiting oscillations, while in clinical and laboratory settings, autoimmunity usually occurs after the initial infection has been fully cleared. To~overcome this limitation, Fatehi et al. \cite{Fatehi2018b} have recently developed a more advanced model that also includes regulatory T cells and cytokines, which has allowed the authors to obtain a more realistic representation of immune response and various dynamical regimes. A particularly important practical insight provided by this model is the observation that it is not only the system parameters, but also the initial level of infection and the initial state of the immune system, that determine whether the host will just successfully clear the infection, or will proceed to develop autoimmunity. Approaching the same problem from another perspective, Fatehi et al. \cite{Fatehi2018a} have investigated the role of stochasticity in driving the dynamics of immune response and determining which of the immune states is more likely to be attained. The authors have also determined an experimentally important characterisation of autoimmune state, as provided by the dependence of variance in cell populations on various system parameters.

In this paper, we develop and analyse a model of autoimmune dynamics, with particular focus on the role of time delays associated with different aspects of immune response, as well as an inhibiting effect of regulatory T cells on secretion of IL-2. In the next section, we introduce the model and discuss its basic properties. Section \ref{S3} contains systematic analysis of all steady states, including conditions for their feasibility and stability. Section \ref{S4} contains a bifurcation analysis of the model and demonstrates various types of behaviour that the system exhibits depending on parameters and initial conditions, which includes identification of attraction basins of various states. The paper concludes in Section \ref{S5} with the discussion of results.

\section{Model Derivation} \label{S2}

To understand how interactions between different parts of the immune system and the cytokine can lead to autoimmunity, we consider a model illustrated in a diagram shown in Figure~\ref{dia}. In this model, unlike earlier work of Blyuss and Nicholson \cite{blyu12,blyu15}, we consider a situation where a viral infection and possible autoimmunity occur in the same organ of the host. The healthy host cells, whose~number is denoted by $A(t)$, in the absence of infection are assumed to grow logistically with linear growth rate $r$ and carrying capacity $N$, and they acquire infection at rate $\beta$ from the infected cells $F(t)$. Since experimental evidence suggests that antibodies play a secondary role compared to T cells \cite{wu10}, and autoimmunity can develop even in the absence of B cells \cite{wolf96}, we do not include antibody response in the model, but focus solely on the dynamics of T cell populations. Na\"ive~(inactivated) T cells $T_{in}(t)$ are assumed to be in homeostasis \cite{blyu12,iwam1,iwam2}, and once they are activated through interaction with infected cells, which occurs at rate $\alpha$, a proportion $p_1$ of them will go on to differentiate into additional regulatory T cells, a fraction $p_2$ will become normal activated T cells $T_{nor}(t)$ able to destroy infected cells at rate $\mu_F$ upon recognition of foreign antigen present on their surface. The remaining proportion of $(1-p_1-p_2)$ of T cells will become autoreactive T cells $T_{aut}(t)$ that, in light of their lower activation threshold will be eliminating both infected cells and healthy host cells at rate $\mu_a$ due to the above-mentioned cross-reactivity between some epitopes in self- and foreign antigens. Regulatory T cells $T_{reg}(t)$ are assumed to be in their own homeostasis \cite{balt}, and their main contribution to immune dynamics lies in suppressing autoreactive T cells at rate $\delta_1$. To reduce the dimensionality of the model, it is assumed that the process of viral production is occurring very fast compared to other characteristic timescales of the model, thus the viral population can be represented by its quasi-steady-state approximation, i.e., it is taken to be proportional to the number of infected cells, and~this eliminates the need for a separate compartment for free virus.

A number of different cytokines mediate immune response to infection, and in the context of T cell dynamics, a particular important role is played by interleukin 2 (IL-2),  represented by the variable $I(t)$ in the model, which acts to enhance the proliferation of T cells, which, in turn, secrete further IL-2. One~of the actions of regulatory T cells is to suppress the expression of IL-2 \cite{Sheva01}, which~is only produced by the activated T cells, but not by the regulatory T cells \cite{abbas,Thorn98}. To represent this mathematically, we will assume that $T_{nor}$ and $T_{aut}$ produce IL-2 at rates $\sigma_1$ and $\sigma_2$, and conversely, IL-2 enhances proliferation of $T_{reg}$, $T_{nor}$ and $T_{aut}$ at rates $\rho_1$, $\rho_2$, and $\rho_3$. We include in the model suppression of IL-2 by regulatory T cells at rate $\delta_2$, in a manner similar to Burroughs et al. \cite{burr2}.
\begin{figure}
	\centering
    \includegraphics[width=0.9\linewidth]{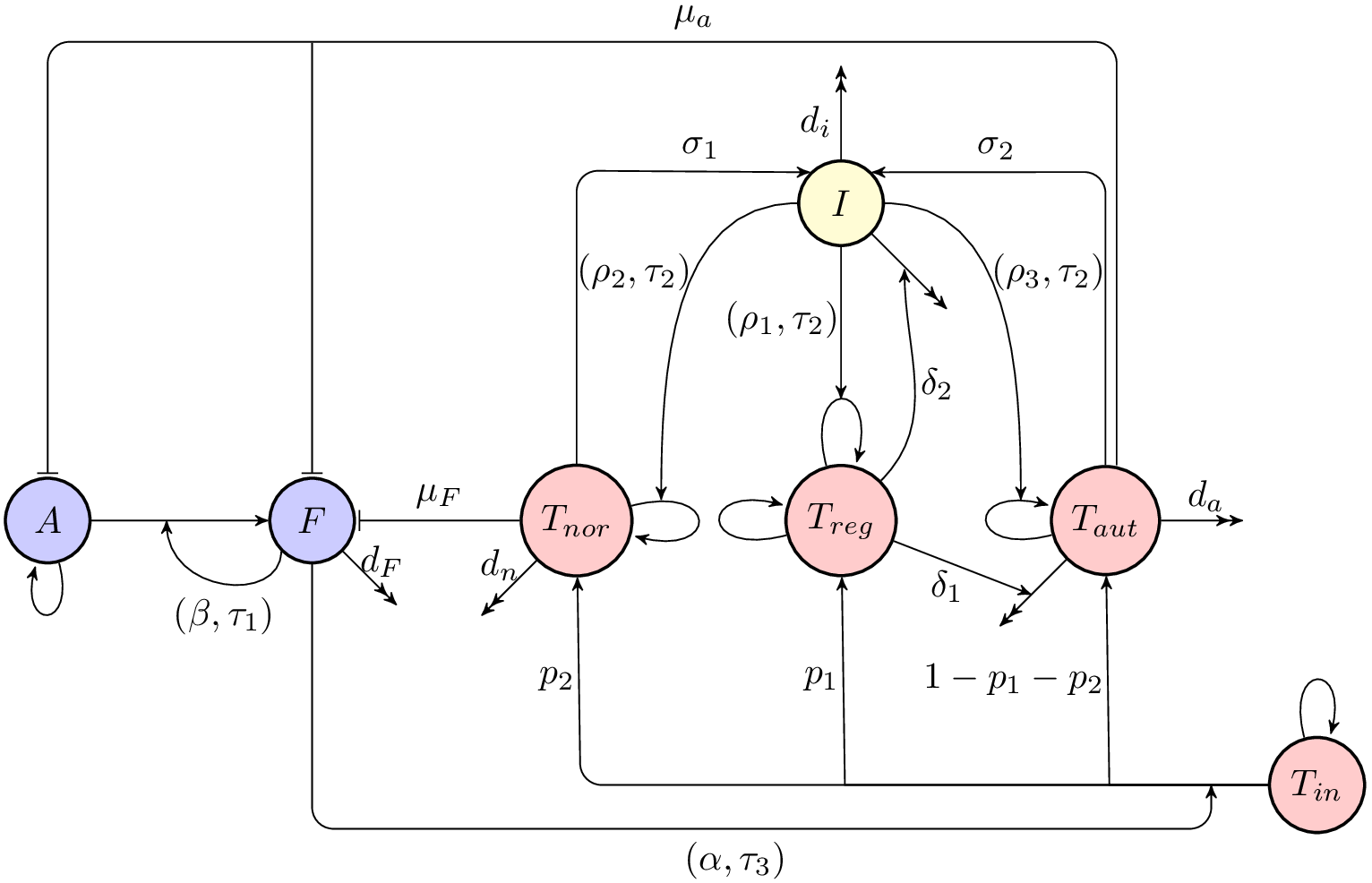}
	\caption{\small A schematic diagram of immune response to infection. Blue indicates host cells (susceptible and infected), red denotes different T cells (na\"ive, regulatory, normal activated, and autoreactive T~cells), yellow shows cytokines (interleukin 2). $\tau_i$ inside each of the subnetworks shows the time delay associated with that process.}
	\label{dia}
\end{figure}

While the production of new virus particles by infected cells is assumed to be fast, we explicitly include in the model time delay $\tau_1$ associated with the actual process of infection, which includes multiple stages of the {\it eclipse phase} of viral life cycle, such as virus attachment, cell penetration and uncoating \cite{bac06,pawel12}. We also include the time delay $\tau_2$ associated with simulation and proliferation of T~cells  by IL-2, and the time delay $\tau_3$ between antigen encounter and resulting T cell expansion \cite{Kim07}. 

With the above assumptions, the complete model takes the form
\[
\begin{array}{l}
\displaystyle{\dfrac{dA}{dt}=rA\left(1-\dfrac{A}{N}\right)-\beta AF-\mu_a T_{aut}A,} \vspace{4pt}\\
\displaystyle{\dfrac{dF}{dt}=\beta A\left(t-\tau_1\right)F\left(t-\tau_1\right)-d_FF-\mu_FT_{nor}F-\mu_a T_{aut}F,} \vspace{4pt}\\
\displaystyle{\dfrac{dT_{in}}{dt}=\lambda_{in}-d_{in}T_{in}-\alpha T_{in}F,} \vspace{4pt}\\
\displaystyle{\dfrac{dT_{reg}}{dt}=\lambda_r-d_rT_{reg}+p_1\alpha T_{in}\left(t-\tau_3\right)F\left(t-\tau_3\right)+\rho_1 I\left(t-\tau_2\right)T_{reg}\left(t-\tau_2\right),} \vspace{4pt}\\
\displaystyle{\dfrac{dT_{nor}}{dt}=p_2\alpha T_{in}\left(t-\tau_3\right)F\left(t-\tau_3\right)-d_nT_{nor}+\rho_2 I\left(t-\tau_2\right)T_{nor}\left(t-\tau_2\right),} \vspace{4pt}\\
\displaystyle{\dfrac{dT_{aut}}{dt}=(1-p_1-p_2)\alpha T_{in}\left(t-\tau_3\right)F\left(t-\tau_3\right)-d_aT_{aut}-\delta_1T_{reg}T_{aut}+\rho_3 I\left(t-\tau_2\right)T_{aut}\left(t-\tau_2\right),} \vspace{4pt}\\
\displaystyle{\dfrac{dI}{dt}=\sigma_1T_{nor}+\sigma_2T_{aut}-\delta_2T_{reg}I-d_iI.}
\end{array}
\]

Introducing non-dimensional variables
\begin{align*}
    &\hat{t}=rt,\quad A=N\hat{A},\quad F=N\hat{F},\quad T_{in}=\dfrac{\lambda_{in}}{d_{in}}\hat{T}_{in},\quad T_{reg}=\dfrac{\lambda_{in}}{d_{in}}\hat{T}_{reg},\\ &T_{nor}=\dfrac{\lambda_{in}}{d_{in}}\hat{T}_{nor},\quad T_{aut}=\dfrac{\lambda_{in}}{d_{in}}\hat{T}_{aut},\quad I=\dfrac{\lambda_{in}}{d_{in}}\hat{I},
\end{align*}
where
\begin{align*}
    &\hat{\beta}=\dfrac{\beta N}{r},\quad \hat{\mu}_a=\dfrac{\mu_a\lambda_{in}}{rd_{in}},\quad \hat{d}_F=\dfrac{d_F}{r},\quad \hat{\mu}_F=\dfrac{\mu_F\lambda_{in}}{rd_{in}},\quad \hat{d}_{in}=\dfrac{d_{in}}{r},\quad \hat{\alpha}=\dfrac{\alpha N}{r},\\
    &\hat{\lambda}_r=\dfrac{\lambda_rd_{in}}{\lambda_{in}r},\quad \hat{d}_r=\dfrac{d_r}{r}, \quad \hat{d}_n=\dfrac{d_n}{r}, \quad \hat{d}_a=\dfrac{d_a}{r}, \quad \hat{\rho}_i=\dfrac{\rho_i \lambda_{in}}{rd_{in}}, \quad i=1,2,3,\\
    &\hat{\delta}_1=\dfrac{\delta_1 \lambda_{in}}{r d_{in}},\quad \hat{\delta}_2=\dfrac{\delta_2 \lambda_{in}}{r d_{in}},\quad \hat{\sigma}_1=\dfrac{\sigma_1}{r},\quad \hat{\sigma}_2=\dfrac{\sigma_2}{r},\quad \hat{d}_i=\dfrac{d_i}{r},
\end{align*}
yields a rescaled model
\begin{myequation}\label{resc_model}
\begin{array}{l}
\displaystyle{\dfrac{dA}{dt}=A\left(1-A\right)-\beta AF- \mu_aT_{aut}A,} \vspace{4pt}\\
\dfrac{dF}{dT}=\beta A\left(T-\tau_1\right)F\left(T-\tau_1\right)- d_FF-\mu_FT_{nor}F-\mu_aT_{aut}F, \vspace{4pt}\\
\dfrac{dT_{in}}{dT}=d_{in}\left(1-T_{in}\right)-\alpha T_{in}F, \vspace{4pt}\\
\dfrac{dT_{reg}}{dT}=\lambda_r-d_rT_{reg}+p_1\alpha T_{in}\left(T-\tau_3\right)F\left(T-\tau_3\right)+\rho_1I\left(T-\tau_2\right)T_{reg}\left(T-\tau_2\right), \vspace{4pt}\\
\dfrac{dT_{nor}}{dT}=p_2\alpha T_{in}\left(T-\tau_3\right)F\left(T-\tau_3\right)-d_nT_{nor}+ \rho_2 I\left(T-\tau_2\right)T_{nor}\left(T-\tau_2\right), \vspace{4pt}\\
\dfrac{dT_{aut}}{dT}=(1-p_1-p_2)\alpha T_{in}\left(T-\tau_3\right)F\left(T-\tau_3\right)-d_aT_{aut}-\delta_1T_{reg}T_{aut}+\rho_3I\left(T-\tau_2\right)T_{aut}\left(T-\tau_2\right), \vspace{4pt}\\
\dfrac{dI}{dT}=\sigma_1T_{nor}+\sigma_2T_{aut}-\delta_2T_{reg}I-d_iI,
\end{array}
\end{myequation}
where all hats in variables and parameters have been dropped for simplicity of notation, and all parameters are assumed to be positive. It is easy to show that this system is well-posed, i.e., solutions with non-negative initial conditions remain non-negative for all $t\geq 0$.

As a first step in the analysis of model~(\ref{resc_model}), we look at its steady states
\[
S^{\ast}=\left(A^{\ast}, F^{\ast}, T^{\ast}_{in}, T^{\ast}_{reg}, T^{\ast}_{nor}, T^{\ast}_{aut}, I^{\ast}\right),
\]
that can be found by equating to zero the right-hand sides of Equation~(\ref{resc_model}) and solving the resulting system of algebraic equations, deferring the discussion of conditionally stable steady states to Section~\ref{S3}. First, we consider a situation where there are no infected cells at a steady state, i.e., $F^{\ast}=0$, which~immediately implies $T^{\ast}_{in}=1$. In this case, there are four possible combinations of steady states depending on whether $T^{\ast}_{nor}$ and $T^{\ast}_{aut}$ are each equal to zero or positive. If $T^{\ast}_{nor}=T^{\ast}_{aut}=0$, there~are two steady states
\[
S^{\ast}_1=\left(0, 0, 1, \dfrac{\lambda_r}{d_r}, 0, 0, 0\right),\quad
S^{\ast}_2=\left(1, 0, 1, \dfrac{\lambda_r}{d_r}, 0, 0, 0\right),
\]
of which $S^{\ast}_1$ is always unstable, and $S^{\ast}_2$ is a disease-free conditionally stable steady state.

For $T^{\ast}_{nor}\neq 0$ and $T^{\ast}_{aut}=0$, we again have two steady states \vspace{-4pt}
\[
S^{\ast}_3=\left(0, 0, 1, \dfrac{\lambda_r\rho_2}{\rho_2d_r-\rho_1d_n}, T_{nor}^{\ast}, 0, \dfrac{d_n}{\rho_2}\right),\quad
S^{\ast}_4=\left(1, 0, 1, \dfrac{\lambda_r\rho_2}{\rho_2d_r-\rho_1d_n}, T_{nor}^{\ast}, 0, \dfrac{d_n}{\rho_2}\right),
\] \vspace{-4pt}
where $T_{nor}^{\ast}=\dfrac{d_n\left(\lambda_r\delta_2\rho_2+d_id_r\rho_2-d_id_n\rho_1\right)}{\rho_2\sigma_1(\rho_2d_r-\rho_1d_n)}$, but they are both unstable for any values of parameters. In the case when $T^{\ast}_{nor}=0$ and $T^{\ast}_{aut}\neq 0$, we have two further steady states $S^{\ast}_5$ and $S^{\ast}_6$,
\[
\begin{array}{l}
\displaystyle{S^{\ast}_{5}=\left(0, 0, 1,T^{\ast}_{reg},0,\dfrac{\left(d_i+\delta_2T^{\ast}_{reg}\right)\left(d_a+\delta_1 T^{\ast}_{reg}\right)}{\rho_3\sigma_2},\dfrac{d_a+\delta_1T^{\ast}_{reg}}{\rho_3}\right),}\\\\
\displaystyle{S^{\ast}_{6}=\left(A^{\ast}, 0, 1,T^{\ast}_{reg},0,\dfrac{\left(d_i+\delta_2T^{\ast}_{reg}\right)\left(d_a+\delta_1 T^{\ast}_{reg}\right)}{\rho_3\sigma_2},\dfrac{d_a+\delta_1T^{\ast}_{reg}}{\rho_3}\right),}
\end{array}
\]
where $A^{\ast}=1-\dfrac{\mu_a\left(d_i+\delta_2T^{\ast}_{reg}\right)\left(d_a+\delta_1 T^{\ast}_{reg}\right)}{\rho_3\sigma_2}$, and
\[
\displaystyle{T^{\ast}_{reg}=\frac{d_r\rho_3-\rho_1d_a\pm\sqrt{\left(d_r\rho_3-\rho_1d_a\right)^2-4\rho_1\delta_1\lambda_r\rho_3}}{2\rho_1\delta}.}
\]

The steady state $S^{\ast}_{5}$ has $A^{\ast}=0$, which implies the death of host cells, whereas the steady state $S^{\ast}_{6}$ corresponds to an autoimmune regime. The steady state $S^{\ast}_{7}$ with $T^{\ast}_{nor}\neq 0$ and $T^{\ast}_{aut}\neq 0$ exists only for a particular combination of parameters, namely, when
\[
\delta_1\rho_2^2\lambda_r=(\rho_3d_n-\rho_2d_a)(\rho_2d_r-\rho_1d_n),
\]
and is always unstable. Finally, when $F^{\ast}\neq 0$, the system (\ref{resc_model}) can have a steady state $S^{\ast}_8$ with all of its components being positive, but it does not appear possible to find a closed form expression for this state.

In summary, besides the unconditionally unstable steady states, the model (\ref{resc_model}) has at most for conditionally stable steady states: the disease-free steady state $S^{\ast}_2$, the steady state with the death of host cells $S^{\ast}_5$, the autoimmune steady state $S^{\ast}_6$, and the persistent or chronic steady state $S^{\ast}_8$.

\section{Stability Analysis of the Steady States} \label{S3}

\subsection{Stability Analysis of the Disease-Free Steady State}

Linearising the system (\ref{resc_model}) near the disease-free steady state $S^{\ast}_2$ yields the following equation for characteristic roots $\lambda$
\begin{equation}\label{DF equation}
\lambda+d_F-\beta e^{-\lambda\tau_1}=0.
\end{equation}

If $d_F<\beta$, the above equation always has a real positive root for any value $\tau_1\geq 0$, implying~that the disease-free steady state is always unstable for any value of the time delays. If, however, the condition $d_F>\beta$ holds, the disease-free steady state is stable for $\tau_1=0$. To find out whether it can lose stability for $\tau_1>0$, we look for solutions of Equation (\ref{DF equation}) in the form $\lambda=i\omega$. Separating real and imaginary parts yields
\[
\begin{array}{l}
d_F=\beta\cos(\omega\tau_1),\\\\
\omega=-\beta\sin(\omega\tau_1).
\end{array}
\]

Squaring and adding these two equations gives the following equation for potential Hopf frequency $\omega$
\[
\omega^2+d_F^2-\beta^2=0.
\]
since $d_F>\beta$, this equation does not have real roots for $\omega$, suggesting that there can be no roots of the form $\lambda=i\omega$ of the characteristic Equation (\ref{DF equation}). This implies that in the case $d_F>\beta$ the disease-free steady state $S^{\ast}_2$ is stable for all values of the time delay $\tau_1\geq 0$.

\subsection{Stability Analysis of the Death, Autoimmune and Chronic Steady States}

The steady state $S^{\ast}_5$ (respectively, $S^{\ast}_6$) is stable if
\begin{equation}\label{S5 stability1}
	P<\dfrac{d_a+\delta_1T^{\ast}_{reg}}{\rho_3}<\dfrac{d_n}{\rho_2},
\end{equation}
and all roots of the following equation have negative real part
\begin{equation}\label{S2 equation}
	\Delta(\tau_2,\lambda)=p_2(\lambda)e^{-2\lambda\tau_2}+p_1(\lambda)e^{-\lambda\tau_2}+p_0(\lambda)=0,
\end{equation}
where
\begin{align*}
	p_2(\lambda)=&\dfrac{\rho_1\left(d_a+\delta_1T^{\ast}_{reg}\right)^2}{\rho_3}\left(\lambda+2d_i+\delta_2T^{\ast}_{reg}\right),\\
	p_1(\lambda)=&-\dfrac{\left(d_a+\delta_1T^{\ast}_{reg}\right)}{\rho_3}\Big\{(\rho_1+\rho_3)\lambda^2+\left[\rho_1\left(d_i+d_a+\delta_1 T^{\ast}_{reg}\right)+\rho_3\left(d_r+2d_i+2\delta_2T^{\ast}_{reg}\right)\right]\lambda\\
	&\hspace{3cm}+d_i(\rho_1d_a+2d_r\rho_3)+\delta_2T^{\ast}_{reg}\left(-\rho_1\delta_1T^{\ast}_{reg}+2d_r\rho_3\right)\Big\},\\
	p_0(\lambda)=&\left(\lambda+d_r\right)\left(\lambda+d_i+\delta_2T^{\ast}_{reg}\right)\left(\lambda+d_a+\delta_1T^{\ast}_{reg}\right),
\end{align*}
and
\[
P=\begin{cases}
\dfrac{\sigma_2}{\mu_a\left(d_i+\delta_2T^{\ast}_{reg}\right)}, &\mbox{ for } S^{\ast}_5,\\\\
\dfrac{\sigma_2\left(\beta-d_F\right)}{\mu_a\left(1+\beta\right)\left(d_i+\delta_2T^{\ast}_{reg}\right)}, &\mbox{ for } S^{\ast}_6.
\end{cases}
\]

This steady state undergoes a steady-state bifurcation if
\begin{equation}\label{SS_bif}
	\dfrac{d_a+\delta_1T^{\ast}_{reg}}{\rho_3}=P,\quad \mbox{or}\quad \dfrac{d_a+\delta_1T^{\ast}_{reg}}{\rho_3}=\dfrac{d_n}{\rho_2},\quad \mbox{or}\quad \delta_1\rho_1\left(T^{\ast}_{reg}\right)^2=\lambda_r\rho_3.
\end{equation}

For $\tau_2=0$ these steady states are stable if $T^{\ast}_{reg}$ satisfies (\ref{S5 stability1}) and
\begin{equation}\label{S5 stability2}
\begin{aligned}
&\delta_1\rho_1\left(T^{\ast}_{reg}\right)^2>\lambda_r\rho_3,\\
&a_5\left(T^{\ast}_{reg}\right)^5+a_4\left(T^{\ast}_{reg}\right)^4+a_3\left(T^{\ast}_{reg}\right)^3+a_2\left(T^{\ast}_{reg}\right)^2+a_1T^{\ast}_{reg}+a_0>0,
\end{aligned}
\end{equation}
where
\begin{align*}
	&a_5=-\delta_1\delta_2(\delta_1\rho_1-\delta_2\rho_1+\delta_2\rho_3),\; a_4=d_a\delta_2(\delta_2\rho_2-\delta_1\rho_1-\delta_2\rho_3)-d_i\delta_1(\delta_1\rho_1-\delta_2\rho_1+2\delta_2\rho_3),\\
	&a_3=-d_i\delta_1(d_a\rho_1+d_i\rho_3)+d_ad_i\delta_2(\rho_1-2\rho_3)+\lambda_r\delta_2(\delta_1\rho_1+\delta_2\rho_3),\\
	&a_2=-d_ad_i^2\rho_3+\lambda_r\delta_2(d_a\rho_1+2d_i\rho_3),\;a_1=\lambda_r\rho_3(d_i^2+\delta_2\lambda_r),\;a_0=d_i\rho_3\lambda_r^2.
\end{align*}

To investigate whether stability can be lost for $\tau_2>0$, we use an iterative procedure described in \cite{li2011,Rahman2015} to determine a function $F(\omega)$, whose roots give the Hopf frequency associated with purely imaginary roots of Equation (\ref{S2 equation}). Substituting $\lambda=i\omega$ into Equation (\ref{S2 equation}), we define $\Delta^{(1)}(\tau_2,\lambda)$ as
$$\Delta^{(1)}(\tau_2,\lambda)=\overline{p_0(i\omega)}\Delta(\tau_2,i\omega)-p_2(i\omega)e^{-2i\omega\tau_2}\overline{\Delta(\tau_2,i\omega)}=p_0^{(1)}(i\omega)+p_1^{(1)}(i\omega)e^{-i\omega\tau_2},$$
where
\begin{align*}
	p_0^{(1)}(i\omega)=&|p_0(i\omega)|^2-|p_2(i\omega)|^2,\\
	p_1^{(1)}(i\omega)=&\overline{p_0(i\omega)}p_1(i\omega)-\overline{p_1(i\omega)}p_2(i\omega),
\end{align*}
and the bar denotes the complex conjugate. If we define
$$F(\omega)=\left|p_0^{(1)}(i\omega)\right|^2-\left|p_1^{(1)}(i\omega)\right|^2,$$
then $\Delta(\tau_2,i\omega)=0$ whenever $\omega$ is a root of $F(\omega)=0$. The function $F(\omega)$ has the explicit form
\[
F(\omega)=\omega^{12}+b_{10}\omega^{10}+b_8\omega^{8}+b_6\omega^6+b_4\omega^4+b_2\omega^2+b_0,
\]
with
\begin{align*}
b_0=\dfrac{\left(\delta_{1}{d_{{a}}+T^{\ast}_{reg}}\right)^{4}}{{{\rho_{3}}^{4}}}&\left( d_{{i}}+\delta_{2}{T^{\ast}_{reg}}\right)\left(2{T^{\ast}_{reg}}\delta_{1}\rho_{1}+d_{{a}}\rho_{1}-d_{r}\rho_{3} \right)\\
&\Big[\left(d_i+\delta_2T^{\ast}_{reg}\right)\left(d_a\rho_1+3d_r\rho_3\right)+2d_i\rho_1\left(d_a+\delta_1T^{\ast}_{reg}\right)\big]\\
&\Big[\rho_1\left(d_a+\delta_1T^{\ast}_{reg}\right)\left(2d_i+\delta_2T^{\ast}_{reg}\right)-d_r\rho_3\left(d_i+\delta_2T^{\ast}_{reg}\right)\Big] ^{2}.
\end{align*}

The explicit formulae for other coefficients of $F(\omega)$ can be found in Appendix \ref{A1}. Introducing $s=\omega^2$, the~equation $F(\omega)=0$ can be equivalently rewritten as follows,
\begin{equation}\label{S2 equation hopf}
h(s)=s^6+b_{10}s^5+b_8s^4+b_6s^3+b_4s^2+b_2s+b_0=0.
\end{equation}

Without loss of generality, suppose that Equation (\ref{S2 equation hopf}) has six distinct positive roots denoted by $s_1$, $s_2$, ... , $s_6$, which means that the equation $F(\omega)=0$ has six positive roots
\[
\omega_i=\sqrt{s_i}, \quad i=1,2,...,6.
\]

Substituting $\lambda_k=i\omega_k$ into Equation (\ref{S2 equation}) gives
\begin{equation*}
	\tau_{k,j}=\dfrac{1}{\omega_k}\left[\arctan\left(\dfrac{\omega_k\left((\rho_1+\rho_3)\omega_k^4+f_2\omega_k^2+f_0\right)}{\left(\rho_3Z-d_r\rho_1-\rho_3^2I^{\ast}-\rho_1\delta_2T^{\ast}_{reg}\right)\omega_k^4+g_2\omega_k^2+g_0}\right)+j\pi\right],
\end{equation*}
for $k=1,2,...,6$, $j=0,1,2,...$, where
\begin{align*}
	f_0=&-{\rho_{1}}^{2}{\rho_{3}}^{2}{I^{\ast}}^{3}Z-\rho_1\rho_3\left(2{\rho_{1}}+3{\rho_{3}}\right){I^{\ast}}^{2}{Z}^{2}+\rho_1\rho_3T^{\ast}_{reg}\left(-\delta_{1}{\rho_{1}}+3\delta_{2}{\rho_{1}}+\delta_{2}{\rho_{3}}\right){I^{\ast}}^{2}Z\\
	&-{{T^{\ast}_{reg}}}^{2}{\delta_{2}}^{2}{\rho_{1}}^{2}\rho_{3}{I^{\ast}}^{2}-{T^{\ast}_{reg}}\delta_{1}\rho_{1}\rho_{3}I^{\ast}{Z}^{2}+d_r\rho_3\left(-\delta_{1}\rho_{1}T^{\ast}_{reg}+{d_{r}}{\rho_{3}}\right)I^{\ast}Z\\
	&+d_r\left(-{T^{\ast}_{reg}}\delta_{1}\rho_{1}+2{d_{r}}\rho_{3}\right){Z}^{2},\\
	f_2=&-{\rho_{1}}^{2}\rho_{3}{I^{\ast}}^{2}+{\rho_{3}}^{2}I^{\ast}Z+ \left( \rho_{1}+2\rho_{3} \right) {Z}^{2}+\rho_1T^{\ast}_{reg}\left(\delta_{1}-\delta_{2}\right) Z-d_r\left({T^{\ast}_{reg}}\delta_{2}\rho_{1}-{d_{r}}\rho_{3}\right),\\
	g_0=~&{\rho_{1}}^{2}{\rho_{3}}^{2}{I^{\ast}}^{3}\left(2{Z}^{2}-3{T^{\ast}_{reg}}\delta_{2}Z+{{T^{\ast}_{reg}}}^{2}{\delta_{2}}^{2}\right)+\rho_1\rho_3\left(-2{T^{\ast}_{reg}}\delta_{1}{\rho_{1}}+3d_{r}{\rho_{3}}\right){I^{\ast}}^{2}{Z}^{2}\\
	&+\rho_1\rho_3\delta_2T^{\ast}_{reg}\left({{T^{\ast}_{reg}}}\delta_{1}{\rho_{1}}-d_{r}{\rho_{3}}\right){I^{\ast}}^{2}Z+d_r\rho_3\left( {T^{\ast}_{reg}}\delta_{1}\rho_{1}-2{d_{r}}{\rho_{3}}\right)I^{\ast}{Z}^{2},\\
	g_2=~&\rho_3I^{\ast}\left({\rho_{1}}^{2}{\rho_{3}}{I^{\ast}}^{2}-{\rho_{1}}^{2}{I^{\ast}}Z-2{\rho_{3}}{Z}^{2}-{T^{\ast}_{reg}}\delta_{1}\rho_{1}Z-{d_{r}}^{2}{\rho_{3}}\right)-\rho_1\left(d_{r}+\delta_{1}{T^{\ast}_{reg}}\right){Z}^{2}\\
	&+d_r\left(-{T^{\ast}_{reg}}\delta_{1}\rho_{1}+{T^{\ast}_{reg}}\delta_{2}\rho_{1}+{d_{r}}\rho_{3}\right)Z,
\end{align*}
and
$$I^{\ast}=\dfrac{d_a+\delta_1T^{\ast}_{reg}}{\rho_3},\quad Z=d_i+\delta_2T^{\ast}_{reg}.$$

This allows us to find
\[
\tau^{\ast}=\tau_{k_0,0}=\min\limits_{1\leq k\leq 6}\{\tau_{k,0}\}, \quad \omega_0=\omega_{k_0},
\]
as the first time delay for which the roots of the characteristic Equation (\ref{S2 equation}) cross the imaginary axis. To determine whether these steady states actually undergo a Hopf bifurcation at $\tau_2=\tau^{\ast}$, we have to compute the sign of $d\mbox{Re}[\lambda(\tau^{\ast})]/d\tau_2$. For $\tau=\tau^{\ast}$, $\lambda(\tau^{\ast})=i\omega_0$, and we also define $s_0=\omega_0^2$.

\begin{Lemma}\label{lemma}
	Suppose $h^{\prime}(s_0)\neq 0$ and $p_0^{(1)}(i\omega_0)\neq 0$. Then the following transversality condition holds \vspace{-4pt}
	$$\emph{sgn}\left\{\left.\dfrac{d\,\emph{Re}(\lambda)}{d\,\tau_2}\right\vert_{\tau_2=\tau^{\ast}}\right\}=\emph{sgn}[p_0^{(1)}(i\omega_0)h^{\prime}(s_0)].$$
\end{Lemma}

\begin{proof}
Considering $p_j(i\omega_0)=x_j(\omega_0)+iy_j(\omega_0)$ for $j=0,1,2$, we have \vspace{-4pt}
\begin{align*}
p_0^{(1)}(i\omega_0)=&x_0^2+y_0^2-x_2^2-y_2^2,\\
p_1^{(1)}(i\omega_0)=&(x_0x_1+y_0y_1-x_1x_2-y_1y_2)+(x_0y_1+x_2y_1-x_1y_0-x_1y_2)i,
\end{align*}
where all $x_j$ and $y_j$ are expressed in terms of system parameters and steady state values of the variables. Substituting these expressions into $\Delta(\tau_2,i\omega_0)=0$ and $\Delta^{(1)}(\tau_2,i\omega_0)=0$, and then separating real and imaginary parts gives
\begin{equation*}
\begin{cases}
x_2\cos(2\omega_0\tau^{\ast})+y_2\sin(2\omega_0\tau^{\ast})+x_1\cos(\omega_0\tau^{\ast})+y_1\sin(\omega_0\tau^{\ast})=-x_0,\\
y_2\cos(2\omega_0\tau^{\ast})-x_2\sin(2\omega_0\tau^{\ast})+y_1\cos(\omega_0\tau^{\ast})-x_1\sin(\omega_0\tau^{\ast})=-y_0,\\
(x_0x_1+y_0y_1-x_1x_2-y_1y_2)\cos(\omega_0\tau^{\ast})+(x_0y_1+x_2y_1-x_1y_0-x_1y_2)\sin(\omega_0\tau^{\ast})\\
\hspace{2cm}=-x_0^2-y_0^2+x_2^2+y_2^2,\\
(x_0y_1+x_2y_1-x_1y_0-x_1y_2)\cos(\omega_0\tau^{\ast})-(x_0x_1+y_0y_1-x_1x_2-y_1y_2)\sin(\omega_0\tau^{\ast})=0.
\end{cases}
\end{equation*}

Solving this system of equations provides the values of $\sin(\omega_0\tau^{\ast})$, $\cos(\omega_0\tau^{\ast})$, $\sin(2\omega_0\tau^{\ast})$, and~$\cos(2\omega_0\tau^{\ast})$. Taking the derivative of Equation (\ref{S2 equation}) with respect to $\tau_2$, one finds
$$\left(\dfrac{d\,\lambda}{d\,\tau_2}\right)^{-1}=\dfrac{p_2^{\prime}(\lambda)e^{-2\lambda\tau_2}+p_1^{\prime}(\lambda)e^{-\lambda\tau_2}+p_0^{\prime}(\lambda)}{\lambda\left(2p_2(\lambda)e^{-2\lambda\tau_2}+p_1(\lambda)e^{-\lambda\tau_2}\right)}-\dfrac{\tau_2}{\lambda}.$$

Hence,
\begin{align*}
&\left(\left.\dfrac{d\,\mbox{Re}(\lambda)}{d\,\tau_2}\right\vert_{\tau_2=\tau^{\ast}}\right)^{-1}=\mbox{Re}\left\{\dfrac{p_2^{\prime}(\lambda)e^{-2\lambda\tau_2}+p_1^{\prime}(\lambda)e^{-\lambda\tau_2}+p_0^{\prime}(\lambda)}{\lambda\left(2p_2(\lambda)e^{-2\lambda\tau_2}+p_1(\lambda)e^{-\lambda\tau_2}\right)}\right\}_{\tau_2=\tau^{\ast}}-\mbox{Re}\left\{\dfrac{\tau_2}{\lambda}\right\}_{\tau_2=\tau^{\ast}}\\
&=\mbox{Re}\left\{\dfrac{p_2^{\prime}(i\omega_0)e^{-2i\omega_0\tau_2}+p_1^{\prime}(i\omega_0)e^{-i\omega_0\tau_2}+p_0^{\prime}(i\omega_0)}{i\omega_0\left(2p_2(i\omega_0)e^{-2i\omega_0\tau_2}+p_1(i\omega_0)e^{-i\omega_0\tau_2}\right)}\right\}\\
&=\dfrac{1}{\omega_0}\,\mbox{Im}\left\{\dfrac{p_2^{\prime}(i\omega_0)e^{-2i\omega_0\tau_2}+p_1^{\prime}(i\omega_0)e^{-i\omega_0\tau_2}+p_0^{\prime}(i\omega_0)}{2p_2(i\omega_0)e^{-2i\omega_0\tau_2}+p_1(i\omega_0)e^{-i\omega_0\tau_2}}\right\}\\
&=\dfrac{1}{\Lambda\omega_0}\Big[-x_2x_2^{\prime}-y_2y_2^{\prime}+x_0x_0^{\prime}+y_0y_0^{\prime}+(x_2y_1^{\prime}-y_2x_1^{\prime}+x_0y_1^{\prime}-x_1^{\prime}y_0)\sin(\omega_0\tau^{\ast})\\
&\hspace{1.55cm}+(x_0x_1^{\prime}+y_0y_1^{\prime}-x_1^{\prime}x_2-y_1^{\prime}y_2)\cos(\omega_0\tau^{\ast})+(x_2y_0^{\prime}-x_0^{\prime}y_2+x_0y_2^{\prime}-x_2^{\prime}y_0)\sin(2\omega_0\tau^{\ast})\\
&\hspace{1.55cm}+(x_0x_2^{\prime}+y_0y_2^{\prime}-x_0^{\prime}x_2-y_0^{\prime}y_2)\cos(2\omega_0\tau^{\ast})\Big],
\end{align*}
where
\[
\Lambda=\left|2p_2(i\omega_0)e^{-2i\omega_0\tau_2}+p_1(i\omega_0)e^{-i\omega_0\tau_2}\right|^2.
\]

\noindent Substituting the values of $\sin(\omega_0\tau^{\ast})$, $\cos(\omega_0\tau^{\ast})$, $\sin(2\omega_0\tau^{\ast})$, and $\cos(2\omega_0\tau^{\ast})$ found earlier gives
\begin{align*}
\left(\left.\dfrac{d\,\mbox{Re}(\lambda)}{d\,\tau_2}\right\vert_{\tau_2=\tau^{\ast}}\right)^{-1}=\dfrac{1}{\Lambda\omega_0}\,\dfrac{F^{\prime}(\omega_0)}{2\,p_0^{(1)}(i\omega_0)}=\dfrac{h^{\prime}(s_0)}{\Lambda\,p_0^{(1)}(i\omega_0)}.
\end{align*}

\begin{figure}
	\centering
	\includegraphics[width=0.8 \linewidth]{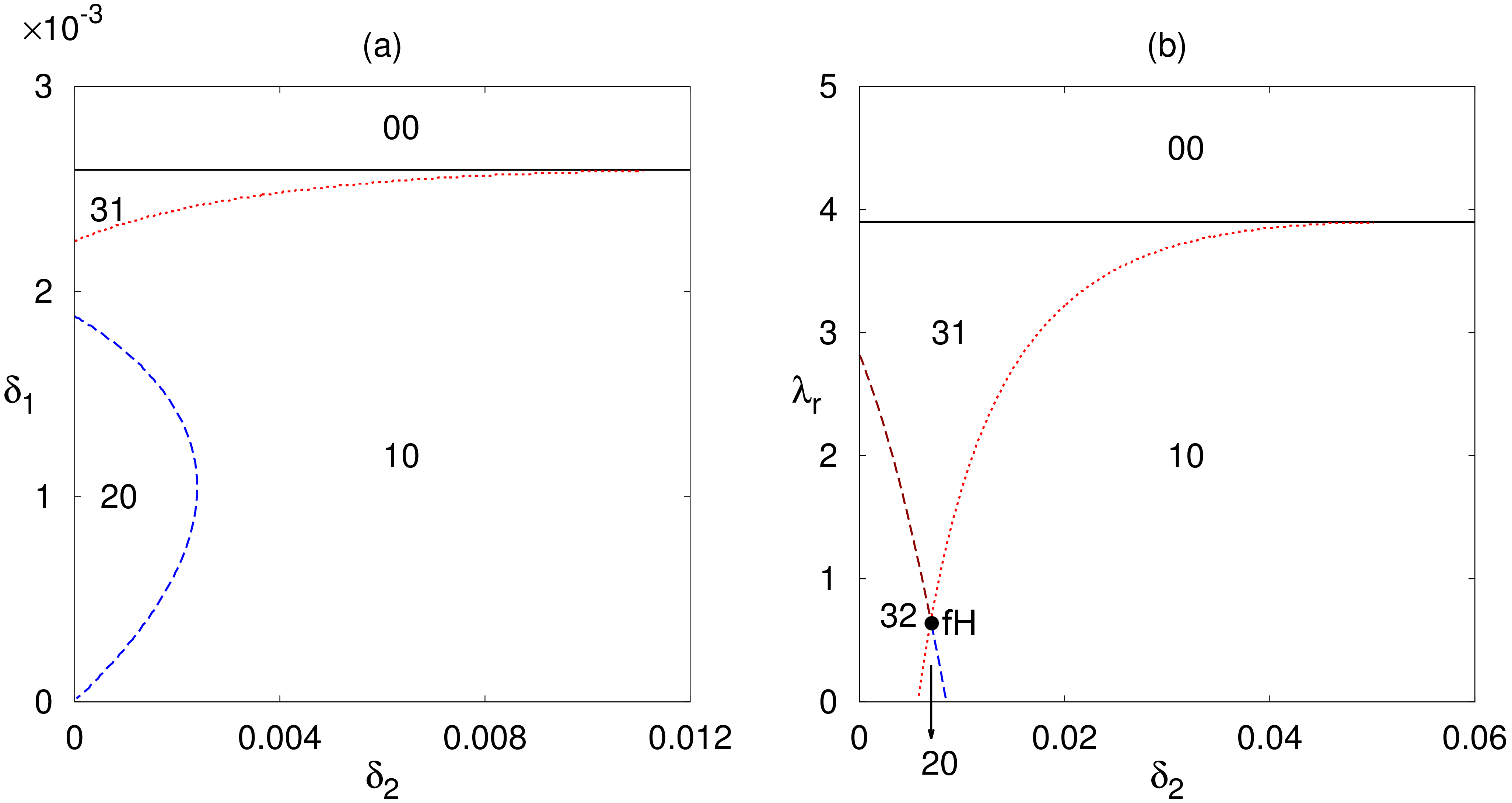}
	\caption{\small Regions of feasibility and stability of the steady states $S^{\ast}_5$ and $S^{\ast}_6$ with parameter values from Table \ref{parameter table} (\textbf{a}), and with $\mu_a=10$ (\textbf{b}). Black and red curves indicate the boundaries of feasibility and the steady-state bifurcation, whereas dashed lines (blue/brown) show the boundaries of Hopf bifurcation of the steady states $S^{\ast}_5$ and $S^{\ast}_6$, respectively, with `fH' indicating the fold-Hopf bifurcation. The first digit of the index refers to $S^{\ast}_5$, while the second corresponds to $S^{\ast}_6$, and they indicate that in that parameter region the respective steady state is unfeasible (index is `0'), stable (index is `1'), unstable via Hopf bifurcation with a periodic solution around this steady state (index is `2'), or unstable via a steady-state bifurcation (index is `3'). In all plots, the condition $\beta<d_F$ holds, so the disease-free steady state $S^{\ast}_2$ is also stable.}
	\label{S5S6}
\end{figure}

\noindent Therefore
\begin{align*}
\mbox{sgn}\left\{\left.\dfrac{d\,\mbox{Re}(\lambda)}{d\,\tau_2}\right\vert_{\tau_2=\tau^{\ast}}\right\}&=\mbox{sgn}\left\{\left(\left.\dfrac{d\,\mbox{Re}(\lambda)}{d\,\tau_2}\right\vert_{\tau_2=\tau^{\ast}}\right)^{-1}\right\}=\mbox{sgn}\left\{\dfrac{h^{\prime}(s_0)}{\Lambda\,p_0^{(1)}(i\omega_0)}\right\}\\
&=\mbox{sgn}[p_0^{(1)}(i\omega_0)h^{\prime}(s_0)],
\end{align*}
which completes the proof.
\end{proof}

We can now formulate the main result concerning stability of the steady states $S^{\ast}_5$ and $S^{\ast}_6$.

\begin{Theorem}\label{Hopf bifurcation}
	Suppose the value of $T_{reg}^{\ast}$ satisfies conditions (\ref{S5 stability1}) and (\ref{S5 stability2}). If Equation (\ref{S2 equation hopf}) has at least one positive root $s_0$, and $p_0^{(1)}(i\omega_0)h^{\prime}(s_0)>0$ with $\omega_0=\sqrt{s_0}$, then the steady state $S^{\ast}_5$ (respectively, $S^{\ast}_6$) is stable for $0\leq \tau_2<\tau^{\ast}$, unstable for $\tau_2>\tau^{\ast}$, and undergoes a Hopf bifurcation at $\tau_2=\tau^{\ast}$.
\end{Theorem}

Since $T_{reg}^{\ast}$ satisfies conditions (\ref{S5 stability1}) and (\ref{S5 stability2}), the steady state $S^{\ast}_5/S^{\ast}_6$ is stable for $\tau_2=0$. Lemma \ref{lemma} then ensures that $\tau^{\ast}$ is the first positive value of the time delay $\tau_2$, for which the roots of the characteristic Equation (\ref{S2 equation}) cross the imaginary axis with positive speed. Hence, the steady state $S^{\ast}_5/S^{\ast}_6$ is stable for $0\leq \tau_2<\tau^{\ast}$, unstable for $\tau_2>\tau^{\ast}$, and undergoes a Hopf bifurcation at $\tau_2=\tau^{\ast}$.

\begin{Remark}
A similar result can be formulated for a subcritical Hopf bifurcation of the steady state $S^{\ast}_5/S^{\ast}_6$ at some higher value of $\tau_2$. 
\end{Remark}

The only remaining steady state is the persistent (chronic) equilibrium $S^{\ast}_8$ with all of its components being positive. Since it did not prove possible to find a closed form expression for this steady state, its~stability also has to be studied numerically.

\section{Numerical Stability Analysis and Simulations} \label{S4}

To investigate the role of different parameters in the dynamics of model (\ref{resc_model}), in this section we perform a detailed numerical bifurcation analysis and simulations of this model. Stability of different steady states is determined numerically by computing the largest real part of the characteristic eigenvalues, which is achieved by using a pseudospectral method implemented in a traceDDE suite in MATLAB \cite{Breda2006}.

Analytical results from the previous section suggest that at $\beta=d_F$, the disease-free steady state $S^{\ast}_2$ undergoes a transcritical bifurcation. For $\beta<d_F$, the disease-free steady state $S^{\ast}_2$ is stable, and the chronic steady state is infeasible. On the contrary, when $\beta>d_F$, the disease-free steady state $S^{\ast}_2$ is unstable, and in this case it makes sense to investigate stability of the chronic steady state. Therefore, these two cases are considered separately, and as a first step we fix the baseline values as given in Table~\ref{parameter table}. For this choice of parameters, we have $d_F-\beta>0$, implying that $S^{\ast}_2$ is always stable, and~Figure~\ref{S5S6} illustrates how the stability of $S^{\ast}_5$ and $S^{\ast}_6$ is affected by parameters. This figure indicates that the steady states $S^{\ast}_5$ and $S^{\ast}_6$ are only biologically feasible if the regulatory T cells do not grow too rapidly and do not clear autoreactive T cells too quickly. Importantly, Figure~\ref{S5S6} shows that the value of the rate $\delta_2$ of clearance of IL-2 by regulatory T cells does not have any effect on the thresholds of $\lambda_r$ and $\delta_1$, where the steady states $S^{\ast}_5$ and $S^{\ast}_6$ lose their feasibility. Moreover, if $\lambda_r$ and $\delta_1$ are small, then increasing the rate $\delta_2$ at which Tregs inhibit the production of IL-2 makes $S^{\ast}_6$ become unfeasible, resulting in a stable steady state $S^{\ast}_5$, which has the zero population of host cells $A$. On the other hand, the steady state $S^{\ast}_6$ associated with autoimmune responses is favoured for higher values of $\delta_1$ and $\lambda_r$. In~the case stable periodic solutions around these steady states, increasing $\delta_2$ results in the disappearance of oscillations and stabilisation of the associated steady state. At the intersection of the lines of Hopf bifurcation and the steady-state bifurcation, as determined by Theorem~\ref{Hopf bifurcation} and conditions (\ref{SS_bif}), one has the co-dimension two fold-Hopf (also known as zero-Hopf or saddle-node Hopf) bifurcation \cite{kuz98}.
\begin{table}
	\centering
	\caption{Table of parameter values.}
	\label{parameter table}
	\begin{tabular}{cccc}
		\toprule
		\textbf{Parameter}    & \textbf{Value}           & \textbf{Parameter}        & \textbf{Value}             \\ 
		\midrule
		$\beta$       & 1               & $\rho_3$         & 2                 \\
		$\mu_a$       & 20              & $d_n$            & 1                 \\ 
		$d_F$         & 1.1             & $d_a$            & 0.001             \\ 
		$\mu_F$       & 6               & $\delta_1$       & 0.0025             \\ 
		$d_{in}$      & 1               & $\delta_2$       & 0.001             \\ 
		$\alpha$      & 0.4             & $\sigma_1$       & 0.15              \\
		$\lambda_r$   & 3               & $\sigma_2$       & 0.33              \\ 
		$d_r$         & 0.4             & $d_i$            & 0.6               \\  
		$p_1$         & 0.4             & $\tau_1$         & 1.4               \\ 
		$p_2$         & 0.4             & $\tau_2$         & 0.6               \\
		$\rho_1$      & 10              & $\tau_3$         & 0.6               \\
		$\rho_2$      & 0.8             &                  &                   \\ 
	\bottomrule
	\end{tabular} 
\end{table}

Since our earlier analysis showed that stability of the steady states $S^{\ast}_5/S^{\ast}_6$ is affected by the time delay $\tau_2$, in Figure~\ref{trace DF} we consider stability of these equilibria depending on $\tau_2$ and the rate $\delta_2$. For the steady state $S^{\ast}_5$, if the effect of IL-2 on promoting proliferation of T cells is fast (i.e., $\tau_2$ is small), there~is a large range of $\delta_2$, starting with some very low values, for which $S^{\ast}_5$ is stable. Increasing the time delay $\tau_2$ results in the Hopf bifurcation of this steady state as described in Theorem~\ref{Hopf bifurcation}. One should note that for intermediate values of $\delta_2$, the~steady state $S^{\ast}_5$ undergoes stability switches, whereby~increasing the delay $\tau_2$ further results in a subcritical Hopf bifurcation, which stabilises $S^{\ast}_5$, but after some number of such stability switches eventually the steady state $S^{\ast}_5$ is unstable. For higher still values of $\delta_2$, the~steady state $S^{\ast}_5$ remains stable for an entire range of $\tau_2$ values, and the only way to lose its stability is via a steady state bifurcation as given by (\ref{SS_bif}). In the case of autoimmune steady state $S^{\ast}_6$, the~situation is somewhat different in that increasing $\delta_2$ beyond some critical values makes this steady state biologically infeasible. At the same time, for an entire range of $\delta_2$ values where it is feasible, this steady state exhibits a single loss of stability through a Hopf bifurcation for some critical value of the time delay $\tau_2$, in agreement with Theorem~\ref{Hopf bifurcation}.
\begin{figure}
	\centering
	\includegraphics[width=1 \linewidth]{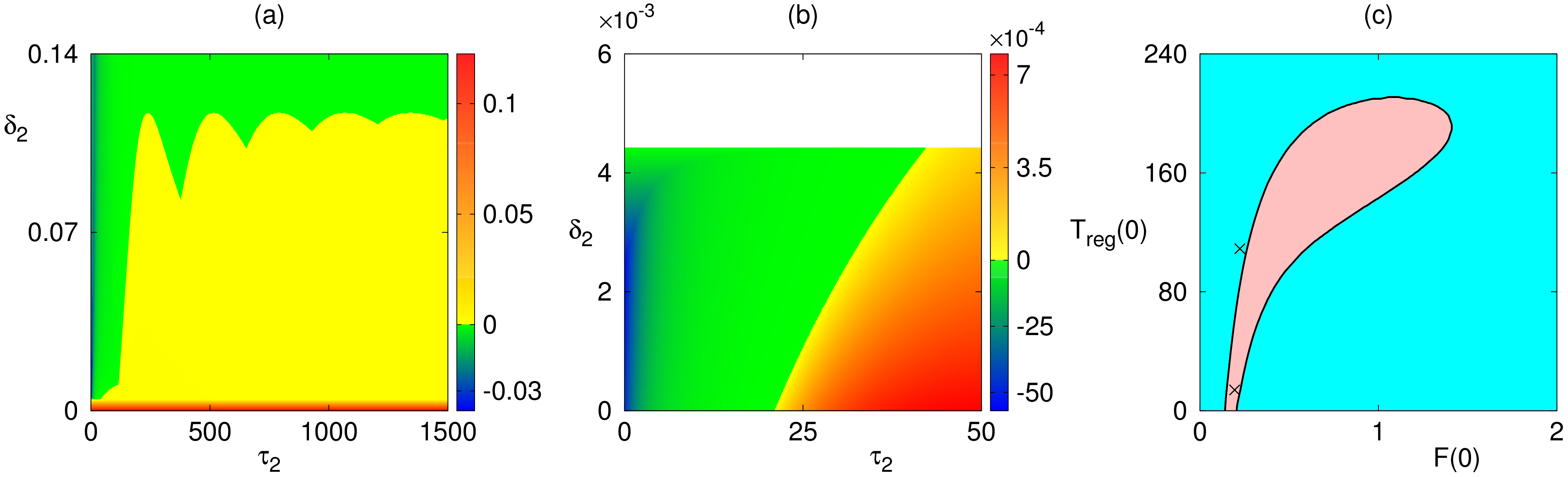}
	\caption{\small Stability of the steady states $S^{\ast}_5$ (\textbf{a}), and $S^{\ast}_6$ (\textbf{b}) with parameter values from Table~\ref{parameter table}. White~area shows the region where the steady state $S^{\ast}_6$ is infeasible. Colour code denotes max[Re($\lambda$)] for the steady states when they are feasible. In all plots the condition $d_F>\beta$ holds, so the disease-free steady state $S^{\ast}_2$ is stable. Basins of attraction of different steady states depending on the initial conditions (\textbf{c}), with~other initial conditions specified in~(\ref{initial condition}), and parameter values from Table \ref{parameter table}, except for $\tau_2=18$. Cyan and pink areas are the basins of attraction of $S^{\ast}_2$ and $S^{\ast}_6$, respectively.}
	\label{trace DF}
\end{figure}

As mentioned earlier, for parameter values used in Figure~\ref{trace DF}, the disease-free steady state $S^{\ast}_2$ is stable. Hence, the system exhibits a bi-stability between a disease-free state and either stable steady states $S^{\ast}_5/S^{\ast}_6$, or periodic solutions around these steady states. To investigate this bi-stability, we choose parameter values as in Table~\ref{parameter table} except for $\tau_2=18$, which corresponds to a stable steady state $S^{\ast}_6$, and we fix initial conditions for state variables as follows,
\begin{equation}\label{initial condition}
(A(0),T_{in}(0),T_{nor}(0),T_{aut}(0),I(0))=(0.9, 0.8, 0, 0, 0),
\end{equation} 
except for initial amounts of infected cells and regulatory T cells that are allowed to vary. Figure~\ref{trace DF}c illustrates the bi-stability between $S^{\ast}_2$ and $S^{\ast}_6$ in terms of their basins of attraction. It is worth noting that recently significant research in approximation theory and meshless interpolation has focused on developing techniques for detection and analysis of attraction basins \cite{cavoretto2015,cavoretto2016,de2016,francomano2016,cavoretto2017,francomano2018}. Figure~\ref{trace DF}c suggests that for very large initial amounts of regulatory T cells, the system converges to the disease-free steady state. It also indicates that if the initial amount of infected cells is very small or is bigger than some specific value, then the infection will be cleared.
Interestingly, increasing the initial amount of the regulatory T cells results in a larger range of initial amounts of infection, for which the system tends to a stable autoimmune state $S^{\ast}_6$. In Figure~\ref{trace DF}b we discovered that increasing $\tau_2$ makes the autoimmune steady state $S^{\ast}_6$ undergo a Hopf bifurcation, in which case the system will exhibit a bi-stability between
\noindent stable $S^{\ast}_2$ and a periodic solution around $S^{\ast}_6$. Our numerical investigation suggests that the shape of basins of attraction in this case is qualitatively similar to that shown in Figure~\ref{trace DF}c, with the basin of attraction of the stable steady state $S^{\ast}_6$ being replaced by the basin of attraction of the periodic solution around this steady state.

Figure \ref{simulation S1 S3} shows temporary evolution of the system (\ref{resc_model}) in the regime of bi-stability between a stable disease-free steady state and a stable autoimmune steady state $S^{\ast}_6$ (similar pattern of behaviour is exhibited in the case of bi-stability between $S^{\ast}_2$ and $S^{\ast}_5$). It also illustrates how the system develops a periodic solution around the steady state $S^{\ast}_6$ for a higher value of $\tau_2$. Periodic oscillations around the steady state $S^{\ast}_6$ biologically correspond to a genuine autoimmune state: after the initial infection is cleared, the system exhibits sustained endogenous oscillations, characterised by periods of significant reduction in the number of organ cells through a negative action of autoreactive T cells, separated by periods of quiescence. This type of behaviour is often observed in clinical manifestations of autoimmune disease \cite{blyu15,Bezra95,Davies97,Nyla12}. This result has substantial biological significance as effectively it suggests that even for the same kinetic parameters of immune response, the ultimate state of the system, which can be either a  successful clearance of infection without lasting consequences, or progression to autoimmunity, also depends on the strength of the initial infection and of the initial state of the immune system, as represented by the initial number of regulatory T cells.
\begin{figure}
	\centering
	\includegraphics[width=0.9 \linewidth]{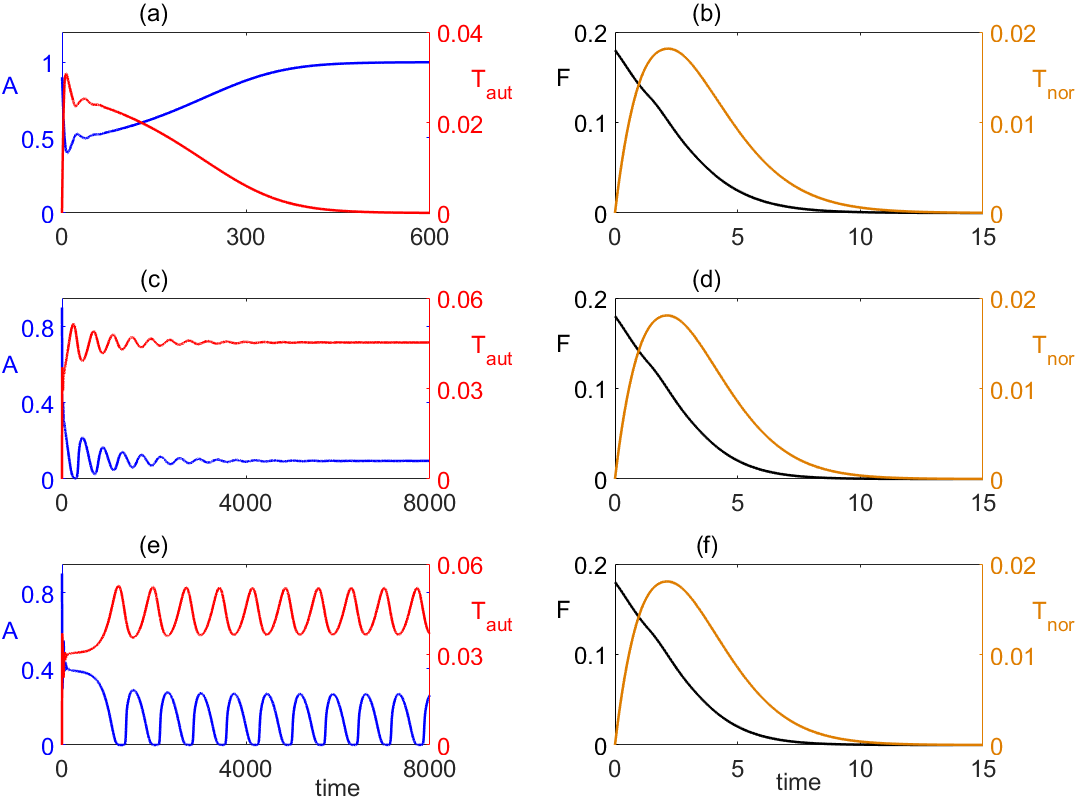}
	\caption{\small Numerical solutions of the model with parameters values from Table \ref{parameter table}, except for $\tau_2=18$. (\textbf{a},\textbf{b})~Stable disease-free steady state $S^{\ast}_2$ for $F(0)=0.18$, and $T_{reg}(0)=100$. (\textbf{c},\textbf{d}) Transient oscillations settling on a stable steady state $S^{\ast}_6$ for $F(0)=0.18$, and $T_{reg}(0)=10$. (\textbf{e},\textbf{f}) Autoimmune dynamics represented by periodic oscillations around the steady state $S^{\ast}_6$ for $\tau_2=32$, $F(0)=0.18$, and $T_{reg}(0)=10$.}
	\label{simulation S1 S3}
\end{figure}

Next we consider a situation where $\beta>d_F$, so the disease-free steady state is unstable, and~the system can have three steady states $S^{\ast}_5$, $S^{\ast}_6$ and $S^{\ast}_8$. Our earlier results \cite{Fatehi2018b} suggest that in the case where regulatory T cells do not inhibit the production of IL-2, i.e., for $\delta_2=0$, the steady state $S^{\ast}_6$ is stable.
\begin{figure}
	\centering
	\includegraphics[width=0.9 \linewidth]{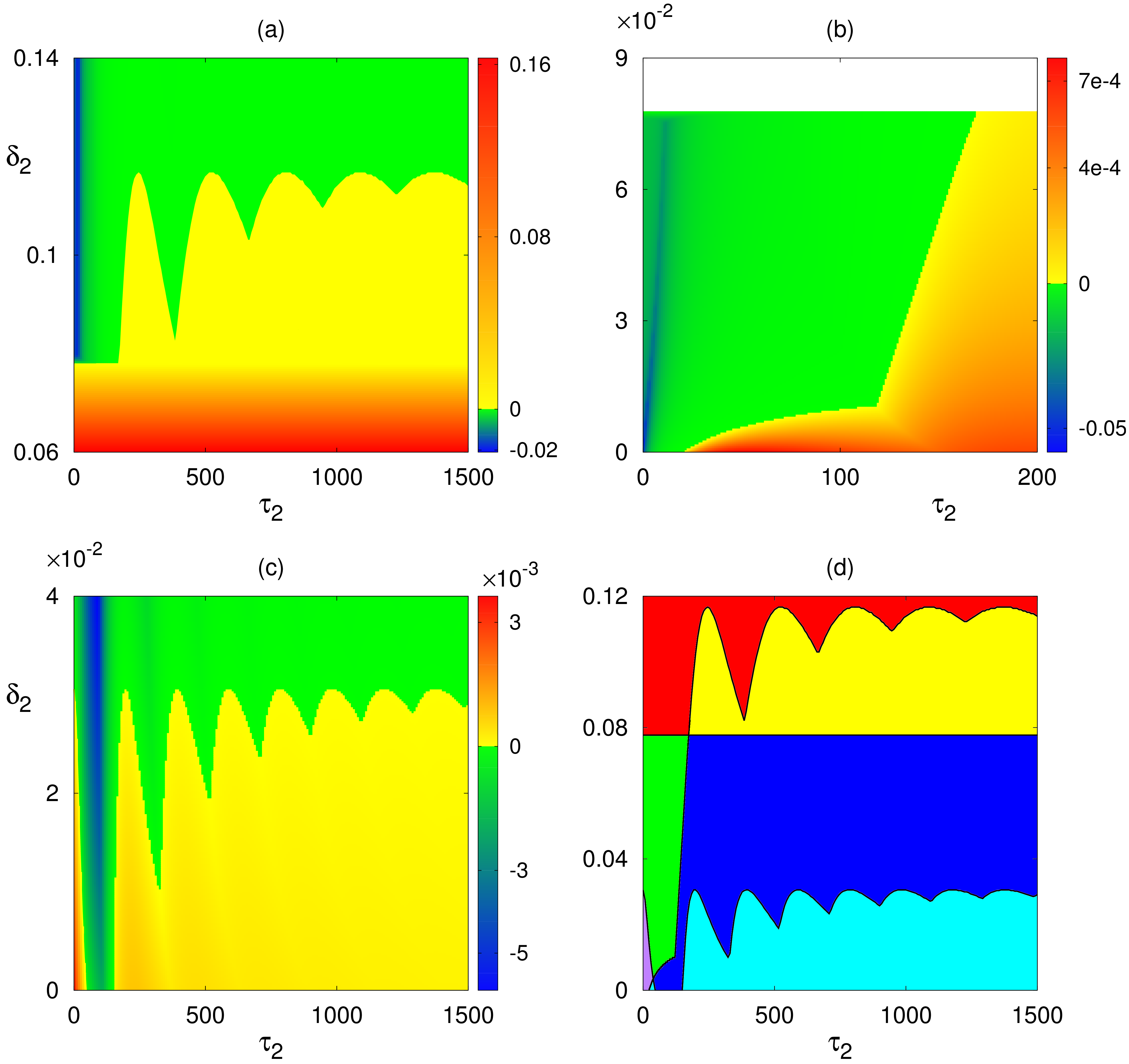} \vspace{-4pt}
	\caption{\small Stability of $S^{\ast}_5$ (\textbf{a}), $S^{\ast}_6$ (\textbf{b}), and $S^{\ast}_8$ (\textbf{c}), with parameter values from Table~\ref{parameter table}, except for $\beta=1.4$ and $\sigma_2=1$, so that $\beta>d_F$. White area shows the region where the steady state is infeasible. Colour code denotes max[Re($\lambda$)] for each steady states when it is feasible. (\textbf{d}) Summary of stability results. Green indicates the region where $S^{\ast}_6$ and $S^{\ast}_8$ are stable, and $S^{\ast}_5$ is unstable, whereas red is the area where $S^{\ast}_5$ and $S^{\ast}_8$ are stable, and $S^{\ast}_6$ is infeasible. Yellow is where $S^{\ast}_8$ is stable, $S^{\ast}_5$ is unstable, and $S^{\ast}_6$ is infeasible. Purple shows the region where $S^{\ast}_6$ is stable, but $S^{\ast}_5$ and $S^{\ast}_8$ are unstable. Blue and cyan indicate the regions where $S^{\ast}_5$ and $S^{\ast}_6$ are unstable, but $S^{\ast}_8$ is stable or unstable, respectively.}
	\label{trace endemic}
\end{figure}
Figure~\ref{trace endemic} shows regions of feasibility and stability of these steady states depending on $\delta_2$ and $\tau_2$ in this case. One observes that $S^{\ast}_5$ and $S^{\ast}_6$, whose stability boundaries are determined by Theorem~\ref{Hopf bifurcation}, exhibit~the same behaviour as in Figure~\ref{trace DF}, namely, for $S^{\ast}_5$ increasing $\tau_2$ causes multiple stability switches for smaller values of $\delta_2$, and the steady state is unstable for very small $\delta_2$ and stable for large $\delta_2$; in contrast, $S^{\ast}_5$ exhibits a single loss of stability via Hopf bifurcation at some critical value of the time delay $\tau_2$, which itself increases with $\delta_2$. Behaviour of $S^{\ast}_8$ is similar to that of $S^{\ast}_5$ in that there are multiple stability switches for increasing value of $\tau_2$ and small to intermediate values of $\delta_2$, while for high values of $\delta_2$, the chronic steady state $S^{\ast}_8$ is stable for all values of $\tau_2$. Figure~\ref{trace endemic}d divides the $\delta_2$-$\tau_2$ plane into different regions based on feasibility and stability of these steady states and shows that increasing $\delta_2$ makes the autoimmune steady state $S^{\ast}_6$ infeasible. In other regions, the system can exhibit a bi-stability between a stable steady state $S^{\ast}_8$ and either a stable steady state $S^{\ast}_5$, or a periodic solution around $S^{\ast}_5$.

Figure \ref{bistability endemic} illustrates the basins of attraction of the steady states $S^{\ast}_5$, $S^{\ast}_6$ and $S^{\ast}_8$, as well as periodic solutions around $S^{\ast}_8$. Figure~\ref{bistability endemic}a shows the basins of attraction of the steady states $S^{\ast}_5$ and $S^{\ast}_8$ and demonstrates that if the initial number of regulatory T cells or infected cells is sufficiently high, or the initial amount of infected cells is very low, the immune response neither eliminates infection nor clears autoreactive T cells, and the system approaches the stable steady state $S^{\ast}_5$. Figure~\ref{bistability endemic}b illustrates bi-stability between the stable steady state $S^{\ast}_6$ and a periodic solution around $S^{\ast}_8$, and has a different behaviour to than shown in Figure~\ref{bistability endemic}a. This figure suggests that for a specific range of $F(0)$ the system converges to a stable autoimmune state $S^{\ast}_6$ for all values of $T_{reg}(0)$. However, if the initial number of infected cells is very high or very low, the system instead develops a periodic solution around the steady state $S^{\ast}_8$ associated with chronic infection. 
\begin{figure}
	\centering
	\includegraphics[width=\linewidth]{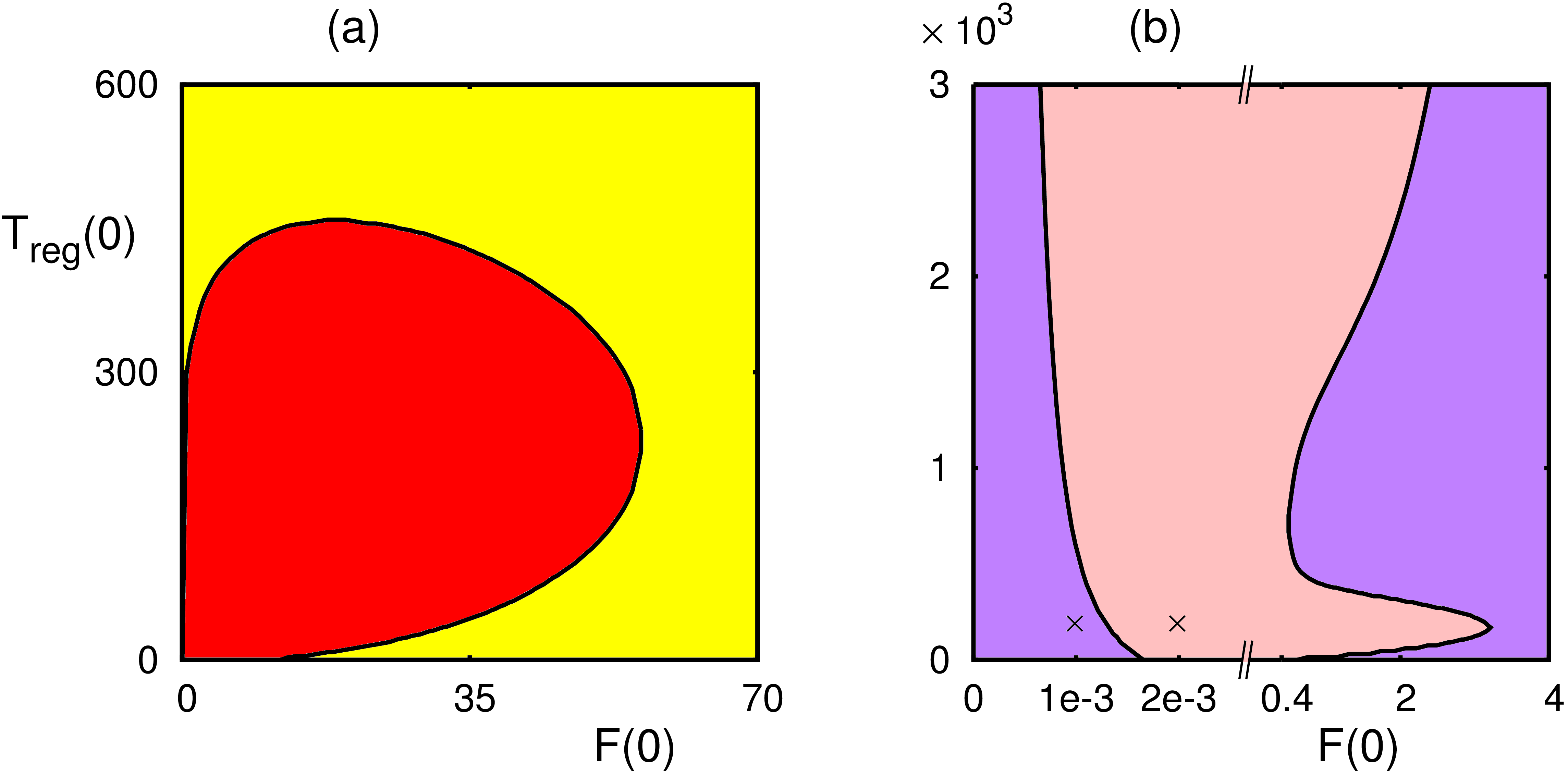}
	\caption{\small Bi-stability analysis of the steady states $S^{\ast}_5$, $S^{\ast}_6$, and $S^{\ast}_8$ with the same parameter values as in Figure~\ref{trace endemic}, except for (\textbf{a}) $\delta_2=0.1$, (\textbf{b}) $\delta_2=0.02$. Yellow indicates the basin of attraction of the chronic steady state $S^{\ast}_8$, purple is the basin of attraction of periodic solutions around $S^{\ast}_8$. Red and pink are the basins of attraction of the steady states $S^{\ast}_5$ and $S^{\ast}_6$, respectively.}
	\label{bistability endemic}
\end{figure}

Figure~\ref{simulation S3 S4} illustrates a regime of bi-stability between a stable steady state $S^{\ast}_6$ and a periodic solution around $S^{\ast}_8$ for combinations of initial conditions indicated by crossed in Figure~\ref{bistability endemic}b. It also illustrates how the system develops a stable solution around the steady state $S^{\ast}_8$ for a higher value of $\tau_2$. This~figure shows that by increasing the initial number of infected cells the behaviour of the system changes, as it then approaches the autoimmune steady state $S^{\ast}_6$. Interestingly, one can observe that for high values of $F(0)$ the system can eliminate the infection, but it cannot clear the autoreactive T cells, in which case the system converges to $S^{\ast}_6$. On the other hand, for a smaller number of infected cells the system develops a periodic solution around the endemic steady state.

Figure \ref{S4 taus} shows how the stability of the chronic infection steady state $S^{\ast}_8$ changes with respect to time delays. Figure~\ref{S4 taus}a indicates that for small values of $\tau_2$ (i.e., when the influence of IL-2 on proliferation of T cells is occurring quite rapidly), the steady state $S^{\ast}_8$ is stable, and increasing the time delay $\tau_1$ associated with viral eclipse phase does not have an effect on its stability. At the same time, {\bf if} $\tau_2$ exceeds some specific value, by increasing $\tau_1$ the chronic steady state switches between being stable or unstable. Figure \ref{S4 taus}b demonstrates a different behaviour, suggesting that for each value of $\tau_1$, there is small range of $\tau_3$ values where $S^{\ast}_8$ is stable, but for smaller and larger values of $\tau_3$ it is unstable. For intermediate values of the eclipse phase delay $\tau_1$, there is an additional narrow range of $\tau_3$ values where $S^{\ast}_8$ is stable. Figure~\ref{S4 taus}c illustrates that for very small, respectively very large, values of $\tau_3$, the~chronic infection steady state is stable, respectively unstable for any value of $\tau_2$; for intermediate values of $\tau_3$, this steady state undergoes a finite number of stability switches for increasing values of $\tau_2$ and eventually becomes unstable.

It should be noted Figure \ref{S4 taus} shows that unlike $\tau_1$ and $\tau_2$, once the steady state $S^{\ast}_8$ loses stability via Hopf bifurcation due to increasing $\tau_3$, it cannot regain stability for higher values of $\tau_3$.
\begin{figure}
	\centering
	\includegraphics[width=\linewidth]{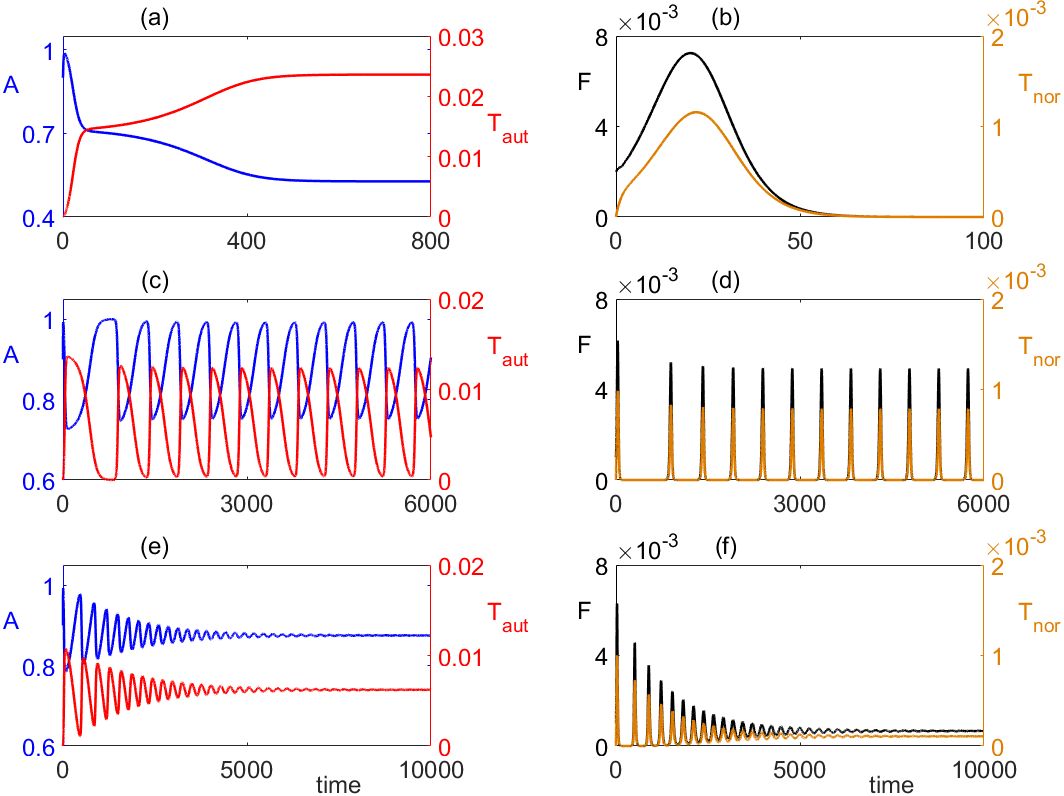}
	\caption{\small Numerical solutions of the model with same parameters values as Figure \ref{bistability endemic}b. (\textbf{a},\textbf{b}) Stable steady state $S^{\ast}_6$ for $F(0)=0.002$ and $T_{reg}(0)=200$. (\textbf{c},\textbf{d}) Periodic oscillations around the steady state $S^{\ast}_8$ for $F(0)=0.001$ and $T_{reg}(0)=200$. (\textbf{e},\textbf{f}) Transient oscillations settling on a stable steady state $S^{\ast}_8$ for $\tau_2=25$, $F(0)=0.001$ and $T_{reg}(0)=200$.}
	\label{simulation S3 S4}
\end{figure}
\unskip
\begin{figure}
	\centering
	\includegraphics[width=\linewidth]{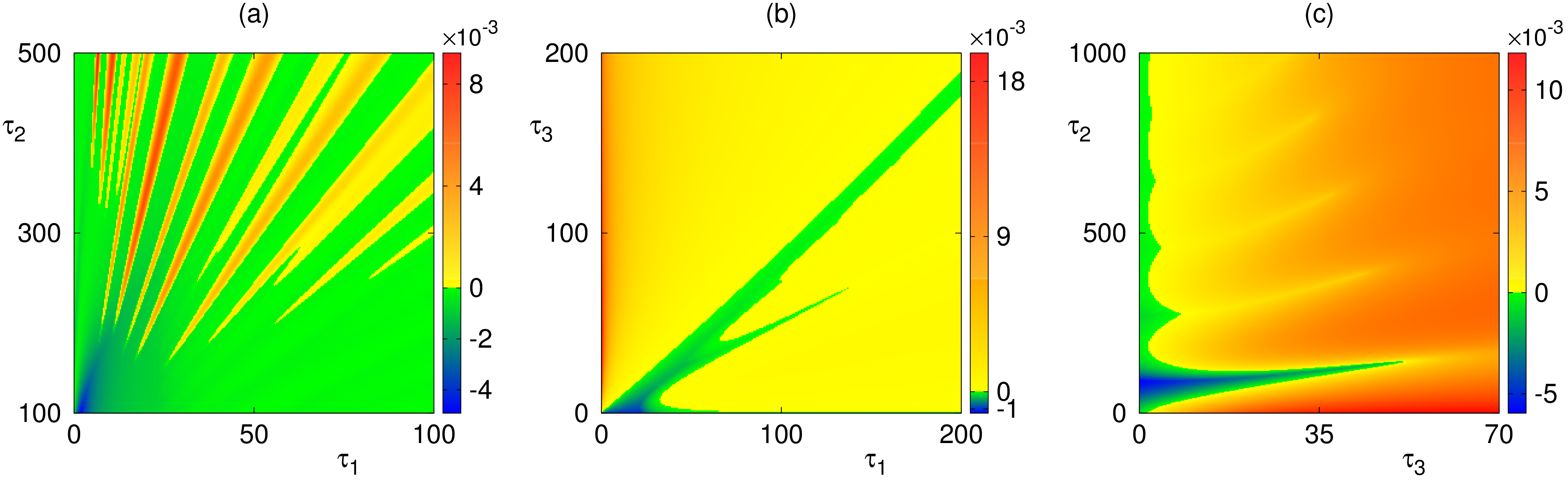}
	\caption{\small Colour code denotes max[Re($\lambda$)] for the endemic steady state $S^{\ast}_8$ depending on $\tau_1-\tau_2$ (\textbf{a}), $\tau_1-\tau_3$ (\textbf{b}), and $\tau_3-\tau_2$ (\textbf{c}), with the parameter values taken from Table \ref{parameter table}, except for $\beta=1.4$, $\sigma_2=1$, and $\delta_2=0.04$.}
	\label{S4 taus}
\end{figure}

\section{Conclusions}  \label{S5}

In this paper, we have developed and analysed a time-delayed model of immune response to a viral infection, which accounts for T cells with different activation thresholds, a cytokine mediating T~cell proliferation, as well as regulatory T cells. Particular attention is payed to the dual suppressive role of regulatory T cells in terms of reducing the amount of autoreactive T cells, and also inhibiting IL-2. To~achieve better biological realism of the model, we have explicitly included time delays associated with the eclipse phase of the virus life cycle, stimulation/proliferation of T cells by IL-2, and suppression of IL-2 by regulatory T cells. Depending on the values of parameters, the system can have four steady states: the disease-free state, the state characterised by the death of host cells, the autoimmune state, and a state of chronic infection. We have established conditions for stability and steady-state or Hopf bifurcations of these steady states in terms of system parameters.

In the case where the natural death rate of infected cells exceeds the infection rate, the immune system is able to clear the infection, and the disease-fee steady state is stable. In this regime, the system can also support the autoimmune steady state or the steady state with the death of host cells, either of which can be stable, or give rise to a periodic solution emerging via Hopf bifurcation. In the opposite case, when the natural death rate of infected cells is smaller than the infection rate, the disease-fee steady state is unstable, but it is possible to have a bi-stability between the other three steady states or periodic solutions around them. To better understand bi-stability between different dynamical regimes, we have used numerical simulations to identify basins of attraction of different steady states and periodic solutions depending on the initial level of infection and the initial number of regulatory T cells. The fact that for the same parameter values the system can exhibit bi-stability between a disease-free steady state and an autoimmune state, represented by sustained periodic oscillations following the clearance of infection, is very important from a clinical point of view, as effectively it suggests that the progress and eventual outcome of viral infection is also determined by the strength of infection and the initial state of the immune system. One counter-intuitive observation is that in the case of bi-stability with a disease-free steady state, for higher initial numbers of regulatory T cells, the autoimmune steady state is actually stable for a wider range of initial levels of infection. In~this regime of bi-stability, increasing the time delay associated with the positive impact of IL-2 on proliferation of T cells results in the loss of stability of autoimmune steady state and emergence of autoimmune dynamics, characterised by stable periodic oscillations. On the contrary, in the case where the disease-free steady state is unstable, increasing this time delay results in stabilisation of the chronic infection.

There are several directions in which the work presented in this paper can be extended. One~direction is exploration of the contributions from other components of immune response, more specifically, antibodies and memory T cells, to the onset and progress of autoimmunity \cite{Skap05,Antia05}. This is particularly important from the perspective that clinically the onset of autoimmune disease is often taking place on a much longer scale than the timescale of a regular immune response to a viral infection, so~memory T cells can be expected to play a more substantial role. While our model has focused on one specific growth cytokine IL-2, a number of other cytokines, such as IL-7 \cite{Schluns00}, TNF$\beta$ and IL-10 \cite{Saka04}, are known to significantly affect homeostasis and proliferation of different types of T cells, as well as mediate their efficiency in eliminating the infection. Including these immune mediators explicitly in the model can provide further significant insights into the dynamics of immune response, as has been recently demonstrated on the example of a detailed model of immune response to hepatitis B \cite{Fatehi2018d}. Another~biologically relevant and mathematically challenging problem is the investigation of the interplay between stochasticity, which is known to be an intrinsic feature of immune response \cite{pere97,Fatehi2018a}, and effects of time delays associated with various aspects of immune dynamics.

\section*{Appendix}\label{A1}

Coefficients of Equation (\ref{S2 equation hopf}) for Hopf frequency are given below.\\

{\small
$b_2=\Big( -{\delta_{{1}}}^{4}{\delta_{{2}}}^{4}{\rho_{{1}}}^{2}{Y}^{2}+2
\,{d_{{r}}}^{2}{\delta_{{1}}}^{4}{\delta_{{2}}}^{4} \Big) {{T_{reg}^{\ast}}}
^{8}+ \Big( -2\,{\delta_{{1}}}^{3}{\delta_{{2}}}^{3}\rho_{{1}}
\Big( d_{{a}}\delta_{{2}}\rho_{{1}}+d_{{i}}\delta_{{1}}\rho_{{1}}-2
\,d_{{r}}\delta_{{2}}\rho_{{3}} \Big) {I^{\ast}}^{2}+8\,{d_{{r}}}^{2}{
	\delta_{{1}}}^{3}{\delta_{{2}}}^{3} \Big( d_{{a}}\delta_{{2}}+d_{{i}}
\delta_{{1}} \Big)  \Big) {{T_{reg}^{\ast}}}^{7}+ \Big( -2\,{\delta_{{1}}
}^{2}{\delta_{{2}}}^{3}{\rho_{{1}}}^{2}\rho_{{3}} \Big( 4\,\delta_{{1
}}\rho_{{1}}+\delta_{{2}}\rho_{{1}}+5\,\delta_{{2}}\rho_{{3}} \Big) 
{I^{\ast}}^{4}-{\delta_{{1}}}^{2}{\delta_{{2}}}^{2} \Big( {d_{{a}}}^{2}{
	\delta_{{2}}}^{2}{\rho_{{1}}}^{2}+2\,d_{{a}}d_{{i}}\delta_{{1}}\delta_
{{2}}{\rho_{{1}}}^{2}-8\,d_{{a}}d_{{r}}{\delta_{{2}}}^{2}\rho_{{1}}
\rho_{{3}}+{d_{{i}}}^{2}{\delta_{{1}}}^{2}{\rho_{{1}}}^{2}-12\,d_{{i}}
d_{{r}}\delta_{{1}}\delta_{{2}}\rho_{{1}}\rho_{{3}}+2\,{d_{{r}}}^{2}{
	\delta_{{1}}}^{2}{\rho_{{1}}}^{2}+2\,{d_{{r}}}^{2}\delta_{{1}}\delta_{
	{2}}{\rho_{{1}}}^{2}+6\,{d_{{r}}}^{2}\delta_{{1}}\delta_{{2}}\rho_{{1}
}\rho_{{3}}+{d_{{r}}}^{2}{\delta_{{2}}}^{2}{\rho_{{1}}}^{2}+8\,{d_{{r}
}}^{2}{\delta_{{2}}}^{2}{\rho_{{3}}}^{2} \Big) {I^{\ast}}^{2}+2\,{d_{{r}}}^
{2}{\delta_{{1}}}^{2}{\delta_{{2}}}^{2} \Big( 6\,{d_{{a}}}^{2}{\delta
	_{{2}}}^{2}+16\,d_{{a}}d_{{i}}\delta_{{1}}\delta_{{2}}+6\,{d_{{i}}}^{2
}{\delta_{{1}}}^{2}+{d_{{r}}}^{2}{\delta_{{1}}}^{2}+{d_{{r}}}^{2}{
	\delta_{{2}}}^{2} \Big)  \Big) {{T_{reg}^{\ast}}}^{6}+ \Big( -2\,\rho_{{3
}}{\delta_{{2}}}^{2}\rho_{{1}}\delta_{{1}} \Big( 6\,d_{{a}}\delta_{{1
}}\delta_{{2}}{\rho_{{1}}}^{2}+6\,d_{{a}}{\delta_{{2}}}^{2}\rho_{{1}}
\rho_{{3}}+9\,d_{{i}}{\delta_{{1}}}^{2}{\rho_{{1}}}^{2}+5\,d_{{i}}
\delta_{{1}}\delta_{{2}}{\rho_{{1}}}^{2}+22\,d_{{i}}\delta_{{1}}\delta
_{{2}}\rho_{{1}}\rho_{{3}}+6\,d_{{r}}{\delta_{{1}}}^{2}{\rho_{{1}}}^{2
}+2\,d_{{r}}\delta_{{1}}\delta_{{2}}{\rho_{{1}}}^{2}-4\,d_{{r}}{\delta
	_{{2}}}^{2}{\rho_{{3}}}^{2} \Big) {I^{\ast}}^{4}+2\,\delta_{{1}}\delta_{{2}
} \Big( {d_{{a}}}^{2}d_{{i}}\delta_{{1}}{\delta_{{2}}}^{2}{\rho_{{1}}
}^{2}+2\,{d_{{a}}}^{2}d_{{r}}{\delta_{{2}}}^{3}\rho_{{1}}\rho_{{3}}+d_
{{a}}{d_{{i}}}^{2}{\delta_{{1}}}^{2}\delta_{{2}}{\rho_{{1}}}^{2}+10\,d
_{{a}}d_{{i}}d_{{r}}\delta_{{1}}{\delta_{{2}}}^{2}\rho_{{1}}\rho_{{3}}
-3\,d_{{a}}{d_{{r}}}^{2}{\delta_{{1}}}^{2}\delta_{{2}}{\rho_{{1}}}^{2}
-2\,d_{{a}}{d_{{r}}}^{2}\delta_{{1}}{\delta_{{2}}}^{2}{\rho_{{1}}}^{2}
-8\,d_{{a}}{d_{{r}}}^{2}\delta_{{1}}{\delta_{{2}}}^{2}\rho_{{1}}\rho_{
	{3}}-8\,d_{{a}}{d_{{r}}}^{2}{\delta_{{2}}}^{3}{\rho_{{3}}}^{2}+6\,{d_{
		{i}}}^{2}d_{{r}}{\delta_{{1}}}^{2}\delta_{{2}}\rho_{{1}}\rho_{{3}}-d_{
	{i}}{d_{{r}}}^{2}{\delta_{{1}}}^{3}{\rho_{{1}}}^{2}-3\,d_{{i}}{d_{{r}}
}^{2}{\delta_{{1}}}^{2}\delta_{{2}}{\rho_{{1}}}^{2}-8\,d_{{i}}{d_{{r}}
}^{2}{\delta_{{1}}}^{2}\delta_{{2}}\rho_{{1}}\rho_{{3}}-d_{{i}}{d_{{r}
}}^{2}\delta_{{1}}{\delta_{{2}}}^{2}{\rho_{{1}}}^{2}-2\,d_{{i}}{d_{{r}
}}^{2}\delta_{{1}}{\delta_{{2}}}^{2}\rho_{{1}}\rho_{{3}}-16\,d_{{i}}{d
	_{{r}}}^{2}\delta_{{1}}{\delta_{{2}}}^{2}{\rho_{{3}}}^{2}+{d_{{r}}}^{3
}{\delta_{{1}}}^{2}\delta_{{2}}\rho_{{1}}\rho_{{3}}+2\,{d_{{r}}}^{3}
\delta_{{1}}{\delta_{{2}}}^{2}\rho_{{1}}\rho_{{3}}+2\,{d_{{r}}}^{3}{
	\delta_{{2}}}^{3}\rho_{{1}}\rho_{{3}} \Big) {I^{\ast}}^{2}+4\,{d_{{r}}}^{2}
\delta_{{1}}\delta_{{2}} \Big( 2\,{d_{{a}}}^{3}{\delta_{{2}}}^{3}+12
\,{d_{{a}}}^{2}d_{{i}}\delta_{{1}}{\delta_{{2}}}^{2}+12\,d_{{a}}{d_{{i
}}}^{2}{\delta_{{1}}}^{2}\delta_{{2}}+2\,d_{{a}}{d_{{r}}}^{2}{\delta_{
		{1}}}^{2}\delta_{{2}}+d_{{a}}{d_{{r}}}^{2}{\delta_{{2}}}^{3}+2\,{d_{{i
}}}^{3}{\delta_{{1}}}^{3}+d_{{i}}{d_{{r}}}^{2}{\delta_{{1}}}^{3}+2\,d_
{{i}}{d_{{r}}}^{2}\delta_{{1}}{\delta_{{2}}}^{2} \Big)  \Big) {{T_{reg}^{\ast}}}^{5}+ \Big( -2\,{\delta_{{2}}}^{2}{\rho_{{1}}}^{2}{\rho_{{3}}
}^{2} \Big( {\delta_{{1}}}^{2}{\rho_{{1}}}^{2}+\delta_{{1}}\delta_{{2
}}{\rho_{{1}}}^{2}+3\,\delta_{{1}}\delta_{{2}}\rho_{{1}}\rho_{{3}}+2\,
{\delta_{{2}}}^{2}{\rho_{{3}}}^{2} \Big) {I^{\ast}}^{6}-2\,\rho_{{3}}\delta
_{{2}}\rho_{{1}} \Big( 2\,{d_{{a}}}^{2}\delta_{{1}}{\delta_{{2}}}^{2}
{\rho_{{1}}}^{2}+{d_{{a}}}^{2}{\delta_{{2}}}^{3}\rho_{{1}}\rho_{{3}}+9
\,d_{{a}}d_{{i}}{\delta_{{1}}}^{2}\delta_{{2}}{\rho_{{1}}}^{2}+24\,d_{
	{a}}d_{{i}}\delta_{{1}}{\delta_{{2}}}^{2}\rho_{{1}}\rho_{{3}}+12\,d_{{
		a}}d_{{r}}{\delta_{{1}}}^{2}\delta_{{2}}{\rho_{{1}}}^{2}-4\,d_{{a}}d_{
	{r}}{\delta_{{2}}}^{3}{\rho_{{3}}}^{2}+4\,{d_{{i}}}^{2}{\delta_{{1}}}^
{3}{\rho_{{1}}}^{2}+8\,{d_{{i}}}^{2}{\delta_{{1}}}^{2}\delta_{{2}}{
	\rho_{{1}}}^{2}+33\,{d_{{i}}}^{2}{\delta_{{1}}}^{2}\delta_{{2}}\rho_{{
		1}}\rho_{{3}}+9\,d_{{i}}d_{{r}}{\delta_{{1}}}^{3}{\rho_{{1}}}^{2}+9\,d
_{{i}}d_{{r}}{\delta_{{1}}}^{2}\delta_{{2}}{\rho_{{1}}}^{2}+2\,d_{{i}}
d_{{r}}\delta_{{1}}{\delta_{{2}}}^{2}{\rho_{{1}}}^{2}-20\,d_{{i}}d_{{r
}}\delta_{{1}}{\delta_{{2}}}^{2}{\rho_{{3}}}^{2}-4\,{d_{{r}}}^{2}{
	\delta_{{1}}}^{2}\delta_{{2}}\rho_{{1}}\rho_{{3}}-6\,{d_{{r}}}^{2}
\delta_{{1}}{\delta_{{2}}}^{2}\rho_{{1}}\rho_{{3}}-4\,{d_{{r}}}^{2}
\delta_{{1}}{\delta_{{2}}}^{2}{\rho_{{3}}}^{2}+{d_{{r}}}^{2}{\delta_{{
			2}}}^{3}\rho_{{1}}\rho_{{3}}-4\,{d_{{r}}}^{2}{\delta_{{2}}}^{3}{\rho_{
		{3}}}^{2} \Big) {I^{\ast}}^{4}+ \Big( 2\,{d_{{a}}}^{3}d_{{i}}\delta_{{1}}{
	\delta_{{2}}}^{3}{\rho_{{1}}}^{2}+6\,{d_{{a}}}^{2}{d_{{i}}}^{2}{\delta
	_{{1}}}^{2}{\delta_{{2}}}^{2}{\rho_{{1}}}^{2}+4\,{d_{{a}}}^{2}d_{{i}}d
_{{r}}\delta_{{1}}{\delta_{{2}}}^{3}\rho_{{1}}\rho_{{3}}-7\,{d_{{a}}}^
{2}{d_{{r}}}^{2}{\delta_{{1}}}^{2}{\delta_{{2}}}^{2}{\rho_{{1}}}^{2}-2
\,{d_{{a}}}^{2}{d_{{r}}}^{2}\delta_{{1}}{\delta_{{2}}}^{3}{\rho_{{1}}}
^{2}-14\,{d_{{a}}}^{2}{d_{{r}}}^{2}\delta_{{1}}{\delta_{{2}}}^{3}\rho_
{{1}}\rho_{{3}}-8\,{d_{{a}}}^{2}{d_{{r}}}^{2}{\delta_{{2}}}^{4}{\rho_{
		{3}}}^{2}+2\,d_{{a}}{d_{{i}}}^{3}{\delta_{{1}}}^{3}\delta_{{2}}{\rho_{
		{1}}}^{2}+12\,d_{{a}}{d_{{i}}}^{2}d_{{r}}{\delta_{{1}}}^{2}{\delta_{{2
}}}^{2}\rho_{{1}}\rho_{{3}}-6\,d_{{a}}d_{{i}}{d_{{r}}}^{2}{\delta_{{1}
}}^{3}\delta_{{2}}{\rho_{{1}}}^{2}-12\,d_{{a}}d_{{i}}{d_{{r}}}^{2}{
	\delta_{{1}}}^{2}{\delta_{{2}}}^{2}{\rho_{{1}}}^{2}-42\,d_{{a}}d_{{i}}
{d_{{r}}}^{2}{\delta_{{1}}}^{2}{\delta_{{2}}}^{2}\rho_{{1}}\rho_{{3}}+
2\,d_{{a}}d_{{i}}{d_{{r}}}^{2}\delta_{{1}}{\delta_{{2}}}^{3}{\rho_{{1}
}}^{2}-8\,d_{{a}}d_{{i}}{d_{{r}}}^{2}\delta_{{1}}{\delta_{{2}}}^{3}
\rho_{{1}}\rho_{{3}}-64\,d_{{a}}d_{{i}}{d_{{r}}}^{2}\delta_{{1}}{
	\delta_{{2}}}^{3}{\rho_{{3}}}^{2}+2\,d_{{a}}{d_{{r}}}^{3}{\delta_{{1}}
}^{2}{\delta_{{2}}}^{2}\rho_{{1}}\rho_{{3}}+8\,d_{{a}}{d_{{r}}}^{3}
\delta_{{1}}{\delta_{{2}}}^{3}\rho_{{1}}\rho_{{3}}+4\,{d_{{i}}}^{3}d_{
	{r}}{\delta_{{1}}}^{3}\delta_{{2}}\rho_{{1}}\rho_{{3}}-{d_{{i}}}^{2}{d
	_{{r}}}^{2}{\delta_{{1}}}^{4}{\rho_{{1}}}^{2}-6\,{d_{{i}}}^{2}{d_{{r}}
}^{2}{\delta_{{1}}}^{3}\delta_{{2}}{\rho_{{1}}}^{2}-14\,{d_{{i}}}^{2}{
	d_{{r}}}^{2}{\delta_{{1}}}^{3}\delta_{{2}}\rho_{{1}}\rho_{{3}}-2\,{d_{
		{i}}}^{2}{d_{{r}}}^{2}{\delta_{{1}}}^{2}{\delta_{{2}}}^{2}{\rho_{{1}}}
^{2}-12\,{d_{{i}}}^{2}{d_{{r}}}^{2}{\delta_{{1}}}^{2}{\delta_{{2}}}^{2
}\rho_{{1}}\rho_{{3}}-48\,{d_{{i}}}^{2}{d_{{r}}}^{2}{\delta_{{1}}}^{2}
{\delta_{{2}}}^{2}{\rho_{{3}}}^{2}+10\,d_{{i}}{d_{{r}}}^{3}{\delta_{{1
}}}^{2}{\delta_{{2}}}^{2}\rho_{{1}}\rho_{{3}}+12\,d_{{i}}{d_{{r}}}^{3}
\delta_{{1}}{\delta_{{2}}}^{3}\rho_{{1}}\rho_{{3}}-5\,{d_{{r}}}^{4}{
	\delta_{{1}}}^{2}{\delta_{{2}}}^{2}{\rho_{{3}}}^{2}-4\,{d_{{r}}}^{4}{
	\delta_{{2}}}^{4}{\rho_{{3}}}^{2} \Big) {I^{\ast}}^{2}+2\,{d_{{r}}}^{2}
\Big( {d_{{a}}}^{4}{\delta_{{2}}}^{4}+16\,{d_{{a}}}^{3}d_{{i}}\delta
_{{1}}{\delta_{{2}}}^{3}+36\,{d_{{a}}}^{2}{d_{{i}}}^{2}{\delta_{{1}}}^
{2}{\delta_{{2}}}^{2}+6\,{d_{{a}}}^{2}{d_{{r}}}^{2}{\delta_{{1}}}^{2}{
	\delta_{{2}}}^{2}+{d_{{a}}}^{2}{d_{{r}}}^{2}{\delta_{{2}}}^{4}+16\,d_{
	{a}}{d_{{i}}}^{3}{\delta_{{1}}}^{3}\delta_{{2}}+8\,d_{{a}}d_{{i}}{d_{{
			r}}}^{2}{\delta_{{1}}}^{3}\delta_{{2}}+8\,d_{{a}}d_{{i}}{d_{{r}}}^{2}
\delta_{{1}}{\delta_{{2}}}^{3}+{d_{{i}}}^{4}{\delta_{{1}}}^{4}+{d_{{i}
}}^{2}{d_{{r}}}^{2}{\delta_{{1}}}^{4}+6\,{d_{{i}}}^{2}{d_{{r}}}^{2}{
	\delta_{{1}}}^{2}{\delta_{{2}}}^{2} \Big)  \Big) {{T_{reg}^{\ast}}}^{4}+
\Big( -2\,\delta_{{2}}{\rho_{{1}}}^{2}{\rho_{{3}}}^{2} \Big( d_{{a}
}\delta_{{1}}\delta_{{2}}{\rho_{{1}}}^{2}+2\,d_{{a}}{\delta_{{2}}}^{2}
\rho_{{1}}\rho_{{3}}+2\,d_{{i}}{\delta_{{1}}}^{2}{\rho_{{1}}}^{2}+5\,d
_{{i}}\delta_{{1}}\delta_{{2}}{\rho_{{1}}}^{2}+14\,d_{{i}}\delta_{{1}}
\delta_{{2}}\rho_{{1}}\rho_{{3}}+2\,d_{{i}}{\delta_{{2}}}^{2}\rho_{{1}
}\rho_{{3}}+12\,d_{{i}}{\delta_{{2}}}^{2}{\rho_{{3}}}^{2}-d_{{r}}
\delta_{{1}}\delta_{{2}}\rho_{{1}}\rho_{{3}}-2\,d_{{r}}{\delta_{{2}}}^
{2}\rho_{{1}}\rho_{{3}} \Big) {I^{\ast}}^{6}+2\,\rho_{{3}}\rho_{{1}}
\Big( 2\,{d_{{a}}}^{2}d_{{i}}\delta_{{1}}{\delta_{{2}}}^{2}{\rho_{{1
}}}^{2}-2\,{d_{{a}}}^{2}d_{{i}}{\delta_{{2}}}^{3}\rho_{{1}}\rho_{{3}}-
7\,{d_{{a}}}^{2}d_{{r}}\delta_{{1}}{\delta_{{2}}}^{2}{\rho_{{1}}}^{2}+
4\,d_{{a}}{d_{{i}}}^{2}{\delta_{{1}}}^{2}\delta_{{2}}{\rho_{{1}}}^{2}+
2\,d_{{a}}{d_{{i}}}^{2}\delta_{{1}}{\delta_{{2}}}^{2}{\rho_{{1}}}^{2}-
30\,d_{{a}}{d_{{i}}}^{2}\delta_{{1}}{\delta_{{2}}}^{2}\rho_{{1}}\rho_{
	{3}}-15\,d_{{a}}d_{{i}}d_{{r}}{\delta_{{1}}}^{2}\delta_{{2}}{\rho_{{1}
}}^{2}-4\,d_{{a}}d_{{i}}d_{{r}}\delta_{{1}}{\delta_{{2}}}^{2}{\rho_{{1
}}}^{2}+20\,d_{{a}}d_{{i}}d_{{r}}{\delta_{{2}}}^{3}{\rho_{{3}}}^{2}+8
\,d_{{a}}{d_{{r}}}^{2}\delta_{{1}}{\delta_{{2}}}^{2}\rho_{{1}}\rho_{{3
}}+4\,d_{{a}}{d_{{r}}}^{2}{\delta_{{2}}}^{3}{\rho_{{3}}}^{2}-4\,{d_{{i
}}}^{3}{\delta_{{1}}}^{2}\delta_{{2}}{\rho_{{1}}}^{2}-20\,{d_{{i}}}^{3
}{\delta_{{1}}}^{2}\delta_{{2}}\rho_{{1}}\rho_{{3}}-2\,{d_{{i}}}^{2}d_
{{r}}{\delta_{{1}}}^{3}{\rho_{{1}}}^{2}-12\,{d_{{i}}}^{2}d_{{r}}{
	\delta_{{1}}}^{2}\delta_{{2}}{\rho_{{1}}}^{2}-7\,{d_{{i}}}^{2}d_{{r}}
\delta_{{1}}{\delta_{{2}}}^{2}{\rho_{{1}}}^{2}+36\,{d_{{i}}}^{2}d_{{r}
}\delta_{{1}}{\delta_{{2}}}^{2}{\rho_{{3}}}^{2}+6\,d_{{i}}{d_{{r}}}^{2
}{\delta_{{1}}}^{2}\delta_{{2}}\rho_{{1}}\rho_{{3}}+20\,d_{{i}}{d_{{r}
}}^{2}\delta_{{1}}{\delta_{{2}}}^{2}\rho_{{1}}\rho_{{3}}+12\,d_{{i}}{d
	_{{r}}}^{2}\delta_{{1}}{\delta_{{2}}}^{2}{\rho_{{3}}}^{2}-2\,d_{{i}}{d
	_{{r}}}^{2}{\delta_{{2}}}^{3}\rho_{{1}}\rho_{{3}}+20\,d_{{i}}{d_{{r}}}
^{2}{\delta_{{2}}}^{3}{\rho_{{3}}}^{2}+3\,{d_{{r}}}^{3}\delta_{{1}}{
	\delta_{{2}}}^{2}{\rho_{{3}}}^{2}-4\,{d_{{r}}}^{3}{\delta_{{2}}}^{3}{
	\rho_{{3}}}^{2} \Big) {I^{\ast}}^{4}+ \Big( 2\,{d_{{a}}}^{3}{d_{{i}}}^{2}
\delta_{{1}}{\delta_{{2}}}^{2}{\rho_{{1}}}^{2}-4\,{d_{{a}}}^{3}d_{{i}}
d_{{r}}{\delta_{{2}}}^{3}\rho_{{1}}\rho_{{3}}-4\,{d_{{a}}}^{3}{d_{{r}}
}^{2}\delta_{{1}}{\delta_{{2}}}^{2}{\rho_{{1}}}^{2}-4\,{d_{{a}}}^{3}{d
	_{{r}}}^{2}{\delta_{{2}}}^{3}\rho_{{1}}\rho_{{3}}+2\,{d_{{a}}}^{2}{d_{
		{i}}}^{3}{\delta_{{1}}}^{2}\delta_{{2}}{\rho_{{1}}}^{2}-12\,{d_{{a}}}^
{2}{d_{{i}}}^{2}d_{{r}}\delta_{{1}}{\delta_{{2}}}^{2}\rho_{{1}}\rho_{{
		3}}-8\,{d_{{a}}}^{2}d_{{i}}{d_{{r}}}^{2}{\delta_{{1}}}^{2}\delta_{{2}}
{\rho_{{1}}}^{2}-6\,{d_{{a}}}^{2}d_{{i}}{d_{{r}}}^{2}\delta_{{1}}{
	\delta_{{2}}}^{2}{\rho_{{1}}}^{2}-36\,{d_{{a}}}^{2}d_{{i}}{d_{{r}}}^{2
}\delta_{{1}}{\delta_{{2}}}^{2}\rho_{{1}}\rho_{{3}}-4\,{d_{{a}}}^{2}d_
{{i}}{d_{{r}}}^{2}{\delta_{{2}}}^{3}\rho_{{1}}\rho_{{3}}-32\,{d_{{a}}}
^{2}d_{{i}}{d_{{r}}}^{2}{\delta_{{2}}}^{3}{\rho_{{3}}}^{2}-2\,{d_{{a}}
}^{2}{d_{{r}}}^{3}\delta_{{1}}{\delta_{{2}}}^{2}\rho_{{1}}\rho_{{3}}+4
\,{d_{{a}}}^{2}{d_{{r}}}^{3}{\delta_{{2}}}^{3}\rho_{{1}}\rho_{{3}}-4\,
d_{{a}}{d_{{i}}}^{3}d_{{r}}{\delta_{{1}}}^{2}\delta_{{2}}\rho_{{1}}
\rho_{{3}}-4\,d_{{a}}{d_{{i}}}^{2}{d_{{r}}}^{2}{\delta_{{1}}}^{3}{\rho
	_{{1}}}^{2}-12\,d_{{a}}{d_{{i}}}^{2}{d_{{r}}}^{2}{\delta_{{1}}}^{2}
\delta_{{2}}{\rho_{{1}}}^{2}-36\,d_{{a}}{d_{{i}}}^{2}{d_{{r}}}^{2}{
	\delta_{{1}}}^{2}\delta_{{2}}\rho_{{1}}\rho_{{3}}+2\,d_{{a}}{d_{{i}}}^
{2}{d_{{r}}}^{2}\delta_{{1}}{\delta_{{2}}}^{2}{\rho_{{1}}}^{2}-24\,d_{
	{a}}{d_{{i}}}^{2}{d_{{r}}}^{2}\delta_{{1}}{\delta_{{2}}}^{2}\rho_{{1}}
\rho_{{3}}-96\,d_{{a}}{d_{{i}}}^{2}{d_{{r}}}^{2}\delta_{{1}}{\delta_{{
			2}}}^{2}{\rho_{{3}}}^{2}-8\,d_{{a}}d_{{i}}{d_{{r}}}^{3}{\delta_{{1}}}^
{2}\delta_{{2}}\rho_{{1}}\rho_{{3}}+20\,d_{{a}}d_{{i}}{d_{{r}}}^{3}
\delta_{{1}}{\delta_{{2}}}^{2}\rho_{{1}}\rho_{{3}}-4\,d_{{a}}d_{{i}}{d
	_{{r}}}^{3}{\delta_{{2}}}^{3}\rho_{{1}}\rho_{{3}}-10\,d_{{a}}{d_{{r}}}
^{4}\delta_{{1}}{\delta_{{2}}}^{2}{\rho_{{3}}}^{2}-2\,{d_{{i}}}^{3}{d_
	{{r}}}^{2}{\delta_{{1}}}^{3}{\rho_{{1}}}^{2}-4\,{d_{{i}}}^{3}{d_{{r}}}
^{2}{\delta_{{1}}}^{3}\rho_{{1}}\rho_{{3}}-2\,{d_{{i}}}^{3}{d_{{r}}}^{
	2}{\delta_{{1}}}^{2}\delta_{{2}}{\rho_{{1}}}^{2}-12\,{d_{{i}}}^{3}{d_{
		{r}}}^{2}{\delta_{{1}}}^{2}\delta_{{2}}\rho_{{1}}\rho_{{3}}-32\,{d_{{i
}}}^{3}{d_{{r}}}^{2}{\delta_{{1}}}^{2}\delta_{{2}}{\rho_{{3}}}^{2}-2\,
{d_{{i}}}^{2}{d_{{r}}}^{3}{\delta_{{1}}}^{3}\rho_{{1}}\rho_{{3}}+8\,{d
	_{{i}}}^{2}{d_{{r}}}^{3}{\delta_{{1}}}^{2}\delta_{{2}}\rho_{{1}}\rho_{
	{3}}+12\,{d_{{i}}}^{2}{d_{{r}}}^{3}\delta_{{1}}{\delta_{{2}}}^{2}\rho_
{{1}}\rho_{{3}}-10\,d_{{i}}{d_{{r}}}^{4}{\delta_{{1}}}^{2}\delta_{{2}}
{\rho_{{3}}}^{2}-16\,d_{{i}}{d_{{r}}}^{4}{\delta_{{2}}}^{3}{\rho_{{3}}
}^{2} \Big) {I^{\ast}}^{2}+8\,{d_{{r}}}^{2} \Big( {d_{{a}}}^{4}d_{{i}}{
	\delta_{{2}}}^{3}+6\,{d_{{a}}}^{3}{d_{{i}}}^{2}\delta_{{1}}{\delta_{{2
}}}^{2}+{d_{{a}}}^{3}{d_{{r}}}^{2}\delta_{{1}}{\delta_{{2}}}^{2}+6\,{d
	_{{i}}}^{3}{\delta_{{1}}}^{2}\delta_{{2}}{d_{{a}}}^{2}+3\,{d_{{a}}}^{2
}d_{{i}}{d_{{r}}}^{2}{\delta_{{1}}}^{2}\delta_{{2}}+{d_{{a}}}^{2}d_{{i
}}{d_{{r}}}^{2}{\delta_{{2}}}^{3}+d_{{a}}{d_{{i}}}^{4}{\delta_{{1}}}^{
	3}+d_{{a}}{d_{{i}}}^{2}{\delta_{{1}}}^{3}{d_{{r}}}^{2}+3\,d_{{a}}{d_{{
			i}}}^{2}\delta_{{1}}{d_{{r}}}^{2}{\delta_{{2}}}^{2}+{d_{{i}}}^{3}{
	\delta_{{1}}}^{2}\delta_{{2}}{d_{{r}}}^{2} \Big)  \Big) {{T_{reg}^{\ast}}}
^{3}+ \Big( 2\,{\delta_{{2}}}^{2}{\rho_{{1}}}^{4}{\rho_{{3}}}^{4}{I^{\ast}}^
{8}-{\rho_{{1}}}^{2}{\rho_{{3}}}^{2} \Big( {d_{{a}}}^{2}{\delta_{{2}}
}^{2}{\rho_{{1}}}^{2}+6\,d_{{a}}d_{{i}}\delta_{{1}}\delta_{{2}}{\rho_{
		{1}}}^{2}+18\,d_{{a}}d_{{i}}{\delta_{{2}}}^{2}\rho_{{1}}\rho_{{3}}+2\,
d_{{a}}d_{{r}}{\delta_{{2}}}^{2}\rho_{{1}}\rho_{{3}}+4\,{d_{{i}}}^{2}{
	\delta_{{1}}}^{2}{\rho_{{1}}}^{2}+16\,{d_{{i}}}^{2}\delta_{{1}}\delta_
{{2}}{\rho_{{1}}}^{2}+40\,{d_{{i}}}^{2}\delta_{{1}}\delta_{{2}}\rho_{{
		1}}\rho_{{3}}+{d_{{i}}}^{2}{\delta_{{2}}}^{2}{\rho_{{1}}}^{2}+20\,{d_{
		{i}}}^{2}{\delta_{{2}}}^{2}\rho_{{1}}\rho_{{3}}+52\,{d_{{i}}}^{2}{
	\delta_{{2}}}^{2}{\rho_{{3}}}^{2}+4\,d_{{i}}d_{{r}}\delta_{{1}}\delta_
{{2}}\rho_{{1}}\rho_{{3}}-18\,d_{{i}}d_{{r}}{\delta_{{2}}}^{2}\rho_{{1
}}\rho_{{3}}+5\,{d_{{r}}}^{2}{\delta_{{2}}}^{2}{\rho_{{3}}}^{2}
\Big) {I^{\ast}}^{6}+2\,\rho_{{1}}\rho_{{3}} \Big( 2\,{d_{{a}}}^{3}d_{{i}
}{\delta_{{2}}}^{2}{\rho_{{1}}}^{2}-{d_{{a}}}^{3}d_{{r}}{\delta_{{2}}}
^{2}{\rho_{{1}}}^{2}+12\,{d_{{a}}}^{2}{d_{{i}}}^{2}\delta_{{1}}\delta_
{{2}}{\rho_{{1}}}^{2}+{d_{{a}}}^{2}{d_{{i}}}^{2}{\delta_{{2}}}^{2}{
	\rho_{{1}}}^{2}+3\,{d_{{a}}}^{2}{d_{{i}}}^{2}{\delta_{{2}}}^{2}\rho_{{
		1}}\rho_{{3}}-5\,{d_{{a}}}^{2}d_{{i}}d_{{r}}\delta_{{1}}\delta_{{2}}{
	\rho_{{1}}}^{2}-2\,{d_{{a}}}^{2}d_{{i}}d_{{r}}{\delta_{{2}}}^{2}{\rho_
	{{1}}}^{2}+4\,{d_{{a}}}^{2}{d_{{r}}}^{2}{\delta_{{2}}}^{2}\rho_{{1}}
\rho_{{3}}+4\,d_{{a}}{d_{{i}}}^{3}{\delta_{{1}}}^{2}{\rho_{{1}}}^{2}+6
\,d_{{a}}{d_{{i}}}^{3}\delta_{{1}}\delta_{{2}}{\rho_{{1}}}^{2}-12\,d_{
	{a}}{d_{{i}}}^{3}\delta_{{1}}\delta_{{2}}\rho_{{1}}\rho_{{3}}-8\,d_{{a
}}{d_{{i}}}^{2}d_{{r}}\delta_{{1}}\delta_{{2}}{\rho_{{1}}}^{2}+d_{{a}}
{d_{{i}}}^{2}d_{{r}}{\delta_{{2}}}^{2}{\rho_{{1}}}^{2}+36\,d_{{a}}{d_{
		{i}}}^{2}d_{{r}}{\delta_{{2}}}^{2}{\rho_{{3}}}^{2}+14\,d_{{a}}d_{{i}}{
	d_{{r}}}^{2}\delta_{{1}}\delta_{{2}}\rho_{{1}}\rho_{{3}}+12\,d_{{a}}d_
{{i}}{d_{{r}}}^{2}{\delta_{{2}}}^{2}{\rho_{{3}}}^{2}+3\,d_{{a}}{d_{{r}
}}^{3}{\delta_{{2}}}^{2}{\rho_{{3}}}^{2}-4\,{d_{{i}}}^{4}{\delta_{{1}}
}^{2}\rho_{{1}}\rho_{{3}}-4\,{d_{{i}}}^{3}d_{{r}}{\delta_{{1}}}^{2}{
	\rho_{{1}}}^{2}-7\,{d_{{i}}}^{3}d_{{r}}\delta_{{1}}\delta_{{2}}{\rho_{
		{1}}}^{2}+28\,{d_{{i}}}^{3}d_{{r}}\delta_{{1}}\delta_{{2}}{\rho_{{3}}}
^{2}+3\,{d_{{i}}}^{2}{d_{{r}}}^{2}{\delta_{{1}}}^{2}\rho_{{1}}\rho_{{3
}}+22\,{d_{{i}}}^{2}{d_{{r}}}^{2}\delta_{{1}}\delta_{{2}}\rho_{{1}}
\rho_{{3}}+12\,{d_{{i}}}^{2}{d_{{r}}}^{2}\delta_{{1}}\delta_{{2}}{\rho
	_{{3}}}^{2}+3\,{d_{{i}}}^{2}{d_{{r}}}^{2}{\delta_{{2}}}^{2}\rho_{{1}}
\rho_{{3}}+36\,{d_{{i}}}^{2}{d_{{r}}}^{2}{\delta_{{2}}}^{2}{\rho_{{3}}
}^{2}+9\,d_{{i}}{d_{{r}}}^{3}\delta_{{1}}\delta_{{2}}{\rho_{{3}}}^{2}-
12\,d_{{i}}{d_{{r}}}^{3}{\delta_{{2}}}^{2}{\rho_{{3}}}^{2} \Big) {I^{\ast}}
^{4}+ \Big( -{d_{{a}}}^{4}{d_{{i}}}^{2}{\delta_{{2}}}^{2}{\rho_{{1}}}
^{2}-{d_{{a}}}^{4}{d_{{r}}}^{2}{\delta_{{2}}}^{2}{\rho_{{1}}}^{2}-2\,{
	d_{{a}}}^{3}{d_{{i}}}^{3}\delta_{{1}}\delta_{{2}}{\rho_{{1}}}^{2}-12\,
{d_{{a}}}^{3}{d_{{i}}}^{2}d_{{r}}{\delta_{{2}}}^{2}\rho_{{1}}\rho_{{3}
}-6\,{d_{{a}}}^{3}d_{{i}}{d_{{r}}}^{2}\delta_{{1}}\delta_{{2}}{\rho_{{
			1}}}^{2}-10\,{d_{{a}}}^{3}d_{{i}}{d_{{r}}}^{2}{\delta_{{2}}}^{2}\rho_{
	{1}}\rho_{{3}}-2\,{d_{{a}}}^{3}{d_{{r}}}^{3}{\delta_{{2}}}^{2}\rho_{{1
}}\rho_{{3}}-{d_{{a}}}^{2}{d_{{i}}}^{4}{\delta_{{1}}}^{2}{\rho_{{1}}}^
{2}-20\,{d_{{a}}}^{2}{d_{{i}}}^{3}d_{{r}}\delta_{{1}}\delta_{{2}}\rho_
{{1}}\rho_{{3}}-7\,{d_{{a}}}^{2}{d_{{i}}}^{2}{d_{{r}}}^{2}{\delta_{{1}
}}^{2}{\rho_{{1}}}^{2}-6\,{d_{{a}}}^{2}{d_{{i}}}^{2}{d_{{r}}}^{2}
\delta_{{1}}\delta_{{2}}{\rho_{{1}}}^{2}-30\,{d_{{a}}}^{2}{d_{{i}}}^{2
}{d_{{r}}}^{2}\delta_{{1}}\delta_{{2}}\rho_{{1}}\rho_{{3}}-2\,{d_{{a}}
}^{2}{d_{{i}}}^{2}{d_{{r}}}^{2}{\delta_{{2}}}^{2}{\rho_{{1}}}^{2}-12\,
{d_{{a}}}^{2}{d_{{i}}}^{2}{d_{{r}}}^{2}{\delta_{{2}}}^{2}\rho_{{1}}
\rho_{{3}}-48\,{d_{{a}}}^{2}{d_{{i}}}^{2}{d_{{r}}}^{2}{\delta_{{2}}}^{
	2}{\rho_{{3}}}^{2}-16\,{d_{{a}}}^{2}d_{{i}}{d_{{r}}}^{3}\delta_{{1}}
\delta_{{2}}\rho_{{1}}\rho_{{3}}+10\,{d_{{a}}}^{2}d_{{i}}{d_{{r}}}^{3}
{\delta_{{2}}}^{2}\rho_{{1}}\rho_{{3}}-5\,{d_{{a}}}^{2}{d_{{r}}}^{4}{
	\delta_{{2}}}^{2}{\rho_{{3}}}^{2}-4\,d_{{a}}{d_{{i}}}^{4}d_{{r}}{
	\delta_{{1}}}^{2}\rho_{{1}}\rho_{{3}}-4\,d_{{a}}{d_{{i}}}^{3}{d_{{r}}}
^{2}{\delta_{{1}}}^{2}{\rho_{{1}}}^{2}-10\,d_{{a}}{d_{{i}}}^{3}{d_{{r}
}}^{2}{\delta_{{1}}}^{2}\rho_{{1}}\rho_{{3}}-2\,d_{{a}}{d_{{i}}}^{3}{d
	_{{r}}}^{2}\delta_{{1}}\delta_{{2}}{\rho_{{1}}}^{2}-24\,d_{{a}}{d_{{i}
}}^{3}{d_{{r}}}^{2}\delta_{{1}}\delta_{{2}}\rho_{{1}}\rho_{{3}}-64\,d_
{{a}}{d_{{i}}}^{3}{d_{{r}}}^{2}\delta_{{1}}\delta_{{2}}{\rho_{{3}}}^{2
}-10\,d_{{a}}{d_{{i}}}^{2}{d_{{r}}}^{3}{\delta_{{1}}}^{2}\rho_{{1}}
\rho_{{3}}+16\,d_{{a}}{d_{{i}}}^{2}{d_{{r}}}^{3}\delta_{{1}}\delta_{{2
}}\rho_{{1}}\rho_{{3}}-12\,d_{{a}}{d_{{i}}}^{2}{d_{{r}}}^{3}{\delta_{{
			2}}}^{2}\rho_{{1}}\rho_{{3}}-20\,d_{{a}}d_{{i}}{d_{{r}}}^{4}\delta_{{1
}}\delta_{{2}}{\rho_{{3}}}^{2}-{d_{{i}}}^{4}{d_{{r}}}^{2}{\delta_{{1}}
}^{2}{\rho_{{1}}}^{2}-4\,{d_{{i}}}^{4}{d_{{r}}}^{2}{\delta_{{1}}}^{2}
\rho_{{1}}\rho_{{3}}-8\,{d_{{i}}}^{4}{d_{{r}}}^{2}{\delta_{{1}}}^{2}{
	\rho_{{3}}}^{2}+2\,{d_{{i}}}^{3}{d_{{r}}}^{3}{\delta_{{1}}}^{2}\rho_{{
		1}}\rho_{{3}}+4\,{d_{{i}}}^{3}{d_{{r}}}^{3}\delta_{{1}}\delta_{{2}}
\rho_{{1}}\rho_{{3}}-5\,{d_{{i}}}^{2}{d_{{r}}}^{4}{\delta_{{1}}}^{2}{
	\rho_{{3}}}^{2}-24\,{d_{{i}}}^{2}{d_{{r}}}^{4}{\delta_{{2}}}^{2}{\rho_
	{{3}}}^{2} \Big) {I^{\ast}}^{2}+2\,{d_{{r}}}^{2} \Big( 6\,{d_{{a}}}^{4}{d_
	{{i}}}^{2}{\delta_{{2}}}^{2}+{d_{{a}}}^{4}{d_{{r}}}^{2}{\delta_{{2}}}^
{2}+16\,{d_{{a}}}^{3}{d_{{i}}}^{3}\delta_{{1}}\delta_{{2}}+8\,{d_{{a}}
}^{3}d_{{i}}{d_{{r}}}^{2}\delta_{{1}}\delta_{{2}}+6\,{d_{{a}}}^{2}{d_{
		{i}}}^{4}{\delta_{{1}}}^{2}+6\,{d_{{a}}}^{2}{d_{{i}}}^{2}{d_{{r}}}^{2}
{\delta_{{1}}}^{2}+6\,{d_{{a}}}^{2}{d_{{i}}}^{2}{d_{{r}}}^{2}{\delta_{
		{2}}}^{2}+8\,d_{{a}}{d_{{i}}}^{3}{d_{{r}}}^{2}\delta_{{1}}\delta_{{2}}
+{d_{{i}}}^{4}{d_{{r}}}^{2}{\delta_{{1}}}^{2} \Big)  \Big) {{T_{reg}^{\ast}}}^{2}+ \Big( 8\,d_{{i}}\delta_{{2}}{\rho_{{1}}}^{4}{\rho_{{3}}}^{4
}{I^{\ast}}^{8}-4\,d_{{i}}{\rho_{{1}}}^{2}{\rho_{{3}}}^{2} \Big( {d_{{a}}}^{
	2}\delta_{{2}}{\rho_{{1}}}^{2}+2\,d_{{a}}d_{{i}}\delta_{{1}}{\rho_{{1}
}}^{2}+6\,d_{{a}}d_{{i}}\delta_{{2}}\rho_{{1}}\rho_{{3}}+3\,d_{{a}}d_{
	{r}}\delta_{{2}}\rho_{{1}}\rho_{{3}}+2\,{d_{{i}}}^{2}\delta_{{1}}{\rho
	_{{1}}}^{2}+4\,{d_{{i}}}^{2}\delta_{{1}}\rho_{{1}}\rho_{{3}}+{d_{{i}}}
^{2}\delta_{{2}}{\rho_{{1}}}^{2}+8\,{d_{{i}}}^{2}\delta_{{2}}\rho_{{1}
}\rho_{{3}}+12\,{d_{{i}}}^{2}\delta_{{2}}{\rho_{{3}}}^{2}+2\,d_{{i}}d_
{{r}}\delta_{{1}}\rho_{{1}}\rho_{{3}}-6\,d_{{i}}d_{{r}}\delta_{{2}}
\rho_{{1}}\rho_{{3}}+3\,{d_{{r}}}^{2}\delta_{{2}}{\rho_{{3}}}^{2}
\Big) {I^{\ast}}^{6}+2\,d_{{i}}\rho_{{1}}\rho_{{3}} \Big( 4\,{d_{{a}}}^{3
}d_{{i}}\delta_{{2}}{\rho_{{1}}}^{2}+{d_{{a}}}^{3}d_{{r}}\delta_{{2}}{
	\rho_{{1}}}^{2}+5\,{d_{{a}}}^{2}{d_{{i}}}^{2}\delta_{{1}}{\rho_{{1}}}^
{2}+3\,{d_{{a}}}^{2}{d_{{i}}}^{2}\delta_{{2}}{\rho_{{1}}}^{2}+8\,{d_{{
			a}}}^{2}{d_{{i}}}^{2}\delta_{{2}}\rho_{{1}}\rho_{{3}}+5\,{d_{{a}}}^{2}
d_{{i}}d_{{r}}\delta_{{1}}{\rho_{{1}}}^{2}-4\,{d_{{a}}}^{2}d_{{i}}d_{{
		r}}\delta_{{2}}{\rho_{{1}}}^{2}+8\,{d_{{a}}}^{2}{d_{{r}}}^{2}\delta_{{
		2}}\rho_{{1}}\rho_{{3}}+4\,d_{{a}}{d_{{i}}}^{3}\delta_{{1}}{\rho_{{1}}
}^{2}-2\,d_{{a}}{d_{{i}}}^{2}d_{{r}}\delta_{{1}}{\rho_{{1}}}^{2}+3\,d_
{{a}}{d_{{i}}}^{2}d_{{r}}\delta_{{2}}{\rho_{{1}}}^{2}+28\,d_{{a}}{d_{{
			i}}}^{2}d_{{r}}\delta_{{2}}{\rho_{{3}}}^{2}+8\,d_{{a}}d_{{i}}{d_{{r}}}
^{2}\delta_{{1}}\rho_{{1}}\rho_{{3}}+12\,d_{{a}}d_{{i}}{d_{{r}}}^{2}
\delta_{{2}}{\rho_{{3}}}^{2}+9\,d_{{a}}{d_{{r}}}^{3}\delta_{{2}}{\rho_
	{{3}}}^{2}-2\,{d_{{i}}}^{3}d_{{r}}\delta_{{1}}{\rho_{{1}}}^{2}+8\,{d_{
		{i}}}^{3}d_{{r}}\delta_{{1}}{\rho_{{3}}}^{2}+8\,{d_{{i}}}^{2}{d_{{r}}}
^{2}\delta_{{1}}\rho_{{1}}\rho_{{3}}+4\,{d_{{i}}}^{2}{d_{{r}}}^{2}
\delta_{{1}}{\rho_{{3}}}^{2}+8\,{d_{{i}}}^{2}{d_{{r}}}^{2}\delta_{{2}}
\rho_{{1}}\rho_{{3}}+28\,{d_{{i}}}^{2}{d_{{r}}}^{2}\delta_{{2}}{\rho_{
		{3}}}^{2}+6\,d_{{i}}{d_{{r}}}^{3}\delta_{{1}}{\rho_{{3}}}^{2}-12\,d_{{
		i}}{d_{{r}}}^{3}\delta_{{2}}{\rho_{{3}}}^{2} \Big) {I^{\ast}}^{4}-2\,d_{{i}
} \Big( {d_{{a}}}^{4}{d_{{i}}}^{2}\delta_{{2}}{\rho_{{1}}}^{2}+{d_{{a
}}}^{4}{d_{{r}}}^{2}\delta_{{2}}{\rho_{{1}}}^{2}+{d_{{a}}}^{3}{d_{{i}}
}^{3}\delta_{{1}}{\rho_{{1}}}^{2}+6\,{d_{{a}}}^{3}{d_{{i}}}^{2}d_{{r}}
\delta_{{2}}\rho_{{1}}\rho_{{3}}+3\,{d_{{a}}}^{3}d_{{i}}{d_{{r}}}^{2}
\delta_{{1}}{\rho_{{1}}}^{2}+4\,{d_{{a}}}^{3}d_{{i}}{d_{{r}}}^{2}
\delta_{{2}}\rho_{{1}}\rho_{{3}}+4\,{d_{{a}}}^{3}{d_{{r}}}^{3}\delta_{
	{2}}\rho_{{1}}\rho_{{3}}+4\,{d_{{a}}}^{2}{d_{{i}}}^{3}d_{{r}}\delta_{{
		1}}\rho_{{1}}\rho_{{3}}+{d_{{a}}}^{2}{d_{{i}}}^{2}{d_{{r}}}^{2}\delta_
{{1}}{\rho_{{1}}}^{2}+4\,{d_{{a}}}^{2}{d_{{i}}}^{2}{d_{{r}}}^{2}\delta
_{{1}}\rho_{{1}}\rho_{{3}}+2\,{d_{{a}}}^{2}{d_{{i}}}^{2}{d_{{r}}}^{2}
\delta_{{2}}{\rho_{{1}}}^{2}+6\,{d_{{a}}}^{2}{d_{{i}}}^{2}{d_{{r}}}^{2
}\delta_{{2}}\rho_{{1}}\rho_{{3}}+16\,{d_{{a}}}^{2}{d_{{i}}}^{2}{d_{{r
}}}^{2}\delta_{{2}}{\rho_{{3}}}^{2}+7\,{d_{{a}}}^{2}d_{{i}}{d_{{r}}}^{
	3}\delta_{{1}}\rho_{{1}}\rho_{{3}}-4\,{d_{{a}}}^{2}d_{{i}}{d_{{r}}}^{3
}\delta_{{2}}\rho_{{1}}\rho_{{3}}+5\,{d_{{a}}}^{2}{d_{{r}}}^{4}\delta_
{{2}}{\rho_{{3}}}^{2}+d_{{a}}{d_{{i}}}^{3}{d_{{r}}}^{2}\delta_{{1}}{
	\rho_{{1}}}^{2}+4\,d_{{a}}{d_{{i}}}^{3}{d_{{r}}}^{2}\delta_{{1}}\rho_{
	{1}}\rho_{{3}}+8\,d_{{a}}{d_{{i}}}^{3}{d_{{r}}}^{2}\delta_{{1}}{\rho_{
		{3}}}^{2}-2\,d_{{a}}{d_{{i}}}^{2}{d_{{r}}}^{3}\delta_{{1}}\rho_{{1}}
\rho_{{3}}+6\,d_{{a}}{d_{{i}}}^{2}{d_{{r}}}^{3}\delta_{{2}}\rho_{{1}}
\rho_{{3}}+5\,d_{{a}}d_{{i}}{d_{{r}}}^{4}\delta_{{1}}{\rho_{{3}}}^{2}+
8\,{d_{{i}}}^{2}{d_{{r}}}^{4}\delta_{{2}}{\rho_{{3}}}^{2} \Big) {I^{\ast}}^
{2}+4\,d_{{a}}d_{{i}}{d_{{r}}}^{2} \Big( 2\,{d_{{a}}}^{3}{d_{{i}}}^{2
}\delta_{{2}}+{d_{{a}}}^{3}{d_{{r}}}^{2}\delta_{{2}}+2\,{d_{{a}}}^{2}{
	d_{{i}}}^{3}\delta_{{1}}+2\,{d_{{a}}}^{2}d_{{i}}{d_{{r}}}^{2}\delta_{{
		1}}+2\,d_{{a}}{d_{{i}}}^{2}{d_{{r}}}^{2}\delta_{{2}}+{d_{{i}}}^{3}{d_{
		{r}}}^{2}\delta_{{1}} \Big)  \Big) {T_{reg}^{\ast}}+8\,{d_{{i}}}^{2}{\rho_
	{{1}}}^{4}{\rho_{{3}}}^{4}{I^{\ast}}^{8}-{d_{{i}}}^{2}{\rho_{{1}}}^{2}{\rho_{
		{3}}}^{2} \Big( 5\,{d_{{a}}}^{2}{\rho_{{1}}}^{2}+8\,d_{{a}}d_{{i}}
\rho_{{1}}\rho_{{3}}+12\,d_{{a}}d_{{r}}\rho_{{1}}\rho_{{3}}+4\,{d_{{i}
}}^{2}{\rho_{{1}}}^{2}+16\,{d_{{i}}}^{2}\rho_{{1}}\rho_{{3}}+16\,{d_{{
			i}}}^{2}{\rho_{{3}}}^{2}-8\,d_{{i}}d_{{r}}\rho_{{1}}\rho_{{3}}+8\,{d_{
		{r}}}^{2}{\rho_{{3}}}^{2} \Big) {I^{\ast}}^{6}+2\,{d_{{i}}}^{2}\rho_{{1}}
\rho_{{3}} \Big( {d_{{a}}}^{3}d_{{i}}{\rho_{{1}}}^{2}+3\,{d_{{a}}}^{3
}d_{{r}}{\rho_{{1}}}^{2}+2\,{d_{{a}}}^{2}{d_{{i}}}^{2}{\rho_{{1}}}^{2}
+4\,{d_{{a}}}^{2}{d_{{i}}}^{2}\rho_{{1}}\rho_{{3}}-{d_{{a}}}^{2}d_{{i}
}d_{{r}}{\rho_{{1}}}^{2}+5\,{d_{{a}}}^{2}{d_{{r}}}^{2}\rho_{{1}}\rho_{
	{3}}+2\,d_{{a}}{d_{{i}}}^{2}d_{{r}}{\rho_{{1}}}^{2}+8\,d_{{a}}{d_{{i}}
}^{2}d_{{r}}{\rho_{{3}}}^{2}+4\,d_{{a}}d_{{i}}{d_{{r}}}^{2}{\rho_{{3}}
}^{2}+6\,d_{{a}}{d_{{r}}}^{3}{\rho_{{3}}}^{2}+4\,{d_{{i}}}^{2}{d_{{r}}
}^{2}\rho_{{1}}\rho_{{3}}+8\,{d_{{i}}}^{2}{d_{{r}}}^{2}{\rho_{{3}}}^{2
}-4\,d_{{i}}{d_{{r}}}^{3}{\rho_{{3}}}^{2} \Big) {I^{\ast}}^{4}-{d_{{i}}}^{2
} \Big( {d_{{a}}}^{4}{d_{{i}}}^{2}{\rho_{{1}}}^{2}+2\,{d_{{a}}}^{4}{d
	_{{r}}}^{2}{\rho_{{1}}}^{2}+4\,{d_{{a}}}^{3}{d_{{i}}}^{2}d_{{r}}\rho_{
	{1}}\rho_{{3}}+2\,{d_{{a}}}^{3}d_{{i}}{d_{{r}}}^{2}\rho_{{1}}\rho_{{3}
}+6\,{d_{{a}}}^{3}{d_{{r}}}^{3}\rho_{{1}}\rho_{{3}}+2\,{d_{{a}}}^{2}{d
	_{{i}}}^{2}{d_{{r}}}^{2}{\rho_{{1}}}^{2}+4\,{d_{{a}}}^{2}{d_{{i}}}^{2}
{d_{{r}}}^{2}\rho_{{1}}\rho_{{3}}+8\,{d_{{a}}}^{2}{d_{{i}}}^{2}{d_{{r}
}}^{2}{\rho_{{3}}}^{2}-2\,{d_{{a}}}^{2}d_{{i}}{d_{{r}}}^{3}\rho_{{1}}
\rho_{{3}}+5\,{d_{{a}}}^{2}{d_{{r}}}^{4}{\rho_{{3}}}^{2}+4\,d_{{a}}{d_
	{{i}}}^{2}{d_{{r}}}^{3}\rho_{{1}}\rho_{{3}}+4\,{d_{{i}}}^{2}{d_{{r}}}^
{4}{\rho_{{3}}}^{2} \Big) {I^{\ast}}^{2}+2\,{d_{{a}}}^{2}{d_{{i}}}^{2}{d_{{
			r}}}^{2} \Big( {d_{{a}}}^{2}{d_{{i}}}^{2}+{d_{{a}}}^{2}{d_{{r}}}^{2}+
{d_{{i}}}^{2}{d_{{r}}}^{2} \Big)$.\vspace{6pt}

$b_4={{T_{reg}^{\ast}}}^{8}{\delta_{{1}}}^{4}{\delta_{{2}}}^{4}+4\,{\delta_{{1}}}^{
	3}{\delta_{{2}}}^{3} \left( d_{{a}}\delta_{{2}}+d_{{i}}\delta_{{1}}
\right) {{T_{reg}^{\ast}}}^{7}+ \Big( -{\delta_{{1}}}^{2}{\delta_{{2}}}^{2}
\Big( 2\,{\delta_{{1}}}^{2}{\rho_{{1}}}^{2}+2\,\delta_{{1}}\delta_{{
		2}}{\rho_{{1}}}^{2}+6\,\delta_{{1}}\delta_{{2}}\rho_{{1}}\rho_{{3}}+{
	\delta_{{2}}}^{2}{\rho_{{1}}}^{2}+4\,{\delta_{{2}}}^{2}{\rho_{{3}}}^{2
} \Big) {I^{\ast}}^{2}+2\,{\delta_{{1}}}^{2}{\delta_{{2}}}^{2} \Big( 3\,{d
	_{{a}}}^{2}{\delta_{{2}}}^{2}+8\,d_{{a}}d_{{i}}\delta_{{1}}\delta_{{2}
}+3\,{d_{{i}}}^{2}{\delta_{{1}}}^{2}+2\,{d_{{r}}}^{2}{\delta_{{1}}}^{2
}+2\,{d_{{r}}}^{2}{\delta_{{2}}}^{2} \Big)  \Big) {{T_{reg}^{\ast}}}^{6}+
\Big( -2\,\delta_{{1}}\delta_{{2}} \Big( 3\,d_{{a}}{\delta_{{1}}}^{
	2}\delta_{{2}}{\rho_{{1}}}^{2}+2\,d_{{a}}\delta_{{1}}{\delta_{{2}}}^{2
}{\rho_{{1}}}^{2}+8\,d_{{a}}\delta_{{1}}{\delta_{{2}}}^{2}\rho_{{1}}
\rho_{{3}}+4\,d_{{a}}{\delta_{{2}}}^{3}{\rho_{{3}}}^{2}+d_{{i}}{\delta
	_{{1}}}^{3}{\rho_{{1}}}^{2}+3\,d_{{i}}{\delta_{{1}}}^{2}\delta_{{2}}{
	\rho_{{1}}}^{2}+8\,d_{{i}}{\delta_{{1}}}^{2}\delta_{{2}}\rho_{{1}}\rho
_{{3}}+d_{{i}}\delta_{{1}}{\delta_{{2}}}^{2}{\rho_{{1}}}^{2}+2\,d_{{i}
}\delta_{{1}}{\delta_{{2}}}^{2}\rho_{{1}}\rho_{{3}}+8\,d_{{i}}\delta_{
	{1}}{\delta_{{2}}}^{2}{\rho_{{3}}}^{2}-d_{{r}}{\delta_{{1}}}^{2}\delta
_{{2}}\rho_{{1}}\rho_{{3}}-2\,d_{{r}}\delta_{{1}}{\delta_{{2}}}^{2}
\rho_{{1}}\rho_{{3}}-2\,d_{{r}}{\delta_{{2}}}^{3}\rho_{{1}}\rho_{{3}}
\Big) {I^{\ast}}^{2}+4\,\delta_{{1}}\delta_{{2}} \Big( {d_{{a}}}^{3}{
	\delta_{{2}}}^{3}+6\,{d_{{a}}}^{2}d_{{i}}\delta_{{1}}{\delta_{{2}}}^{2
}+6\,d_{{a}}{d_{{i}}}^{2}{\delta_{{1}}}^{2}\delta_{{2}}+4\,d_{{a}}{d_{
		{r}}}^{2}{\delta_{{1}}}^{2}\delta_{{2}}+2\,d_{{a}}{d_{{r}}}^{2}{\delta
	_{{2}}}^{3}+{d_{{i}}}^{3}{\delta_{{1}}}^{3}+2\,d_{{i}}{d_{{r}}}^{2}{
	\delta_{{1}}}^{3}+4\,d_{{i}}{d_{{r}}}^{2}\delta_{{1}}{\delta_{{2}}}^{2
} \Big)  \Big) {{T_{reg}^{\ast}}}^{5}+ \Big( 2\,\rho_{{3}}\delta_{{2}}
\rho_{{1}} \Big( 4\,{\delta_{{1}}}^{3}{\rho_{{1}}}^{2}+4\,{\delta_{{1
}}}^{2}\delta_{{2}}\rho_{{1}}\rho_{{3}}-2\,\delta_{{1}}{\delta_{{2}}}^
{2}{\rho_{{1}}}^{2}+6\,\delta_{{1}}{\delta_{{2}}}^{2}\rho_{{1}}\rho_{{
		3}}+4\,\delta_{{1}}{\delta_{{2}}}^{2}{\rho_{{3}}}^{2}-{\delta_{{2}}}^{
	3}\rho_{{1}}\rho_{{3}}+4\,{\delta_{{2}}}^{3}{\rho_{{3}}}^{2} \Big) {
	Y}^{4}+ \Big( -7\,{d_{{a}}}^{2}{\delta_{{1}}}^{2}{\delta_{{2}}}^{2}{
	\rho_{{1}}}^{2}-2\,{d_{{a}}}^{2}\delta_{{1}}{\delta_{{2}}}^{3}{\rho_{{
			1}}}^{2}-14\,{d_{{a}}}^{2}\delta_{{1}}{\delta_{{2}}}^{3}\rho_{{1}}\rho
_{{3}}-4\,{d_{{a}}}^{2}{\delta_{{2}}}^{4}{\rho_{{3}}}^{2}-6\,d_{{a}}d_
{{i}}{\delta_{{1}}}^{3}\delta_{{2}}{\rho_{{1}}}^{2}-12\,d_{{a}}d_{{i}}
{\delta_{{1}}}^{2}{\delta_{{2}}}^{2}{\rho_{{1}}}^{2}-42\,d_{{a}}d_{{i}
}{\delta_{{1}}}^{2}{\delta_{{2}}}^{2}\rho_{{1}}\rho_{{3}}+2\,d_{{a}}d_
{{i}}\delta_{{1}}{\delta_{{2}}}^{3}{\rho_{{1}}}^{2}-8\,d_{{a}}d_{{i}}
\delta_{{1}}{\delta_{{2}}}^{3}\rho_{{1}}\rho_{{3}}-32\,d_{{a}}d_{{i}}
\delta_{{1}}{\delta_{{2}}}^{3}{\rho_{{3}}}^{2}+2\,d_{{a}}d_{{r}}{
	\delta_{{1}}}^{2}{\delta_{{2}}}^{2}\rho_{{1}}\rho_{{3}}+8\,d_{{a}}d_{{
		r}}\delta_{{1}}{\delta_{{2}}}^{3}\rho_{{1}}\rho_{{3}}-{d_{{i}}}^{2}{
	\delta_{{1}}}^{4}{\rho_{{1}}}^{2}-6\,{d_{{i}}}^{2}{\delta_{{1}}}^{3}
\delta_{{2}}{\rho_{{1}}}^{2}-14\,{d_{{i}}}^{2}{\delta_{{1}}}^{3}\delta
_{{2}}\rho_{{1}}\rho_{{3}}-2\,{d_{{i}}}^{2}{\delta_{{1}}}^{2}{\delta_{
		{2}}}^{2}{\rho_{{1}}}^{2}-12\,{d_{{i}}}^{2}{\delta_{{1}}}^{2}{\delta_{
		{2}}}^{2}\rho_{{1}}\rho_{{3}}-24\,{d_{{i}}}^{2}{\delta_{{1}}}^{2}{
	\delta_{{2}}}^{2}{\rho_{{3}}}^{2}+10\,d_{{i}}d_{{r}}{\delta_{{1}}}^{2}
{\delta_{{2}}}^{2}\rho_{{1}}\rho_{{3}}+12\,d_{{i}}d_{{r}}\delta_{{1}}{
	\delta_{{2}}}^{3}\rho_{{1}}\rho_{{3}}-{d_{{r}}}^{2}{\delta_{{1}}}^{4}{
	\rho_{{1}}}^{2}-2\,{d_{{r}}}^{2}{\delta_{{1}}}^{3}\delta_{{2}}{\rho_{{
			1}}}^{2}-6\,{d_{{r}}}^{2}{\delta_{{1}}}^{3}\delta_{{2}}\rho_{{1}}\rho_
{{3}}-3\,{d_{{r}}}^{2}{\delta_{{1}}}^{2}{\delta_{{2}}}^{2}{\rho_{{1}}}
^{2}-2\,{d_{{r}}}^{2}{\delta_{{1}}}^{2}{\delta_{{2}}}^{2}\rho_{{1}}
\rho_{{3}}-10\,{d_{{r}}}^{2}{\delta_{{1}}}^{2}{\delta_{{2}}}^{2}{\rho_
	{{3}}}^{2}-2\,{d_{{r}}}^{2}\delta_{{1}}{\delta_{{2}}}^{3}{\rho_{{1}}}^
{2}-6\,{d_{{r}}}^{2}\delta_{{1}}{\delta_{{2}}}^{3}\rho_{{1}}\rho_{{3}}
-8\,{d_{{r}}}^{2}{\delta_{{2}}}^{4}{\rho_{{3}}}^{2} \Big) {I^{\ast}}^{2}+{d
	_{{a}}}^{4}{\delta_{{2}}}^{4}+16\,{d_{{a}}}^{3}d_{{i}}\delta_{{1}}{
	\delta_{{2}}}^{3}+36\,{d_{{a}}}^{2}{d_{{i}}}^{2}{\delta_{{1}}}^{2}{
	\delta_{{2}}}^{2}+24\,{d_{{a}}}^{2}{d_{{r}}}^{2}{\delta_{{1}}}^{2}{
	\delta_{{2}}}^{2}+4\,{d_{{a}}}^{2}{d_{{r}}}^{2}{\delta_{{2}}}^{4}+16\,
d_{{a}}{d_{{i}}}^{3}{\delta_{{1}}}^{3}\delta_{{2}}+32\,d_{{a}}d_{{i}}{
	d_{{r}}}^{2}{\delta_{{1}}}^{3}\delta_{{2}}+32\,d_{{a}}d_{{i}}{d_{{r}}}
^{2}\delta_{{1}}{\delta_{{2}}}^{3}+{d_{{i}}}^{4}{\delta_{{1}}}^{4}+4\,
{d_{{i}}}^{2}{d_{{r}}}^{2}{\delta_{{1}}}^{4}+24\,{d_{{i}}}^{2}{d_{{r}}
}^{2}{\delta_{{1}}}^{2}{\delta_{{2}}}^{2}+{d_{{r}}}^{4}{\delta_{{1}}}^
{4}+4\,{d_{{r}}}^{4}{\delta_{{1}}}^{2}{\delta_{{2}}}^{2}+{d_{{r}}}^{4}
{\delta_{{2}}}^{4} \Big) {{T_{reg}^{\ast}}}^{4}+ \Big( 2\,\rho_{{3}}\rho_{{
		1}} \Big( 10\,d_{{a}}{\delta_{{1}}}^{2}\delta_{{2}}{\rho_{{1}}}^{2}-2
\,d_{{a}}\delta_{{1}}{\delta_{{2}}}^{2}{\rho_{{1}}}^{2}+8\,d_{{a}}
\delta_{{1}}{\delta_{{2}}}^{2}\rho_{{1}}\rho_{{3}}+4\,d_{{a}}{\delta_{
		{2}}}^{3}{\rho_{{3}}}^{2}+3\,d_{{i}}{\delta_{{1}}}^{3}{\rho_{{1}}}^{2}
+3\,d_{{i}}{\delta_{{1}}}^{2}\delta_{{2}}{\rho_{{1}}}^{2}+6\,d_{{i}}{
	\delta_{{1}}}^{2}\delta_{{2}}\rho_{{1}}\rho_{{3}}-2\,d_{{i}}\delta_{{1
}}{\delta_{{2}}}^{2}{\rho_{{1}}}^{2}+20\,d_{{i}}\delta_{{1}}{\delta_{{
			2}}}^{2}\rho_{{1}}\rho_{{3}}+12\,d_{{i}}\delta_{{1}}{\delta_{{2}}}^{2}
{\rho_{{3}}}^{2}-2\,d_{{i}}{\delta_{{2}}}^{3}\rho_{{1}}\rho_{{3}}+20\,
d_{{i}}{\delta_{{2}}}^{3}{\rho_{{3}}}^{2}+d_{{r}}{\delta_{{1}}}^{3}{
	\rho_{{1}}}^{2}-2\,d_{{r}}{\delta_{{1}}}^{2}\delta_{{2}}{\rho_{{1}}}^{
	2}-5\,d_{{r}}\delta_{{1}}{\delta_{{2}}}^{2}{\rho_{{1}}}^{2}+3\,d_{{r}}
\delta_{{1}}{\delta_{{2}}}^{2}{\rho_{{3}}}^{2}-4\,d_{{r}}{\delta_{{2}}
}^{3}{\rho_{{3}}}^{2} \Big) {I^{\ast}}^{4}+ \Big( -4\,{d_{{a}}}^{3}\delta_
{{1}}{\delta_{{2}}}^{2}{\rho_{{1}}}^{2}-4\,{d_{{a}}}^{3}{\delta_{{2}}}
^{3}\rho_{{1}}\rho_{{3}}-8\,{d_{{a}}}^{2}d_{{i}}{\delta_{{1}}}^{2}
\delta_{{2}}{\rho_{{1}}}^{2}-6\,{d_{{a}}}^{2}d_{{i}}\delta_{{1}}{
	\delta_{{2}}}^{2}{\rho_{{1}}}^{2}-36\,{d_{{a}}}^{2}d_{{i}}\delta_{{1}}
{\delta_{{2}}}^{2}\rho_{{1}}\rho_{{3}}-4\,{d_{{a}}}^{2}d_{{i}}{\delta_
	{{2}}}^{3}\rho_{{1}}\rho_{{3}}-16\,{d_{{a}}}^{2}d_{{i}}{\delta_{{2}}}^
{3}{\rho_{{3}}}^{2}-2\,{d_{{a}}}^{2}d_{{r}}\delta_{{1}}{\delta_{{2}}}^
{2}\rho_{{1}}\rho_{{3}}+4\,{d_{{a}}}^{2}d_{{r}}{\delta_{{2}}}^{3}\rho_
{{1}}\rho_{{3}}-4\,d_{{a}}{d_{{i}}}^{2}{\delta_{{1}}}^{3}{\rho_{{1}}}^
{2}-12\,d_{{a}}{d_{{i}}}^{2}{\delta_{{1}}}^{2}\delta_{{2}}{\rho_{{1}}}
^{2}-36\,d_{{a}}{d_{{i}}}^{2}{\delta_{{1}}}^{2}\delta_{{2}}\rho_{{1}}
\rho_{{3}}+2\,d_{{a}}{d_{{i}}}^{2}\delta_{{1}}{\delta_{{2}}}^{2}{\rho_
	{{1}}}^{2}-24\,d_{{a}}{d_{{i}}}^{2}\delta_{{1}}{\delta_{{2}}}^{2}\rho_
{{1}}\rho_{{3}}-48\,d_{{a}}{d_{{i}}}^{2}\delta_{{1}}{\delta_{{2}}}^{2}
{\rho_{{3}}}^{2}-8\,d_{{a}}d_{{i}}d_{{r}}{\delta_{{1}}}^{2}\delta_{{2}
}\rho_{{1}}\rho_{{3}}+20\,d_{{a}}d_{{i}}d_{{r}}\delta_{{1}}{\delta_{{2
}}}^{2}\rho_{{1}}\rho_{{3}}-4\,d_{{a}}d_{{i}}d_{{r}}{\delta_{{2}}}^{3}
\rho_{{1}}\rho_{{3}}-4\,d_{{a}}{d_{{r}}}^{2}{\delta_{{1}}}^{3}{\rho_{{
			1}}}^{2}-4\,d_{{a}}{d_{{r}}}^{2}{\delta_{{1}}}^{2}\delta_{{2}}{\rho_{{
			1}}}^{2}-16\,d_{{a}}{d_{{r}}}^{2}{\delta_{{1}}}^{2}\delta_{{2}}\rho_{{
		1}}\rho_{{3}}-4\,d_{{a}}{d_{{r}}}^{2}\delta_{{1}}{\delta_{{2}}}^{2}{
	\rho_{{1}}}^{2}-4\,d_{{a}}{d_{{r}}}^{2}\delta_{{1}}{\delta_{{2}}}^{2}
\rho_{{1}}\rho_{{3}}-20\,d_{{a}}{d_{{r}}}^{2}\delta_{{1}}{\delta_{{2}}
}^{2}{\rho_{{3}}}^{2}-4\,d_{{a}}{d_{{r}}}^{2}{\delta_{{2}}}^{3}\rho_{{
		1}}\rho_{{3}}-2\,{d_{{i}}}^{3}{\delta_{{1}}}^{3}{\rho_{{1}}}^{2}-4\,{d
	_{{i}}}^{3}{\delta_{{1}}}^{3}\rho_{{1}}\rho_{{3}}-2\,{d_{{i}}}^{3}{
	\delta_{{1}}}^{2}\delta_{{2}}{\rho_{{1}}}^{2}-12\,{d_{{i}}}^{3}{\delta
	_{{1}}}^{2}\delta_{{2}}\rho_{{1}}\rho_{{3}}-16\,{d_{{i}}}^{3}{\delta_{
		{1}}}^{2}\delta_{{2}}{\rho_{{3}}}^{2}-2\,{d_{{i}}}^{2}d_{{r}}{\delta_{
		{1}}}^{3}\rho_{{1}}\rho_{{3}}+8\,{d_{{i}}}^{2}d_{{r}}{\delta_{{1}}}^{2
}\delta_{{2}}\rho_{{1}}\rho_{{3}}+12\,{d_{{i}}}^{2}d_{{r}}\delta_{{1}}
{\delta_{{2}}}^{2}\rho_{{1}}\rho_{{3}}-2\,d_{{i}}{d_{{r}}}^{2}{\delta_
	{{1}}}^{3}{\rho_{{1}}}^{2}-4\,d_{{i}}{d_{{r}}}^{2}{\delta_{{1}}}^{3}
\rho_{{1}}\rho_{{3}}-4\,d_{{i}}{d_{{r}}}^{2}{\delta_{{1}}}^{2}\delta_{
	{2}}{\rho_{{1}}}^{2}-8\,d_{{i}}{d_{{r}}}^{2}{\delta_{{1}}}^{2}\delta_{
	{2}}\rho_{{1}}\rho_{{3}}-20\,d_{{i}}{d_{{r}}}^{2}{\delta_{{1}}}^{2}
\delta_{{2}}{\rho_{{3}}}^{2}-6\,d_{{i}}{d_{{r}}}^{2}\delta_{{1}}{
	\delta_{{2}}}^{2}{\rho_{{1}}}^{2}-16\,d_{{i}}{d_{{r}}}^{2}\delta_{{1}}
{\delta_{{2}}}^{2}\rho_{{1}}\rho_{{3}}-4\,d_{{i}}{d_{{r}}}^{2}{\delta_
	{{2}}}^{3}\rho_{{1}}\rho_{{3}}-32\,d_{{i}}{d_{{r}}}^{2}{\delta_{{2}}}^
{3}{\rho_{{3}}}^{2}-2\,{d_{{r}}}^{3}{\delta_{{1}}}^{3}\rho_{{1}}\rho_{
	{3}}+4\,{d_{{r}}}^{3}{\delta_{{1}}}^{2}\delta_{{2}}\rho_{{1}}\rho_{{3}
}+2\,{d_{{r}}}^{3}\delta_{{1}}{\delta_{{2}}}^{2}\rho_{{1}}\rho_{{3}}+4
\,{d_{{r}}}^{3}{\delta_{{2}}}^{3}\rho_{{1}}\rho_{{3}} \Big) {I^{\ast}}^{2}+
4\,{d_{{a}}}^{4}d_{{i}}{\delta_{{2}}}^{3}+24\,{d_{{a}}}^{3}{d_{{i}}}^{
	2}\delta_{{1}}{\delta_{{2}}}^{2}+16\,{d_{{a}}}^{3}{d_{{r}}}^{2}\delta_
{{1}}{\delta_{{2}}}^{2}+24\,{d_{{i}}}^{3}{\delta_{{1}}}^{2}\delta_{{2}
}{d_{{a}}}^{2}+48\,{d_{{a}}}^{2}d_{{i}}{d_{{r}}}^{2}{\delta_{{1}}}^{2}
\delta_{{2}}+16\,{d_{{a}}}^{2}d_{{i}}{d_{{r}}}^{2}{\delta_{{2}}}^{3}+4
\,d_{{a}}{d_{{i}}}^{4}{\delta_{{1}}}^{3}+16\,d_{{a}}{d_{{i}}}^{2}{
	\delta_{{1}}}^{3}{d_{{r}}}^{2}+48\,d_{{a}}{d_{{i}}}^{2}\delta_{{1}}{d_
	{{r}}}^{2}{\delta_{{2}}}^{2}+4\,d_{{a}}{d_{{r}}}^{4}{\delta_{{1}}}^{3}
+8\,d_{{a}}{d_{{r}}}^{4}\delta_{{1}}{\delta_{{2}}}^{2}+16\,{d_{{i}}}^{
	3}{\delta_{{1}}}^{2}\delta_{{2}}{d_{{r}}}^{2}+8\,d_{{i}}{d_{{r}}}^{4}{
	\delta_{{1}}}^{2}\delta_{{2}}+4\,d_{{i}}{d_{{r}}}^{4}{\delta_{{2}}}^{3
} \Big) {{T_{reg}^{\ast}}}^{3}+ \Big( -{\rho_{{1}}}^{2}{\rho_{{3}}}^{2}
\Big( {\delta_{{1}}}^{2}{\rho_{{1}}}^{2}+2\,\delta_{{1}}\delta_{{2}}
{\rho_{{1}}}^{2}+6\,\delta_{{1}}\delta_{{2}}\rho_{{1}}\rho_{{3}}+{
	\delta_{{2}}}^{2}{\rho_{{1}}}^{2}+2\,{\delta_{{2}}}^{2}\rho_{{1}}\rho_
{{3}}+5\,{\delta_{{2}}}^{2}{\rho_{{3}}}^{2} \Big) {I^{\ast}}^{6}+2\,\rho_{{
		3}}\rho_{{1}} \Big( 8\,{d_{{a}}}^{2}\delta_{{1}}\delta_{{2}}{\rho_{{1
}}}^{2}-{d_{{a}}}^{2}{\delta_{{2}}}^{2}{\rho_{{1}}}^{2}+4\,{d_{{a}}}^{
	2}{\delta_{{2}}}^{2}\rho_{{1}}\rho_{{3}}+7\,d_{{a}}d_{{i}}{\delta_{{1}
}}^{2}{\rho_{{1}}}^{2}+2\,d_{{a}}d_{{i}}\delta_{{1}}\delta_{{2}}{\rho_
	{{1}}}^{2}+14\,d_{{a}}d_{{i}}\delta_{{1}}\delta_{{2}}\rho_{{1}}\rho_{{
		3}}+2\,d_{{a}}d_{{i}}{\delta_{{2}}}^{2}{\rho_{{1}}}^{2}+12\,d_{{a}}d_{
	{i}}{\delta_{{2}}}^{2}{\rho_{{3}}}^{2}+3\,d_{{a}}d_{{r}}{\delta_{{1}}}
^{2}{\rho_{{1}}}^{2}-4\,d_{{a}}d_{{r}}\delta_{{1}}\delta_{{2}}{\rho_{{
			1}}}^{2}-d_{{a}}d_{{r}}{\delta_{{2}}}^{2}{\rho_{{1}}}^{2}+3\,d_{{a}}d_
{{r}}{\delta_{{2}}}^{2}{\rho_{{3}}}^{2}+4\,{d_{{i}}}^{2}{\delta_{{1}}}
^{2}{\rho_{{1}}}^{2}+3\,{d_{{i}}}^{2}{\delta_{{1}}}^{2}\rho_{{1}}\rho_
{{3}}+4\,{d_{{i}}}^{2}\delta_{{1}}\delta_{{2}}{\rho_{{1}}}^{2}+22\,{d_
	{{i}}}^{2}\delta_{{1}}\delta_{{2}}\rho_{{1}}\rho_{{3}}+12\,{d_{{i}}}^{
	2}\delta_{{1}}\delta_{{2}}{\rho_{{3}}}^{2}+{d_{{i}}}^{2}{\delta_{{2}}}
^{2}{\rho_{{1}}}^{2}+3\,{d_{{i}}}^{2}{\delta_{{2}}}^{2}\rho_{{1}}\rho_
{{3}}+36\,{d_{{i}}}^{2}{\delta_{{2}}}^{2}{\rho_{{3}}}^{2}-d_{{i}}d_{{r
}}{\delta_{{1}}}^{2}{\rho_{{1}}}^{2}-7\,d_{{i}}d_{{r}}\delta_{{1}}
\delta_{{2}}{\rho_{{1}}}^{2}+9\,d_{{i}}d_{{r}}\delta_{{1}}\delta_{{2}}
{\rho_{{3}}}^{2}-2\,d_{{i}}d_{{r}}{\delta_{{2}}}^{2}{\rho_{{1}}}^{2}-
12\,d_{{i}}d_{{r}}{\delta_{{2}}}^{2}{\rho_{{3}}}^{2}+{d_{{r}}}^{2}{
	\delta_{{1}}}^{2}\rho_{{1}}\rho_{{3}}+2\,{d_{{r}}}^{2}\delta_{{1}}
\delta_{{2}}\rho_{{1}}\rho_{{3}}+2\,{d_{{r}}}^{2}\delta_{{1}}\delta_{{
		2}}{\rho_{{3}}}^{2}+4\,{d_{{r}}}^{2}{\delta_{{2}}}^{2}\rho_{{1}}\rho_{
	{3}}+3\,{d_{{r}}}^{2}{\delta_{{2}}}^{2}{\rho_{{3}}}^{2} \Big) {I^{\ast}}^{4
}+ \Big( -{d_{{a}}}^{4}{\delta_{{2}}}^{2}{\rho_{{1}}}^{2}-6\,{d_{{a}}
}^{3}d_{{i}}\delta_{{1}}\delta_{{2}}{\rho_{{1}}}^{2}-10\,{d_{{a}}}^{3}
d_{{i}}{\delta_{{2}}}^{2}\rho_{{1}}\rho_{{3}}-2\,{d_{{a}}}^{3}d_{{r}}{
	\delta_{{2}}}^{2}\rho_{{1}}\rho_{{3}}-7\,{d_{{a}}}^{2}{d_{{i}}}^{2}{
	\delta_{{1}}}^{2}{\rho_{{1}}}^{2}-6\,{d_{{a}}}^{2}{d_{{i}}}^{2}\delta_
{{1}}\delta_{{2}}{\rho_{{1}}}^{2}-30\,{d_{{a}}}^{2}{d_{{i}}}^{2}\delta
_{{1}}\delta_{{2}}\rho_{{1}}\rho_{{3}}-2\,{d_{{a}}}^{2}{d_{{i}}}^{2}{
	\delta_{{2}}}^{2}{\rho_{{1}}}^{2}-12\,{d_{{a}}}^{2}{d_{{i}}}^{2}{
	\delta_{{2}}}^{2}\rho_{{1}}\rho_{{3}}-24\,{d_{{a}}}^{2}{d_{{i}}}^{2}{
	\delta_{{2}}}^{2}{\rho_{{3}}}^{2}-16\,{d_{{a}}}^{2}d_{{i}}d_{{r}}
\delta_{{1}}\delta_{{2}}\rho_{{1}}\rho_{{3}}+10\,{d_{{a}}}^{2}d_{{i}}d
_{{r}}{\delta_{{2}}}^{2}\rho_{{1}}\rho_{{3}}-6\,{d_{{a}}}^{2}{d_{{r}}}
^{2}{\delta_{{1}}}^{2}{\rho_{{1}}}^{2}-2\,{d_{{a}}}^{2}{d_{{r}}}^{2}
\delta_{{1}}\delta_{{2}}{\rho_{{1}}}^{2}-14\,{d_{{a}}}^{2}{d_{{r}}}^{2
}\delta_{{1}}\delta_{{2}}\rho_{{1}}\rho_{{3}}-2\,{d_{{a}}}^{2}{d_{{r}}
}^{2}{\delta_{{2}}}^{2}{\rho_{{1}}}^{2}-2\,{d_{{a}}}^{2}{d_{{r}}}^{2}{
	\delta_{{2}}}^{2}\rho_{{1}}\rho_{{3}}-10\,{d_{{a}}}^{2}{d_{{r}}}^{2}{
	\delta_{{2}}}^{2}{\rho_{{3}}}^{2}-4\,d_{{a}}{d_{{i}}}^{3}{\delta_{{1}}
}^{2}{\rho_{{1}}}^{2}-10\,d_{{a}}{d_{{i}}}^{3}{\delta_{{1}}}^{2}\rho_{
	{1}}\rho_{{3}}-2\,d_{{a}}{d_{{i}}}^{3}\delta_{{1}}\delta_{{2}}{\rho_{{
			1}}}^{2}-24\,d_{{a}}{d_{{i}}}^{3}\delta_{{1}}\delta_{{2}}\rho_{{1}}
\rho_{{3}}-32\,d_{{a}}{d_{{i}}}^{3}\delta_{{1}}\delta_{{2}}{\rho_{{3}}
}^{2}-10\,d_{{a}}{d_{{i}}}^{2}d_{{r}}{\delta_{{1}}}^{2}\rho_{{1}}\rho_
{{3}}+16\,d_{{a}}{d_{{i}}}^{2}d_{{r}}\delta_{{1}}\delta_{{2}}\rho_{{1}
}\rho_{{3}}-12\,d_{{a}}{d_{{i}}}^{2}d_{{r}}{\delta_{{2}}}^{2}\rho_{{1}
}\rho_{{3}}-4\,d_{{a}}d_{{i}}{d_{{r}}}^{2}{\delta_{{1}}}^{2}{\rho_{{1}
}}^{2}-10\,d_{{a}}d_{{i}}{d_{{r}}}^{2}{\delta_{{1}}}^{2}\rho_{{1}}\rho
_{{3}}-6\,d_{{a}}d_{{i}}{d_{{r}}}^{2}\delta_{{1}}\delta_{{2}}{\rho_{{1
}}}^{2}-16\,d_{{a}}d_{{i}}{d_{{r}}}^{2}\delta_{{1}}\delta_{{2}}\rho_{{
		1}}\rho_{{3}}-40\,d_{{a}}d_{{i}}{d_{{r}}}^{2}\delta_{{1}}\delta_{{2}}{
	\rho_{{3}}}^{2}-10\,d_{{a}}d_{{i}}{d_{{r}}}^{2}{\delta_{{2}}}^{2}\rho_
{{1}}\rho_{{3}}-6\,d_{{a}}{d_{{r}}}^{3}{\delta_{{1}}}^{2}\rho_{{1}}
\rho_{{3}}+8\,d_{{a}}{d_{{r}}}^{3}\delta_{{1}}\delta_{{2}}\rho_{{1}}
\rho_{{3}}-2\,d_{{a}}{d_{{r}}}^{3}{\delta_{{2}}}^{2}\rho_{{1}}\rho_{{3
}}-{d_{{i}}}^{4}{\delta_{{1}}}^{2}{\rho_{{1}}}^{2}-4\,{d_{{i}}}^{4}{
	\delta_{{1}}}^{2}\rho_{{1}}\rho_{{3}}-4\,{d_{{i}}}^{4}{\delta_{{1}}}^{
	2}{\rho_{{3}}}^{2}+2\,{d_{{i}}}^{3}d_{{r}}{\delta_{{1}}}^{2}\rho_{{1}}
\rho_{{3}}+4\,{d_{{i}}}^{3}d_{{r}}\delta_{{1}}\delta_{{2}}\rho_{{1}}
\rho_{{3}}-3\,{d_{{i}}}^{2}{d_{{r}}}^{2}{\delta_{{1}}}^{2}{\rho_{{1}}}
^{2}-6\,{d_{{i}}}^{2}{d_{{r}}}^{2}{\delta_{{1}}}^{2}\rho_{{1}}\rho_{{3
}}-10\,{d_{{i}}}^{2}{d_{{r}}}^{2}{\delta_{{1}}}^{2}{\rho_{{3}}}^{2}-6
\,{d_{{i}}}^{2}{d_{{r}}}^{2}\delta_{{1}}\delta_{{2}}{\rho_{{1}}}^{2}-
14\,{d_{{i}}}^{2}{d_{{r}}}^{2}\delta_{{1}}\delta_{{2}}\rho_{{1}}\rho_{
	{3}}-{d_{{i}}}^{2}{d_{{r}}}^{2}{\delta_{{2}}}^{2}{\rho_{{1}}}^{2}-12\,
{d_{{i}}}^{2}{d_{{r}}}^{2}{\delta_{{2}}}^{2}\rho_{{1}}\rho_{{3}}-48\,{
	d_{{i}}}^{2}{d_{{r}}}^{2}{\delta_{{2}}}^{2}{\rho_{{3}}}^{2}+2\,d_{{i}}
{d_{{r}}}^{3}{\delta_{{1}}}^{2}\rho_{{1}}\rho_{{3}}+10\,d_{{i}}{d_{{r}
}}^{3}{\delta_{{2}}}^{2}\rho_{{1}}\rho_{{3}}-{d_{{r}}}^{4}{\delta_{{1}
}}^{2}{\rho_{{3}}}^{2}-5\,{d_{{r}}}^{4}{\delta_{{2}}}^{2}{\rho_{{3}}}^
{2} \Big) {I^{\ast}}^{2}+6\,{d_{{a}}}^{4}{d_{{i}}}^{2}{\delta_{{2}}}^{2}+4
\,{d_{{a}}}^{4}{d_{{r}}}^{2}{\delta_{{2}}}^{2}+16\,{d_{{a}}}^{3}{d_{{i
}}}^{3}\delta_{{1}}\delta_{{2}}+32\,{d_{{a}}}^{3}d_{{i}}{d_{{r}}}^{2}
\delta_{{1}}\delta_{{2}}+6\,{d_{{a}}}^{2}{d_{{i}}}^{4}{\delta_{{1}}}^{
	2}+24\,{d_{{a}}}^{2}{d_{{i}}}^{2}{d_{{r}}}^{2}{\delta_{{1}}}^{2}+24\,{
	d_{{a}}}^{2}{d_{{i}}}^{2}{d_{{r}}}^{2}{\delta_{{2}}}^{2}+6\,{d_{{a}}}^
{2}{d_{{r}}}^{4}{\delta_{{1}}}^{2}+4\,{d_{{a}}}^{2}{d_{{r}}}^{4}{
	\delta_{{2}}}^{2}+32\,d_{{a}}{d_{{i}}}^{3}{d_{{r}}}^{2}\delta_{{1}}
\delta_{{2}}+16\,d_{{a}}d_{{i}}{d_{{r}}}^{4}\delta_{{1}}\delta_{{2}}+4
\,{d_{{i}}}^{4}{d_{{r}}}^{2}{\delta_{{1}}}^{2}+4\,{d_{{i}}}^{2}{d_{{r}
}}^{4}{\delta_{{1}}}^{2}+6\,{d_{{i}}}^{2}{d_{{r}}}^{4}{\delta_{{2}}}^{
	2} \Big) {{T_{reg}^{\ast}}}^{2}+ \Big( -2\,{\rho_{{1}}}^{2}{\rho_{{3}}}^{2}
\Big( d_{{a}}\delta_{{1}}{\rho_{{1}}}^{2}+2\,d_{{a}}\delta_{{2}}\rho
_{{1}}\rho_{{3}}+d_{{i}}\delta_{{1}}{\rho_{{1}}}^{2}+2\,d_{{i}}\delta_
{{1}}\rho_{{1}}\rho_{{3}}+2\,d_{{i}}\delta_{{2}}{\rho_{{1}}}^{2}+6\,d_
{{i}}\delta_{{2}}\rho_{{1}}\rho_{{3}}+6\,d_{{i}}\delta_{{2}}{\rho_{{3}
}}^{2}+d_{{r}}\delta_{{1}}\rho_{{1}}\rho_{{3}}-2\,d_{{r}}\delta_{{2}}
\rho_{{1}}\rho_{{3}} \Big) {I^{\ast}}^{6}+2\,\rho_{{3}}\rho_{{1}} \Big( 2
\,{d_{{a}}}^{3}\delta_{{2}}{\rho_{{1}}}^{2}+5\,{d_{{a}}}^{2}d_{{i}}
\delta_{{1}}{\rho_{{1}}}^{2}+{d_{{a}}}^{2}d_{{i}}\delta_{{2}}{\rho_{{1
}}}^{2}+8\,{d_{{a}}}^{2}d_{{i}}\delta_{{2}}\rho_{{1}}\rho_{{3}}+3\,{d_
	{{a}}}^{2}d_{{r}}\delta_{{1}}{\rho_{{1}}}^{2}-2\,{d_{{a}}}^{2}d_{{r}}
\delta_{{2}}{\rho_{{1}}}^{2}+6\,d_{{a}}{d_{{i}}}^{2}\delta_{{1}}{\rho_
	{{1}}}^{2}+8\,d_{{a}}{d_{{i}}}^{2}\delta_{{1}}\rho_{{1}}\rho_{{3}}+4\,
d_{{a}}{d_{{i}}}^{2}\delta_{{2}}{\rho_{{1}}}^{2}+12\,d_{{a}}{d_{{i}}}^
{2}\delta_{{2}}{\rho_{{3}}}^{2}-2\,d_{{a}}d_{{i}}d_{{r}}\delta_{{1}}{
	\rho_{{1}}}^{2}+d_{{a}}d_{{i}}d_{{r}}\delta_{{2}}{\rho_{{1}}}^{2}+9\,d
_{{a}}d_{{i}}d_{{r}}\delta_{{2}}{\rho_{{3}}}^{2}+2\,d_{{a}}{d_{{r}}}^{
	2}\delta_{{1}}\rho_{{1}}\rho_{{3}}+2\,d_{{a}}{d_{{r}}}^{2}\delta_{{2}}
{\rho_{{3}}}^{2}+3\,{d_{{i}}}^{3}\delta_{{1}}{\rho_{{1}}}^{2}+8\,{d_{{
			i}}}^{3}\delta_{{1}}\rho_{{1}}\rho_{{3}}+4\,{d_{{i}}}^{3}\delta_{{1}}{
	\rho_{{3}}}^{2}+3\,{d_{{i}}}^{3}\delta_{{2}}{\rho_{{1}}}^{2}+8\,{d_{{i
}}}^{3}\delta_{{2}}\rho_{{1}}\rho_{{3}}+28\,{d_{{i}}}^{3}\delta_{{2}}{
	\rho_{{3}}}^{2}-{d_{{i}}}^{2}d_{{r}}\delta_{{1}}{\rho_{{1}}}^{2}+6\,{d
	_{{i}}}^{2}d_{{r}}\delta_{{1}}{\rho_{{3}}}^{2}-4\,{d_{{i}}}^{2}d_{{r}}
\delta_{{2}}{\rho_{{1}}}^{2}-12\,{d_{{i}}}^{2}d_{{r}}\delta_{{2}}{\rho
	_{{3}}}^{2}+2\,d_{{i}}{d_{{r}}}^{2}\delta_{{1}}\rho_{{1}}\rho_{{3}}+d_
{{i}}{d_{{r}}}^{2}\delta_{{1}}{\rho_{{3}}}^{2}+8\,d_{{i}}{d_{{r}}}^{2}
\delta_{{2}}\rho_{{1}}\rho_{{3}}+9\,d_{{i}}{d_{{r}}}^{2}\delta_{{2}}{
	\rho_{{3}}}^{2}+{d_{{r}}}^{3}\delta_{{1}}{\rho_{{3}}}^{2}-2\,{d_{{r}}}
^{3}\delta_{{2}}{\rho_{{3}}}^{2} \Big) {I^{\ast}}^{4}+ \Big( -2\,{d_{{a}}}
^{4}d_{{i}}\delta_{{2}}{\rho_{{1}}}^{2}-6\,{d_{{a}}}^{3}{d_{{i}}}^{2}
\delta_{{1}}{\rho_{{1}}}^{2}-8\,{d_{{a}}}^{3}{d_{{i}}}^{2}\delta_{{2}}
\rho_{{1}}\rho_{{3}}-8\,{d_{{a}}}^{3}d_{{i}}d_{{r}}\delta_{{2}}\rho_{{
		1}}\rho_{{3}}-4\,{d_{{a}}}^{3}{d_{{r}}}^{2}\delta_{{1}}{\rho_{{1}}}^{2
}-4\,{d_{{a}}}^{3}{d_{{r}}}^{2}\delta_{{2}}\rho_{{1}}\rho_{{3}}-2\,{d_
	{{a}}}^{2}{d_{{i}}}^{3}\delta_{{1}}{\rho_{{1}}}^{2}-8\,{d_{{a}}}^{2}{d
	_{{i}}}^{3}\delta_{{1}}\rho_{{1}}\rho_{{3}}-4\,{d_{{a}}}^{2}{d_{{i}}}^
{3}\delta_{{2}}{\rho_{{1}}}^{2}-12\,{d_{{a}}}^{2}{d_{{i}}}^{3}\delta_{
	{2}}\rho_{{1}}\rho_{{3}}-16\,{d_{{a}}}^{2}{d_{{i}}}^{3}\delta_{{2}}{
	\rho_{{3}}}^{2}-14\,{d_{{a}}}^{2}{d_{{i}}}^{2}d_{{r}}\delta_{{1}}\rho_
{{1}}\rho_{{3}}+8\,{d_{{a}}}^{2}{d_{{i}}}^{2}d_{{r}}\delta_{{2}}\rho_{
	{1}}\rho_{{3}}-2\,{d_{{a}}}^{2}d_{{i}}{d_{{r}}}^{2}\delta_{{1}}{\rho_{
		{1}}}^{2}-8\,{d_{{a}}}^{2}d_{{i}}{d_{{r}}}^{2}\delta_{{1}}\rho_{{1}}
\rho_{{3}}-4\,{d_{{a}}}^{2}d_{{i}}{d_{{r}}}^{2}\delta_{{2}}{\rho_{{1}}
}^{2}-8\,{d_{{a}}}^{2}d_{{i}}{d_{{r}}}^{2}\delta_{{2}}\rho_{{1}}\rho_{
	{3}}-20\,{d_{{a}}}^{2}d_{{i}}{d_{{r}}}^{2}\delta_{{2}}{\rho_{{3}}}^{2}
-6\,{d_{{a}}}^{2}{d_{{r}}}^{3}\delta_{{1}}\rho_{{1}}\rho_{{3}}+4\,{d_{
		{a}}}^{2}{d_{{r}}}^{3}\delta_{{2}}\rho_{{1}}\rho_{{3}}-2\,d_{{a}}{d_{{
			i}}}^{4}\delta_{{1}}{\rho_{{1}}}^{2}-8\,d_{{a}}{d_{{i}}}^{4}\delta_{{1
}}\rho_{{1}}\rho_{{3}}-8\,d_{{a}}{d_{{i}}}^{4}\delta_{{1}}{\rho_{{3}}}
^{2}+4\,d_{{a}}{d_{{i}}}^{3}d_{{r}}\delta_{{1}}\rho_{{1}}\rho_{{3}}-12
\,d_{{a}}{d_{{i}}}^{3}d_{{r}}\delta_{{2}}\rho_{{1}}\rho_{{3}}-6\,d_{{a
}}{d_{{i}}}^{2}{d_{{r}}}^{2}\delta_{{1}}{\rho_{{1}}}^{2}-12\,d_{{a}}{d
	_{{i}}}^{2}{d_{{r}}}^{2}\delta_{{1}}\rho_{{1}}\rho_{{3}}-20\,d_{{a}}{d
	_{{i}}}^{2}{d_{{r}}}^{2}\delta_{{1}}{\rho_{{3}}}^{2}-8\,d_{{a}}{d_{{i}
}}^{2}{d_{{r}}}^{2}\delta_{{2}}\rho_{{1}}\rho_{{3}}+4\,d_{{a}}d_{{i}}{
	d_{{r}}}^{3}\delta_{{1}}\rho_{{1}}\rho_{{3}}-8\,d_{{a}}d_{{i}}{d_{{r}}
}^{3}\delta_{{2}}\rho_{{1}}\rho_{{3}}-2\,d_{{a}}{d_{{r}}}^{4}\delta_{{
		1}}{\rho_{{3}}}^{2}-2\,{d_{{i}}}^{3}{d_{{r}}}^{2}\delta_{{1}}{\rho_{{1
}}}^{2}-4\,{d_{{i}}}^{3}{d_{{r}}}^{2}\delta_{{1}}\rho_{{1}}\rho_{{3}}-
2\,{d_{{i}}}^{3}{d_{{r}}}^{2}\delta_{{2}}{\rho_{{1}}}^{2}-12\,{d_{{i}}
}^{3}{d_{{r}}}^{2}\delta_{{2}}\rho_{{1}}\rho_{{3}}-32\,{d_{{i}}}^{3}{d
	_{{r}}}^{2}\delta_{{2}}{\rho_{{3}}}^{2}-2\,{d_{{i}}}^{2}{d_{{r}}}^{3}
\delta_{{1}}\rho_{{1}}\rho_{{3}}+8\,{d_{{i}}}^{2}{d_{{r}}}^{3}\delta_{
	{2}}\rho_{{1}}\rho_{{3}}-10\,d_{{i}}{d_{{r}}}^{4}\delta_{{2}}{\rho_{{3
}}}^{2} \Big) {I^{\ast}}^{2}+4\,{d_{{a}}}^{4}{d_{{i}}}^{3}\delta_{{2}}+8\,{
	d_{{a}}}^{4}d_{{i}}{d_{{r}}}^{2}\delta_{{2}}+4\,{d_{{a}}}^{3}{d_{{i}}}
^{4}\delta_{{1}}+16\,{d_{{a}}}^{3}{d_{{i}}}^{2}{d_{{r}}}^{2}\delta_{{1
}}+4\,{d_{{a}}}^{3}{d_{{r}}}^{4}\delta_{{1}}+16\,{d_{{a}}}^{2}{d_{{i}}
}^{3}{d_{{r}}}^{2}\delta_{{2}}+8\,{d_{{a}}}^{2}d_{{i}}{d_{{r}}}^{4}
\delta_{{2}}+8\,d_{{a}}{d_{{i}}}^{4}{d_{{r}}}^{2}\delta_{{1}}+8\,d_{{a
}}{d_{{i}}}^{2}{d_{{r}}}^{4}\delta_{{1}}+4\,{d_{{i}}}^{3}{d_{{r}}}^{4}
\delta_{{2}} \Big) {T_{reg}^{\ast}}+{I^{\ast}}^{8}{\rho_{{1}}}^{4}{\rho_{{3}}}^{4}-
{\rho_{{1}}}^{2}{\rho_{{3}}}^{2} \Big( {d_{{a}}}^{2}{\rho_{{1}}}^{2}+
2\,d_{{a}}d_{{i}}\rho_{{1}}\rho_{{3}}+2\,d_{{a}}d_{{r}}\rho_{{1}}\rho_
{{3}}+5\,{d_{{i}}}^{2}{\rho_{{1}}}^{2}+12\,{d_{{i}}}^{2}\rho_{{1}}\rho
_{{3}}+8\,{d_{{i}}}^{2}{\rho_{{3}}}^{2}-2\,d_{{i}}d_{{r}}\rho_{{1}}
\rho_{{3}}+{d_{{r}}}^{2}{\rho_{{3}}}^{2} \Big) {I^{\ast}}^{6}+2\,\rho_{{3}}
\rho_{{1}} \Big( {d_{{a}}}^{3}d_{{i}}{\rho_{{1}}}^{2}+{d_{{a}}}^{3}d_
{{r}}{\rho_{{1}}}^{2}+3\,{d_{{a}}}^{2}{d_{{i}}}^{2}{\rho_{{1}}}^{2}+5
\,{d_{{a}}}^{2}{d_{{i}}}^{2}\rho_{{1}}\rho_{{3}}-{d_{{a}}}^{2}d_{{i}}d
_{{r}}{\rho_{{1}}}^{2}+{d_{{a}}}^{2}{d_{{r}}}^{2}\rho_{{1}}\rho_{{3}}+
d_{{a}}{d_{{i}}}^{3}{\rho_{{1}}}^{2}+4\,d_{{a}}{d_{{i}}}^{3}{\rho_{{3}
}}^{2}+3\,d_{{a}}{d_{{i}}}^{2}d_{{r}}{\rho_{{1}}}^{2}+6\,d_{{a}}{d_{{i
}}}^{2}d_{{r}}{\rho_{{3}}}^{2}+d_{{a}}d_{{i}}{d_{{r}}}^{2}{\rho_{{3}}}
^{2}+d_{{a}}{d_{{r}}}^{3}{\rho_{{3}}}^{2}+2\,{d_{{i}}}^{4}{\rho_{{1}}}
^{2}+4\,{d_{{i}}}^{4}\rho_{{1}}\rho_{{3}}+8\,{d_{{i}}}^{4}{\rho_{{3}}}
^{2}-{d_{{i}}}^{3}d_{{r}}{\rho_{{1}}}^{2}-4\,{d_{{i}}}^{3}d_{{r}}{\rho
	_{{3}}}^{2}+5\,{d_{{i}}}^{2}{d_{{r}}}^{2}\rho_{{1}}\rho_{{3}}+6\,{d_{{
			i}}}^{2}{d_{{r}}}^{2}{\rho_{{3}}}^{2}-d_{{i}}{d_{{r}}}^{3}{\rho_{{3}}}
^{2} \Big) {I^{\ast}}^{4}+ \Big( -2\,{d_{{a}}}^{4}{d_{{i}}}^{2}{\rho_{{1}}
}^{2}-{d_{{a}}}^{4}{d_{{r}}}^{2}{\rho_{{1}}}^{2}-2\,{d_{{a}}}^{3}{d_{{
			i}}}^{3}\rho_{{1}}\rho_{{3}}-6\,{d_{{a}}}^{3}{d_{{i}}}^{2}d_{{r}}\rho_
{{1}}\rho_{{3}}-2\,{d_{{a}}}^{3}d_{{i}}{d_{{r}}}^{2}\rho_{{1}}\rho_{{3
}}-2\,{d_{{a}}}^{3}{d_{{r}}}^{3}\rho_{{1}}\rho_{{3}}-2\,{d_{{a}}}^{2}{
	d_{{i}}}^{4}{\rho_{{1}}}^{2}-4\,{d_{{a}}}^{2}{d_{{i}}}^{4}\rho_{{1}}
\rho_{{3}}-4\,{d_{{a}}}^{2}{d_{{i}}}^{4}{\rho_{{3}}}^{2}+2\,{d_{{a}}}^
{2}{d_{{i}}}^{3}d_{{r}}\rho_{{1}}\rho_{{3}}-4\,{d_{{a}}}^{2}{d_{{i}}}^
{2}{d_{{r}}}^{2}{\rho_{{1}}}^{2}-6\,{d_{{a}}}^{2}{d_{{i}}}^{2}{d_{{r}}
}^{2}\rho_{{1}}\rho_{{3}}-10\,{d_{{a}}}^{2}{d_{{i}}}^{2}{d_{{r}}}^{2}{
	\rho_{{3}}}^{2}+2\,{d_{{a}}}^{2}d_{{i}}{d_{{r}}}^{3}\rho_{{1}}\rho_{{3
}}-{d_{{a}}}^{2}{d_{{r}}}^{4}{\rho_{{3}}}^{2}-4\,d_{{a}}{d_{{i}}}^{4}d
_{{r}}\rho_{{1}}\rho_{{3}}-2\,d_{{a}}{d_{{i}}}^{3}{d_{{r}}}^{2}\rho_{{
		1}}\rho_{{3}}-6\,d_{{a}}{d_{{i}}}^{2}{d_{{r}}}^{3}\rho_{{1}}\rho_{{3}}
-{d_{{i}}}^{4}{d_{{r}}}^{2}{\rho_{{1}}}^{2}-4\,{d_{{i}}}^{4}{d_{{r}}}^
{2}\rho_{{1}}\rho_{{3}}-8\,{d_{{i}}}^{4}{d_{{r}}}^{2}{\rho_{{3}}}^{2}+
2\,{d_{{i}}}^{3}{d_{{r}}}^{3}\rho_{{1}}\rho_{{3}}-5\,{d_{{i}}}^{2}{d_{
		{r}}}^{4}{\rho_{{3}}}^{2} \Big) {I^{\ast}}^{2}+{d_{{a}}}^{4}{d_{{i}}}^{4}+4
\,{d_{{a}}}^{4}{d_{{i}}}^{2}{d_{{r}}}^{2}+{d_{{a}}}^{4}{d_{{r}}}^{4}+4
\,{d_{{a}}}^{2}{d_{{i}}}^{4}{d_{{r}}}^{2}+4\,{d_{{a}}}^{2}{d_{{i}}}^{2
}{d_{{r}}}^{4}+{d_{{i}}}^{4}{d_{{r}}}^{4}$.\vspace{6pt}

$b_6=2\,{\delta_{{1}}}^{2}{\delta_{{2}}}^{2} \Big( {\delta_{{1}}}^{2}+{
	\delta_{{2}}}^{2} \Big) {{T_{reg}^{\ast}}}^{6}+4\,\delta_{{1}}\delta_{{2}}
\Big( 2\,d_{{a}}{\delta_{{1}}}^{2}\delta_{{2}}+d_{{a}}{\delta_{{2}}}
^{3}+d_{{i}}{\delta_{{1}}}^{3}+2\,d_{{i}}\delta_{{1}}{\delta_{{2}}}^{2
} \Big) {{T_{reg}^{\ast}}}^{5}+ \Big(  \Big( -{\delta_{{1}}}^{4}{\rho_{{1}
}}^{2}-2\,{\delta_{{1}}}^{3}\delta_{{2}}{\rho_{{1}}}^{2}-6\,{\delta_{{
			1}}}^{3}\delta_{{2}}\rho_{{1}}\rho_{{3}}-3\,{\delta_{{1}}}^{2}{\delta_
	{{2}}}^{2}{\rho_{{1}}}^{2}-2\,{\delta_{{1}}}^{2}{\delta_{{2}}}^{2}\rho
_{{1}}\rho_{{3}}-5\,{\delta_{{1}}}^{2}{\delta_{{2}}}^{2}{\rho_{{3}}}^{
	2}-2\,\delta_{{1}}{\delta_{{2}}}^{3}{\rho_{{1}}}^{2}-6\,\delta_{{1}}{
	\delta_{{2}}}^{3}\rho_{{1}}\rho_{{3}}-4\,{\delta_{{2}}}^{4}{\rho_{{3}}
}^{2} \Big) {I^{\ast}}^{2}+12\,{d_{{a}}}^{2}{\delta_{{1}}}^{2}{\delta_{{2}}
}^{2}+2\,{d_{{a}}}^{2}{\delta_{{2}}}^{4}+16\,d_{{a}}d_{{i}}{\delta_{{1
}}}^{3}\delta_{{2}}+16\,d_{{a}}d_{{i}}\delta_{{1}}{\delta_{{2}}}^{3}+2
\,{d_{{i}}}^{2}{\delta_{{1}}}^{4}+12\,{d_{{i}}}^{2}{\delta_{{1}}}^{2}{
	\delta_{{2}}}^{2}+2\,{d_{{r}}}^{2}{\delta_{{1}}}^{4}+8\,{d_{{r}}}^{2}{
	\delta_{{1}}}^{2}{\delta_{{2}}}^{2}+2\,{d_{{r}}}^{2}{\delta_{{2}}}^{4}
\Big) {{T_{reg}^{\ast}}}^{4}+ \Big(  \Big( -4\,d_{{a}}{\delta_{{1}}}^{3}{
	\rho_{{1}}}^{2}-4\,d_{{a}}{\delta_{{1}}}^{2}\delta_{{2}}{\rho_{{1}}}^{
	2}-16\,d_{{a}}{\delta_{{1}}}^{2}\delta_{{2}}\rho_{{1}}\rho_{{3}}-4\,d_
{{a}}\delta_{{1}}{\delta_{{2}}}^{2}{\rho_{{1}}}^{2}-4\,d_{{a}}\delta_{
	{1}}{\delta_{{2}}}^{2}\rho_{{1}}\rho_{{3}}-10\,d_{{a}}\delta_{{1}}{
	\delta_{{2}}}^{2}{\rho_{{3}}}^{2}-4\,d_{{a}}{\delta_{{2}}}^{3}\rho_{{1
}}\rho_{{3}}-2\,d_{{i}}{\delta_{{1}}}^{3}{\rho_{{1}}}^{2}-4\,d_{{i}}{
	\delta_{{1}}}^{3}\rho_{{1}}\rho_{{3}}-4\,d_{{i}}{\delta_{{1}}}^{2}
\delta_{{2}}{\rho_{{1}}}^{2}-8\,d_{{i}}{\delta_{{1}}}^{2}\delta_{{2}}
\rho_{{1}}\rho_{{3}}-10\,d_{{i}}{\delta_{{1}}}^{2}\delta_{{2}}{\rho_{{
			3}}}^{2}-6\,d_{{i}}\delta_{{1}}{\delta_{{2}}}^{2}{\rho_{{1}}}^{2}-16\,
d_{{i}}\delta_{{1}}{\delta_{{2}}}^{2}\rho_{{1}}\rho_{{3}}-4\,d_{{i}}{
	\delta_{{2}}}^{3}\rho_{{1}}\rho_{{3}}-16\,d_{{i}}{\delta_{{2}}}^{3}{
	\rho_{{3}}}^{2}-2\,d_{{r}}{\delta_{{1}}}^{3}\rho_{{1}}\rho_{{3}}+4\,d_
{{r}}{\delta_{{1}}}^{2}\delta_{{2}}\rho_{{1}}\rho_{{3}}+2\,d_{{r}}
\delta_{{1}}{\delta_{{2}}}^{2}\rho_{{1}}\rho_{{3}}+4\,d_{{r}}{\delta_{
		{2}}}^{3}\rho_{{1}}\rho_{{3}} \Big) {I^{\ast}}^{2}+8\,{d_{{a}}}^{3}\delta_{
	{1}}{\delta_{{2}}}^{2}+24\,{d_{{a}}}^{2}d_{{i}}{\delta_{{1}}}^{2}
\delta_{{2}}+8\,{d_{{a}}}^{2}d_{{i}}{\delta_{{2}}}^{3}+8\,d_{{a}}{d_{{
			i}}}^{2}{\delta_{{1}}}^{3}+24\,d_{{a}}{d_{{i}}}^{2}\delta_{{1}}{\delta
	_{{2}}}^{2}+8\,d_{{a}}{d_{{r}}}^{2}{\delta_{{1}}}^{3}+16\,d_{{a}}{d_{{
			r}}}^{2}\delta_{{1}}{\delta_{{2}}}^{2}+8\,{d_{{i}}}^{3}{\delta_{{1}}}^
{2}\delta_{{2}}+16\,d_{{i}}{d_{{r}}}^{2}{\delta_{{1}}}^{2}\delta_{{2}}
+8\,d_{{i}}{d_{{r}}}^{2}{\delta_{{2}}}^{3} \Big) {{T_{reg}^{\ast}}}^{3}+
\Big( 2\,\rho_{{1}}\rho_{{3}} \Big( {\delta_{{1}}}^{2}{\rho_{{1}}}^
{2}+{\delta_{{1}}}^{2}\rho_{{1}}\rho_{{3}}+4\,\delta_{{1}}\delta_{{2}}
{\rho_{{1}}}^{2}+2\,\delta_{{1}}\delta_{{2}}\rho_{{1}}\rho_{{3}}+2\,
\delta_{{1}}\delta_{{2}}{\rho_{{3}}}^{2}-{\delta_{{2}}}^{2}{\rho_{{1}}
}^{2}+4\,{\delta_{{2}}}^{2}\rho_{{1}}\rho_{{3}}+3\,{\delta_{{2}}}^{2}{
	\rho_{{3}}}^{2} \Big) {I^{\ast}}^{4}+ \Big( -6\,{d_{{a}}}^{2}{\delta_{{1}}
}^{2}{\rho_{{1}}}^{2}-2\,{d_{{a}}}^{2}\delta_{{1}}\delta_{{2}}{\rho_{{
			1}}}^{2}-14\,{d_{{a}}}^{2}\delta_{{1}}\delta_{{2}}\rho_{{1}}\rho_{{3}}
-2\,{d_{{a}}}^{2}{\delta_{{2}}}^{2}{\rho_{{1}}}^{2}-2\,{d_{{a}}}^{2}{
	\delta_{{2}}}^{2}\rho_{{1}}\rho_{{3}}-5\,{d_{{a}}}^{2}{\delta_{{2}}}^{
	2}{\rho_{{3}}}^{2}-4\,d_{{a}}d_{{i}}{\delta_{{1}}}^{2}{\rho_{{1}}}^{2}
-10\,d_{{a}}d_{{i}}{\delta_{{1}}}^{2}\rho_{{1}}\rho_{{3}}-6\,d_{{a}}d_
{{i}}\delta_{{1}}\delta_{{2}}{\rho_{{1}}}^{2}-16\,d_{{a}}d_{{i}}\delta
_{{1}}\delta_{{2}}\rho_{{1}}\rho_{{3}}-20\,d_{{a}}d_{{i}}\delta_{{1}}
\delta_{{2}}{\rho_{{3}}}^{2}-10\,d_{{a}}d_{{i}}{\delta_{{2}}}^{2}\rho_
{{1}}\rho_{{3}}-6\,d_{{a}}d_{{r}}{\delta_{{1}}}^{2}\rho_{{1}}\rho_{{3}
}+8\,d_{{a}}d_{{r}}\delta_{{1}}\delta_{{2}}\rho_{{1}}\rho_{{3}}-2\,d_{
	{a}}d_{{r}}{\delta_{{2}}}^{2}\rho_{{1}}\rho_{{3}}-3\,{d_{{i}}}^{2}{
	\delta_{{1}}}^{2}{\rho_{{1}}}^{2}-6\,{d_{{i}}}^{2}{\delta_{{1}}}^{2}
\rho_{{1}}\rho_{{3}}-5\,{d_{{i}}}^{2}{\delta_{{1}}}^{2}{\rho_{{3}}}^{2
}-6\,{d_{{i}}}^{2}\delta_{{1}}\delta_{{2}}{\rho_{{1}}}^{2}-14\,{d_{{i}
}}^{2}\delta_{{1}}\delta_{{2}}\rho_{{1}}\rho_{{3}}-{d_{{i}}}^{2}{
	\delta_{{2}}}^{2}{\rho_{{1}}}^{2}-12\,{d_{{i}}}^{2}{\delta_{{2}}}^{2}
\rho_{{1}}\rho_{{3}}-24\,{d_{{i}}}^{2}{\delta_{{2}}}^{2}{\rho_{{3}}}^{
	2}+2\,d_{{i}}d_{{r}}{\delta_{{1}}}^{2}\rho_{{1}}\rho_{{3}}+10\,d_{{i}}
d_{{r}}{\delta_{{2}}}^{2}\rho_{{1}}\rho_{{3}}-2\,{d_{{r}}}^{2}{\delta_
	{{1}}}^{2}{\rho_{{1}}}^{2}-2\,{d_{{r}}}^{2}{\delta_{{1}}}^{2}\rho_{{1}
}\rho_{{3}}-2\,{d_{{r}}}^{2}{\delta_{{1}}}^{2}{\rho_{{3}}}^{2}-2\,{d_{
		{r}}}^{2}\delta_{{1}}\delta_{{2}}{\rho_{{1}}}^{2}-6\,{d_{{r}}}^{2}
\delta_{{1}}\delta_{{2}}\rho_{{1}}\rho_{{3}}-{d_{{r}}}^{2}{\delta_{{2}
}}^{2}{\rho_{{1}}}^{2}-2\,{d_{{r}}}^{2}{\delta_{{2}}}^{2}\rho_{{1}}
\rho_{{3}}-10\,{d_{{r}}}^{2}{\delta_{{2}}}^{2}{\rho_{{3}}}^{2}
\Big) {I^{\ast}}^{2}+2\,{d_{{a}}}^{4}{\delta_{{2}}}^{2}+16\,{d_{{a}}}^{3}d
_{{i}}\delta_{{1}}\delta_{{2}}+12\,{d_{{a}}}^{2}{d_{{i}}}^{2}{\delta_{
		{1}}}^{2}+12\,{d_{{a}}}^{2}{d_{{i}}}^{2}{\delta_{{2}}}^{2}+12\,{d_{{a}
}}^{2}{d_{{r}}}^{2}{\delta_{{1}}}^{2}+8\,{d_{{a}}}^{2}{d_{{r}}}^{2}{
	\delta_{{2}}}^{2}+16\,d_{{a}}{d_{{i}}}^{3}\delta_{{1}}\delta_{{2}}+32
\,d_{{a}}d_{{i}}{d_{{r}}}^{2}\delta_{{1}}\delta_{{2}}+2\,{d_{{i}}}^{4}
{\delta_{{1}}}^{2}+8\,{d_{{i}}}^{2}{d_{{r}}}^{2}{\delta_{{1}}}^{2}+12
\,{d_{{i}}}^{2}{d_{{r}}}^{2}{\delta_{{2}}}^{2}+2\,{d_{{r}}}^{4}{\delta
	_{{1}}}^{2}+2\,{d_{{r}}}^{4}{\delta_{{2}}}^{2} \Big) {{T_{reg}^{\ast}}}^{2}+
\Big( 2\,\rho_{{1}}\rho_{{3}} \Big( 2\,d_{{a}}\delta_{{1}}{\rho_{{1
}}}^{2}+2\,d_{{a}}\delta_{{1}}\rho_{{1}}\rho_{{3}}+2\,d_{{a}}\delta_{{
		2}}{\rho_{{1}}}^{2}+2\,d_{{a}}\delta_{{2}}{\rho_{{3}}}^{2}+3\,d_{{i}}
\delta_{{1}}{\rho_{{1}}}^{2}+2\,d_{{i}}\delta_{{1}}\rho_{{1}}\rho_{{3}
}+d_{{i}}\delta_{{1}}{\rho_{{3}}}^{2}+d_{{i}}\delta_{{2}}{\rho_{{1}}}^
{2}+8\,d_{{i}}\delta_{{2}}\rho_{{1}}\rho_{{3}}+9\,d_{{i}}\delta_{{2}}{
	\rho_{{3}}}^{2}+d_{{r}}\delta_{{1}}{\rho_{{1}}}^{2}+d_{{r}}\delta_{{1}
}{\rho_{{3}}}^{2}-2\,d_{{r}}\delta_{{2}}{\rho_{{1}}}^{2}-2\,d_{{r}}
\delta_{{2}}{\rho_{{3}}}^{2} \Big) {I^{\ast}}^{4}+ \Big( -4\,{d_{{a}}}^{3}
\delta_{{1}}{\rho_{{1}}}^{2}-4\,{d_{{a}}}^{3}\delta_{{2}}\rho_{{1}}
\rho_{{3}}-2\,{d_{{a}}}^{2}d_{{i}}\delta_{{1}}{\rho_{{1}}}^{2}-8\,{d_{
		{a}}}^{2}d_{{i}}\delta_{{1}}\rho_{{1}}\rho_{{3}}-4\,{d_{{a}}}^{2}d_{{i
}}\delta_{{2}}{\rho_{{1}}}^{2}-8\,{d_{{a}}}^{2}d_{{i}}\delta_{{2}}\rho
_{{1}}\rho_{{3}}-10\,{d_{{a}}}^{2}d_{{i}}\delta_{{2}}{\rho_{{3}}}^{2}-
6\,{d_{{a}}}^{2}d_{{r}}\delta_{{1}}\rho_{{1}}\rho_{{3}}+4\,{d_{{a}}}^{
	2}d_{{r}}\delta_{{2}}\rho_{{1}}\rho_{{3}}-6\,d_{{a}}{d_{{i}}}^{2}
\delta_{{1}}{\rho_{{1}}}^{2}-12\,d_{{a}}{d_{{i}}}^{2}\delta_{{1}}\rho_
{{1}}\rho_{{3}}-10\,d_{{a}}{d_{{i}}}^{2}\delta_{{1}}{\rho_{{3}}}^{2}-8
\,d_{{a}}{d_{{i}}}^{2}\delta_{{2}}\rho_{{1}}\rho_{{3}}+4\,d_{{a}}d_{{i
}}d_{{r}}\delta_{{1}}\rho_{{1}}\rho_{{3}}-8\,d_{{a}}d_{{i}}d_{{r}}
\delta_{{2}}\rho_{{1}}\rho_{{3}}-4\,d_{{a}}{d_{{r}}}^{2}\delta_{{1}}{
	\rho_{{1}}}^{2}-4\,d_{{a}}{d_{{r}}}^{2}\delta_{{1}}\rho_{{1}}\rho_{{3}
}-4\,d_{{a}}{d_{{r}}}^{2}\delta_{{1}}{\rho_{{3}}}^{2}-4\,d_{{a}}{d_{{r
}}}^{2}\delta_{{2}}\rho_{{1}}\rho_{{3}}-2\,{d_{{i}}}^{3}\delta_{{1}}{
	\rho_{{1}}}^{2}-4\,{d_{{i}}}^{3}\delta_{{1}}\rho_{{1}}\rho_{{3}}-2\,{d
	_{{i}}}^{3}\delta_{{2}}{\rho_{{1}}}^{2}-12\,{d_{{i}}}^{3}\delta_{{2}}
\rho_{{1}}\rho_{{3}}-16\,{d_{{i}}}^{3}\delta_{{2}}{\rho_{{3}}}^{2}-2\,
{d_{{i}}}^{2}d_{{r}}\delta_{{1}}\rho_{{1}}\rho_{{3}}+8\,{d_{{i}}}^{2}d
_{{r}}\delta_{{2}}\rho_{{1}}\rho_{{3}}-2\,d_{{i}}{d_{{r}}}^{2}\delta_{
	{1}}{\rho_{{1}}}^{2}-4\,d_{{i}}{d_{{r}}}^{2}\delta_{{1}}\rho_{{1}}\rho
_{{3}}-2\,d_{{i}}{d_{{r}}}^{2}\delta_{{2}}{\rho_{{1}}}^{2}-8\,d_{{i}}{
	d_{{r}}}^{2}\delta_{{2}}\rho_{{1}}\rho_{{3}}-20\,d_{{i}}{d_{{r}}}^{2}
\delta_{{2}}{\rho_{{3}}}^{2}-2\,{d_{{r}}}^{3}\delta_{{1}}\rho_{{1}}
\rho_{{3}}+4\,{d_{{r}}}^{3}\delta_{{2}}\rho_{{1}}\rho_{{3}} \Big) {Y
}^{2}+4\,{d_{{a}}}^{4}d_{{i}}\delta_{{2}}+8\,{d_{{a}}}^{3}{d_{{i}}}^{2
}\delta_{{1}}+8\,{d_{{a}}}^{3}{d_{{r}}}^{2}\delta_{{1}}+8\,{d_{{a}}}^{
	2}{d_{{i}}}^{3}\delta_{{2}}+16\,{d_{{a}}}^{2}d_{{i}}{d_{{r}}}^{2}
\delta_{{2}}+4\,d_{{a}}{d_{{i}}}^{4}\delta_{{1}}+16\,d_{{a}}{d_{{i}}}^
{2}{d_{{r}}}^{2}\delta_{{1}}+4\,d_{{a}}{d_{{r}}}^{4}\delta_{{1}}+8\,{d
	_{{i}}}^{3}{d_{{r}}}^{2}\delta_{{2}}+4\,d_{{i}}{d_{{r}}}^{4}\delta_{{2
}} \Big) {T_{reg}^{\ast}}-{\rho_{{1}}}^{2}{\rho_{{3}}}^{2} \Big( \rho_{{1}}
+\rho_{{3}} \Big) ^{2}{I^{\ast}}^{6}+2\,\rho_{{1}}\rho_{{3}} \Big( {d_{{a}
}}^{2}{\rho_{{1}}}^{2}+{d_{{a}}}^{2}\rho_{{1}}\rho_{{3}}+d_{{a}}d_{{i}
}{\rho_{{1}}}^{2}+d_{{a}}d_{{i}}{\rho_{{3}}}^{2}+d_{{a}}d_{{r}}{\rho_{
		{1}}}^{2}+d_{{a}}d_{{r}}{\rho_{{3}}}^{2}+3\,{d_{{i}}}^{2}{\rho_{{1}}}^
{2}+5\,{d_{{i}}}^{2}\rho_{{1}}\rho_{{3}}+6\,{d_{{i}}}^{2}{\rho_{{3}}}^
{2}-d_{{i}}d_{{r}}{\rho_{{1}}}^{2}-d_{{i}}d_{{r}}{\rho_{{3}}}^{2}+{d_{
		{r}}}^{2}\rho_{{1}}\rho_{{3}}+{d_{{r}}}^{2}{\rho_{{3}}}^{2} \Big) {Y
}^{4}+ \Big( -{d_{{a}}}^{4}{\rho_{{1}}}^{2}-2\,{d_{{a}}}^{3}d_{{i}}
\rho_{{1}}\rho_{{3}}-2\,{d_{{a}}}^{3}d_{{r}}\rho_{{1}}\rho_{{3}}-4\,{d
	_{{a}}}^{2}{d_{{i}}}^{2}{\rho_{{1}}}^{2}-6\,{d_{{a}}}^{2}{d_{{i}}}^{2}
\rho_{{1}}\rho_{{3}}-5\,{d_{{a}}}^{2}{d_{{i}}}^{2}{\rho_{{3}}}^{2}+2\,
{d_{{a}}}^{2}d_{{i}}d_{{r}}\rho_{{1}}\rho_{{3}}-2\,{d_{{a}}}^{2}{d_{{r
}}}^{2}{\rho_{{1}}}^{2}-2\,{d_{{a}}}^{2}{d_{{r}}}^{2}\rho_{{1}}\rho_{{
		3}}-2\,{d_{{a}}}^{2}{d_{{r}}}^{2}{\rho_{{3}}}^{2}-2\,d_{{a}}{d_{{i}}}^
{3}\rho_{{1}}\rho_{{3}}-6\,d_{{a}}{d_{{i}}}^{2}d_{{r}}\rho_{{1}}\rho_{
	{3}}-2\,d_{{a}}d_{{i}}{d_{{r}}}^{2}\rho_{{1}}\rho_{{3}}-2\,d_{{a}}{d_{
		{r}}}^{3}\rho_{{1}}\rho_{{3}}-{d_{{i}}}^{4}{\rho_{{1}}}^{2}-4\,{d_{{i}
}}^{4}\rho_{{1}}\rho_{{3}}-4\,{d_{{i}}}^{4}{\rho_{{3}}}^{2}+2\,{d_{{i}
}}^{3}d_{{r}}\rho_{{1}}\rho_{{3}}-2\,{d_{{i}}}^{2}{d_{{r}}}^{2}{\rho_{
		{1}}}^{2}-6\,{d_{{i}}}^{2}{d_{{r}}}^{2}\rho_{{1}}\rho_{{3}}-10\,{d_{{i
}}}^{2}{d_{{r}}}^{2}{\rho_{{3}}}^{2}+2\,d_{{i}}{d_{{r}}}^{3}\rho_{{1}}
\rho_{{3}}-{d_{{r}}}^{4}{\rho_{{3}}}^{2} \Big) {I^{\ast}}^{2}+2\,{d_{{a}}}^
{4}{d_{{i}}}^{2}+2\,{d_{{a}}}^{4}{d_{{r}}}^{2}+2\,{d_{{a}}}^{2}{d_{{i}
}}^{4}+8\,{d_{{a}}}^{2}{d_{{i}}}^{2}{d_{{r}}}^{2}+2\,{d_{{a}}}^{2}{d_{
		{r}}}^{4}+2\,{d_{{i}}}^{4}{d_{{r}}}^{2}+2\,{d_{{i}}}^{2}{d_{{r}}}^{4}$.\vspace{6pt}
	
$b_8=\Big( {\delta_{{1}}}^{4}+4\,{\delta_{{1}}}^{2}{\delta_{{2}}}^{2}+{
	\delta_{{2}}}^{4} \Big) {{T_{reg}^{\ast}}}^{4}+ \Big( 4\,d_{{a}}{\delta_{{1
}}}^{3}+8\,d_{{a}}\delta_{{1}}{\delta_{{2}}}^{2}+8\,d_{{i}}{\delta_{{1
}}}^{2}\delta_{{2}}+4\,d_{{i}}{\delta_{{2}}}^{3} \Big) {{T_{reg}^{\ast}}}^{3
}+ \Big(  \Big( -2\,{\delta_{{1}}}^{2}{\rho_{{1}}}^{2}-2\,{\delta_{{
			1}}}^{2}\rho_{{1}}\rho_{{3}}-{\delta_{{1}}}^{2}{\rho_{{3}}}^{2}-2\,
\delta_{{1}}\delta_{{2}}{\rho_{{1}}}^{2}-6\,\delta_{{1}}\delta_{{2}}
\rho_{{1}}\rho_{{3}}-{\delta_{{2}}}^{2}{\rho_{{1}}}^{2}-2\,{\delta_{{2
}}}^{2}\rho_{{1}}\rho_{{3}}-5\,{\delta_{{2}}}^{2}{\rho_{{3}}}^{2}
\Big) {I^{\ast}}^{2}+6\,{d_{{a}}}^{2}{\delta_{{1}}}^{2}+4\,{d_{{a}}}^{2}{
	\delta_{{2}}}^{2}+16\,d_{{a}}d_{{i}}\delta_{{1}}\delta_{{2}}+4\,{d_{{i
}}}^{2}{\delta_{{1}}}^{2}+6\,{d_{{i}}}^{2}{\delta_{{2}}}^{2}+4\,{d_{{r
}}}^{2}{\delta_{{1}}}^{2}+4\,{d_{{r}}}^{2}{\delta_{{2}}}^{2} \Big) {
	{T_{reg}^{\ast}}}^{2}+ \Big(  \Big( -4\,d_{{a}}\delta_{{1}}{\rho_{{1}}}^{2}-
4\,d_{{a}}\delta_{{1}}\rho_{{1}}\rho_{{3}}-2\,d_{{a}}\delta_{{1}}{\rho
	_{{3}}}^{2}-4\,d_{{a}}\delta_{{2}}\rho_{{1}}\rho_{{3}}-2\,d_{{i}}
\delta_{{1}}{\rho_{{1}}}^{2}-4\,d_{{i}}\delta_{{1}}\rho_{{1}}\rho_{{3}
}-2\,d_{{i}}\delta_{{2}}{\rho_{{1}}}^{2}-8\,d_{{i}}\delta_{{2}}\rho_{{
		1}}\rho_{{3}}-10\,d_{{i}}\delta_{{2}}{\rho_{{3}}}^{2}-2\,d_{{r}}\delta
_{{1}}\rho_{{1}}\rho_{{3}}+4\,d_{{r}}\delta_{{2}}\rho_{{1}}\rho_{{3}}
\Big) {I^{\ast}}^{2}+4\,{d_{{a}}}^{3}\delta_{{1}}+8\,{d_{{a}}}^{2}d_{{i}}
\delta_{{2}}+8\,d_{{a}}{d_{{i}}}^{2}\delta_{{1}}+8\,d_{{a}}{d_{{r}}}^{
	2}\delta_{{1}}+4\,{d_{{i}}}^{3}\delta_{{2}}+8\,d_{{i}}{d_{{r}}}^{2}
\delta_{{2}} \Big) {T_{reg}^{\ast}}+2\,\rho_{{1}}\rho_{{3}} \Big( {\rho_{{1
}}}^{2}+\rho_{{3}}\rho_{{1}}+{\rho_{{3}}}^{2} \Big) {I^{\ast}}^{4}+ \Big( 
-2\,{d_{{a}}}^{2}{\rho_{{1}}}^{2}-2\,{d_{{a}}}^{2}\rho_{{1}}\rho_{{3}}
-{d_{{a}}}^{2}{\rho_{{3}}}^{2}-2\,d_{{a}}d_{{i}}\rho_{{1}}\rho_{{3}}-2
\,d_{{a}}d_{{r}}\rho_{{1}}\rho_{{3}}-2\,{d_{{i}}}^{2}{\rho_{{1}}}^{2}-
6\,{d_{{i}}}^{2}\rho_{{1}}\rho_{{3}}-5\,{d_{{i}}}^{2}{\rho_{{3}}}^{2}+
2\,d_{{i}}d_{{r}}\rho_{{1}}\rho_{{3}}-{d_{{r}}}^{2}{\rho_{{1}}}^{2}-2
\,{d_{{r}}}^{2}\rho_{{1}}\rho_{{3}}-2\,{d_{{r}}}^{2}{\rho_{{3}}}^{2}
\Big) {I^{\ast}}^{2}+{d_{{a}}}^{4}+4\,{d_{{a}}}^{2}{d_{{i}}}^{2}+4\,{d_{{a
}}}^{2}{d_{{r}}}^{2}+{d_{{i}}}^{4}+4\,{d_{{i}}}^{2}{d_{{r}}}^{2}+{d_{{
			r}}}^{4}$.\vspace{6pt}
		
$b_{10}=\Big( 2\,{\delta_{{1}}}^{2}+2\,{\delta_{{2}}}^{2} \Big) {{T_{reg}^{\ast}}}
^{2}+ \Big( 4\,d_{{a}}\delta_{{1}}+4\,d_{{i}}\delta_{{2}} \Big) {T{reg}^{\ast}}- \Big( \rho_{{1}}+\rho_{{3}} \Big) ^{2}{I^{\ast}}^{2}+2\,{d_{{a}}}
^{2}+2\,{d_{{i}}}^{2}+2\,{d_{{r}}}^{2}$.}

\bibliographystyle{ieeetr}

\end{document}